\documentclass[a4paper,11pt]{article}
\usepackage{jheppub} 
\makeatletter
\gdef\@fpheader{}
\makeatother
\usepackage{amsmath}
\usepackage{amssymb}
\usepackage{amsthm}
\usepackage{empheq}
\usepackage{bbm}
\usepackage{amsfonts}
\usepackage{calligra}
\usepackage{amssymb}
\usepackage{braket}
\usepackage{enumitem}
\usepackage[english]{babel}
\usepackage{xcolor}
\usepackage{amssymb}
\usepackage{stmaryrd}
\usepackage{tikz-cd} 
\usepackage{tikz} 
\usepackage{todonotes} 
\usepackage{cases} 
\usepackage[utf8]{inputenc} 
\usepackage{hyperref} 
\usepackage{cleveref} 
\usepackage{mathrsfs} 
\usepackage{cancel} 
\usepackage[percent]{overpic} 
\usepackage{longtable} 
\usepackage{array} 
\usepackage{booktabs} 
\usepackage{tabularx} 
\usepackage{colortbl} 
\usepackage{xcolor} 
\usepackage{hhline} 
\usepackage{float} 
\usepackage{tabularx} 
\usepackage{enumitem} 

\newtheorem{theorem}{Theorem}[section] 
\newtheorem{lemma}[theorem]{Lemma} 
\newtheorem{result}[theorem]{Result} 
\newtheorem{remark}[theorem]{Remark} 
\crefname{subsubsection}{section}{sections} 
\crefname{lemma}{lemma}{lemmas} 

\def\del{{\partial}}
\def\cA{{\mathcal{A}}}

\def\I{{\mathcal{I}}}
\def\zb{\bar{z}}

\newcommand{\black}{\color{black}}
\newcommand{\red}{\color{red}}
\newcommand{\white}{\color{white}}
\newcommand{\be}{\begin{equation}}
\newcommand{\ee}{\end{equation}}
\setuptodonotes{linecolor=black!50,backgroundcolor=yellow!80,size=footnotesize}
\allowdisplaybreaks

\title{\boldmath Homotopy kinematic algebras at null infinity}

\author[a]{Felipe D\'iaz-Jaramillo,}
\author[b]{Silvia Nagy,}
\author[b]{Giorgio Pizzolo}

\affiliation[a]{Institute for Physics, Humboldt University Berlin,
Zum Großen Windkanal 6, D-12489 Berlin, Germany}
\affiliation[b]{Department of Mathematical Sciences, Durham University, Durham, DH1 3LE, UK}

\emailAdd{felipe.diaz-jaramillo@hu-berlin.de}
\emailAdd{silvia.nagy@durham.ac.uk}
\emailAdd{giorgio.pizzolo@durham.ac.uk}

\abstract{ 
We present the first formulation of a homotopy algebra adapted to a $1/r$ expansion near future null infinity ($\mathcal{I^+}$). Focusing on self-dual Yang-Mills theory in Bondi coordinates, we demonstrate that imposing the homotopy algebra relations naturally yields the physically consistent fall-off behaviour of the fields near $\mathcal{I^+}$. Furthermore, we employ this framework to systematically construct kinematic algebras, uncovering novel infinite families of such algebras that satisfy the Jacobi identity on slices near $\mathcal{I^+}$.
\black
}

\begin{document}
\maketitle
\flushbottom

\section{Introduction}

Holography, initially in the context of the AdS/CFT correspondence\cite{Maldacena:1997re}, has garnered significant attention for its profound insights on both sides of the duality. More recently, celestial holography, which represents its flat-space counterpart, has been gaining a lot of traction\footnote{See reviews \cite{Pasterski:2021raf,Raclariu:2021zjz} and references within.}. In this framework, an important step is understanding the behaviour of theories close to the relevant boundary of flat space-time, which is null infinity. This in turn controls the so-called asymptotic symmetries, which govern the behaviours or scattering amplitudes, and have natural counterparts in the study of CFT correlators.

Yet another connection between gravity and quantum field theory emerges through the double copy program, which posits that gravity can be obtained by ``squaring" Yang-Mills theory \cite{Bern:2008qj,Bern:2010ue}. In the context of scattering amplitudes, this process relies on the so-called colour-kinematics duality which requires that the colour and kinematic information of Yang-Mills\footnote{This also holds for other theories with colour degrees of freedom.} amplitudes obey the same algebraic relations. The colour information is encoded in the structure constants (or traces of generators) of a Lie algebra, while the kinematic information is contained in the polarization vectors and momenta of the external particles participating in the scattering event. Thus, colour-kinematics duality suggests the existence of an algebraic structure describing the kinematics of Yang-Mills theory dubbed the \textit{kinematic algebra}. Due to colour being encoded in a Lie algebra, the search for the kinematic algebra was initially also focused on Lie algebras.

Early success was found in self-dual Yang-Mills theory (SDYM)\cite{Monteiro:2011pc} where a convenient scalar description of the theory was available \cite{Plebanski:1975wn,Bardeen:1995gk,Chalmers:1996rq,Prasad:1979zc,Dolan:1983bp,Parkes:1992rz,Cangemi:1996pf,Popov:1996uu,Popov:1998pc}. This has led to an elegant double copy to self-dual gravity and numerous extensions \cite{Monteiro:2011pc,Chen:2019ywi,Chen:2021chy,Monteiro:2022nqt,Monteiro:2013rya,Elor:2020nqe,Armstrong-Williams:2022apo,Farnsworth:2021wvs,Skvortsov:2022unu,Campiglia:2021srh,Brandhuber:2022enp,Ben-Shahar:2022ixa,CarrilloGonzalez:2024sto,Berman:2018hwd,Nagy:2022xxs,Monteiro:2022lwm,Krasnov:2021cva,Monteiro:2022xwq,Lipstein:2023pih,Borsten:2023paw,Brown:2023zxm,Doran:2023cmj,Chowdhury:2024dcy,LopesCardoso:2024ttc,Correa:2024mub}. Kinematic algebras have also been revealed in Chern-Simons theory \cite{Ben-Shahar:2021zww, Borsten:2022vtg, Bonezzi:2024dlv}, the nonlinear sigma model \cite{Cheung:2016prv}, a certain formulation of YM following from Heavy-Mass Effective Field Theory \cite{Chen:2022nei,Brandhuber:2022enp,Brandhuber:2021bsf}, via pure spinor actions \cite{Ben-Shahar:2021doh,Borsten:2023reb}, and even in abelian theories \cite{Armstrong-Williams:2024icu}. 

In recent years, the algebraic structures underlying perturbative field theories, called homotopy algebras, have been used to formalise the construction of kinematic algebras and systematize the implementation of the double copy \cite{Reiterer:2019dys,Bonezzi:2022bse,Bonezzi:2022yuh,Borsten:2021hua,Borsten:2021gyl,Borsten:2020zgj,Bonezzi:2023ced,Borsten:2022vtg,Escudero:2022zdz}. Homotopy algebras, as we shall explain in \autoref{sec:Homotopy algebras and homotopy transfer}, are generalizations of more commonly known algebras in the physics literature, such as \textit{strict} Lie algebras\footnote{These are generally just referred to as Lie algebras in the physics literature, however, we use the term strict to differentiate from homotopy algebras which will feature in the article.}, that encode the symmetries and dynamics of perturbative field theories. Formulating physical theories with colour degrees of freedom in this framework allows one to algebraically factorize the theories in a well-defined manner into colour and kinematics, construct a kinematic algebra, and finally take the tensor product of two copies of the kinematic algebra which yields a new homotopy algebra that encodes a gravitational theory. 

Notably, in general, the kinematic algebra that arises from the homotopy-algebraic framework is not a strict Lie algebra. Rather, the general structure obtained by these means is a so-called BV$_{\infty}^{\Box}$ algebra, first introduced by Reiterer in \cite{Reiterer:2019dys} to understand colour-kinematics duality. To this end, however, one can work with simpler, strict algebras at the cost of locality \cite{Borsten:2023ned,Borsten:2020zgj,Borsten:2021hua} or by making particular gauge choices. Moreover, in the context of holography, homotopy algebras have been used to transfer information from the bulk to the boundary \cite{Chiaffrino:2023wxk,Alfonsi:2024utl}.

Thus homotopy algebras hold great promise in both the study of holography, and that of the double copy. In this paper, we make this idea more tangible by focusing on self-dual Yang-Mills theory at null infinity. We demonstrate that the physical fall-off behaviour of fields, gauge parameters, and equations of motion can be derived from a homotopy algebraic perspective. Furthermore, we identify a strict kinematic Lie algebra that arises in carefully selected slices of the radial expansion near the boundary, providing a novel approach to its formulation.

Our strategy is as follows: we begin by recasting the homotopy algebra formulation of SDYM in flat Bondi coordinates and then perform an expansion in powers of 1/r. Without making any physical assumptions, we consider slices corresponding to approximations of varying precision near null infinity. By requiring that these slices preserve all the defining conditions of the homotopy algebras, we derive constraints on the fall-off behaviour of the fields that precisely match the physical fall-off conditions described in \cite{Strominger:2013lka,Strominger:2017zoo} and reviewed in \autoref{Review of physical fall-off at null infinity}:
\begin{align}\label{A_falls_intro}
        A_r=\frac{A_r^{(-2)}}{r^2}+..., \qquad A_u=A^{(0)}_u+\frac{A_u^{(-1)}}{r}+..., \qquad A_z=A^{(0)}_z+..., \qquad A_{\bar{z}}=A^{(0)}_{\Bar{z}}+...
    \end{align}
as detailed in \autoref{sec:SDYM at null infinity via slice truncation}.

Armed with a well-defined homotopy algebra near null infinity, we then construct the kinematic algebra, and specifically look for strict Lie algebra versions of these. We manage to find a number of novel ways to partially gauge fix and truncate the fields in order to reveal new strict kinematic algebras at null infinity. The solutions are of the form
\be \label{intro_slice}
A_\mu=\sum_{n=t_\mu^k}^{f_{i,\mu}^k}r^{n}A_\mu^{(n)}
\ee 
where $t_\mu^k$ and $f_{i,\mu}^k$, with $t_\mu^k\leq f_{i,\mu}^k\leq0$, define ``slices" in the expansion in $1/r$. We refer to $t_\mu^k$ as the truncation, as it encodes the precision to which we work in $1/r$, whereas $f_{i,\mu}^k$ is determined by the fall-off \eqref{A_falls_intro}, together with possible additional constraints. We find infinite families of solutions labelled by $k$, where $|k|$ encodes the ``thickness" of the slice, summarised in the table below:

\begin{table}[H]
    \centering
    \arrayrulecolor{black}
    \renewcommand{\arraystretch}{1.5} 
    \begin{tabular}{|c|c|c|c|c!{\vrule width 0.4mm}c|c|c|c|c|c|c|}
        \cline{2-12}
        \multicolumn{1}{c|}{} & \multicolumn{4}{c!{\vrule width 0.4mm}}{$k=-1$} & \multicolumn{7}{c|}{$k\leq-2$} \\
        \cline{2-12}
        \hhline{~~~~~~~~----}
        \multicolumn{1}{c|}{} & $t_\mu^{-1}$ & $f_{1,\mu}^{-1}$ & $f_{2,\mu}^{-1}$ & $f_{3,\mu}^{-1}$ & $t_\mu^k$ & $f_{4,\mu}^k$ & $f_{5,\mu}^k$ & \cellcolor[gray]{0.9}$f_{6,\mu}^k$ & \cellcolor[gray]{0.9}$f_{7,\mu}^k$ & \cellcolor[gray]{0.9}$f_{8,\mu}^k$ & \cellcolor[gray]{0.9}$f_{9,\mu}^k$ \\
        \hline
        $r$ & $-3$ & $-2$ & $-2$ & $-3$ & $k-2$ & $-2$ & $k-2$ & \cellcolor[gray]{0.9}$k-1$ & \cellcolor[gray]{0.9}$-2$  & \cellcolor[gray]{0.9}$-2$ & \cellcolor[gray]{0.9}$k-1$  \\
        $u$ & $-2$ & $-1$ & $0$ & $0$ & $k-1$ & $k+1$ & $0$ & \cellcolor[gray]{0.9}$0$ & \cellcolor[gray]{0.9}$k$ & \cellcolor[gray]{0.9}$0$ & \cellcolor[gray]{0.9}$k$   \\
        $z$ & $-1$ & $0$ & $0$ & $-1$ & $k$ & $0$ & $k$ & \cellcolor[gray]{0.9}$0$ & \cellcolor[gray]{0.9}$k+1$  & \cellcolor[gray]{0.9}$k+1$ & \cellcolor[gray]{0.9}$0$  \\
        $\bar z$ & $-2$ & $0$ & $-1$ & $0$ & $k-1$ & $k+1$ & 0 & \cellcolor[gray]{0.9}$k$ & \cellcolor[gray]{0.9}$0$  & \cellcolor[gray]{0.9}$k$ & \cellcolor[gray]{0.9}$0$   \\
        \hline
    \end{tabular}
    \caption{Slices manifesting strict kinematic algebras. The slices are defined in Eq. \eqref{intro_slice}. Note that while solutions 1-3 can be obtained via gauge choices, the families of solutions 4-7 require further elimination of unphysical d.o.f. and/or a restriction to a subsector of the theory. In the limit $k\to -\infty$ and upon imposing the Lorenz gauge condition, slices 6 and 7 reproduce two versions of the area-preserving diffeomorphisms in \cite{Monteiro:2011pc}. The families of solutions 6-9 in the shaded boxes only become strict upon imposing additional constrains coming from certain components of the e.o.m.. See \autoref{sec:Strict kinematic algebras} for more details.}
    \label{table: introduction}
\end{table}

The article is organised as follows: in \autoref{Physical motivation}, we start with a brief overview of the two main topics to be studied. The first is the physical fall-off of fields, equations of motion and gauge parameters in YM in a radial expansion near null infinity, and the second is the kinematic algebra arising as a consequence of color-kinematics duality in the double copy programme. In \autoref{sec:Homotopy algebras and homotopy transfer} we review our main mathematical tool, i.e. the homotopy algebras. In the brief \autoref{Homotopy algebra formulation of SDYM in Bondi coordinates} we present the homotopy algebra formulation of SDYM in flat Bondi coordinates, which are best adapted for studying asymptotic phenomena in this theory. In \autoref{sec:SDYM at null infinity via slice truncation}, we formulate and solve the constraints arising from the requirement that the projection to slices near null infinity preserves the defining conditions of homotopy algebras, giving an alternative derivation of the physical fall-off \eqref{A_falls_intro}. In \autoref{sec:Strict kinematic algebras}, we employ the homotopy algebra machinery to find new infinite families of kinematic algebras, as summarised in \autoref{table: introduction}.  Then in \autoref{sec:algebra_via_homotopy_transfer} we introduce some further homotopy algebra machinery, namely co-chain maps and quasi-isomorphisms, and use these to give an alternative proof for the fall-off. We conclude in \autoref{sec: Conclusions}. \autoref{sec:appendix_proofs} contains the analytic proof of the solution presented in \autoref{sec:SDYM at null infinity via slice truncation}, while \autoref{sec:appendix other proofs} details the construction of the strict kinematic algebras. Finally, \autoref {app:Homotopy transfer} gives details of the homotopy transfer procedure. 

\vspace{15 pt}


\section{Physical motivation}
\label{Physical motivation}
In this section we review the two areas where we will be applying the homotopy algebra techniques. In \autoref{ckduality}, we give a very brief overview of colo-kinematics duality, and the emergence of the kinematic algebra, focusing on SDYM. Then, in \autoref{Review of physical fall-off at null infinity}, we give a quick description of gauge fields near null infinity. 

\subsection{Color-kinematics duality and the kinematic algebra}\label{ckduality}

Color-kinematics duality is a revolutionary observation about scattering amplitudes for non-abelian gauge theory \cite{Bern:2008qj,Bern:2010ue}, which has led to immense progress in relating gravity to gauge theory via the double copy. To state it, we first formulate n-point amplitudes as sums over cubic diagrams
\be 
\cA_n=g^{n-2}\sum_i\frac{c_i n_i}{D_i}\ ,
\ee 
where $c_i$ are the color factors (consisting of products of the structure constants of the gauge group), $n_i$ are kinematic factors (depending on momenta and polarisation vectors), and $D_i$ are products of Feynman propagators. The observation is that there exists a gauge freedom at the level of the expression above, such that we can choose kinematic factors which exactly mimic the behaviour of the color factors. The two relevant ingredients are antisymmetry under exchange of particle labels
\be 
c_i\ \to\ -c_i \quad\Rightarrow\quad n_i\to -n_i  
\ee 
and the Jacobi identity
\be \label{Jacobi_amplitudes}
c_i+c_j+c_k=0 \quad\Rightarrow\quad n_i+n_j+n_k=0 \ .
\ee 
The LHS of the above is of course just a consequence of the fact that the gauge algebra is a \emph{strict} Lie algebra\footnote{Strict Lie algebras are usually just referred to as Lie algebras in the physics literature, but we need to be more precise in order to differentiate from the structures introduced later in the paper.}. This suggests the existence of a strict Lie algebra governing the kinematics of the theory, which has been dubbed the kinematic algebra. 

A very clean manifestation of this algebra arises in the self-dual sector of Yang-Mills theory, given by the following constraint in Lorentzian signature:
\be \label{SD_eqn_original}
F_{\mu\nu}=\frac{i}{2}\epsilon_{\mu\nu\rho\sigma} F^{\rho\sigma} \ .
\ee 
The above encapsulates both the e.o.m. and Bianchi identities. Working in light-cone coordinates
\be\label{LCcoord} 
x^-=\tfrac{1}{\sqrt{2}}(t-z),\quad x^+=\tfrac{1}{\sqrt{2}}(t+z),\quad w=\tfrac{1}{\sqrt{2}}(x+iy),\quad \bar{w}=\tfrac{1}{\sqrt{2}}(x-iy) \ ,
\ee 
with line element
\be 
ds^2=-2dx^-dx^+\,+\, 2dwd\bar{w} \ .
\ee 
In order to reveal the kinematic algebra, we first impose the light-cone gauge 
\be 
A_+=0 \ .
\ee 
Then additionally solving the constraint equations coming from \eqref{SD_eqn_original}, we find 
\be 
A_w=0\ . 
\ee 
Working with the remaining non-vanishing components $A^m=(A^+,A^w)=(A_-,A_{\bar{w}})$, we find a kinematic bracket in the form of a Schouten-Nijenhuis bracket  (\cite{Bonezzi:2023pox})
\be \label{SNbracket}
\big[\cA_1,\cA_2\big]^m_{\rm SN}=\cA_1^n\del_n\cA_2^m-\cA_2^n\del_n\cA_1^m \ .
\ee
In the above, $\cA^m$ is the so-called color-stripped version on $A^m$ (see \cref{sec:Homotopy algebras and homotopy transfer} for details), i.e. we are treating it as a standard 2 dimensional vector. The bracket \eqref{SNbracket} corresponds to a diffeomorphism algebra, which is additionally area-preserving due to the constraint            $\partial_m\cA^m=0$, arising from \eqref{SD_eqn_original}. This, of course, satisfies Jacobi, and is the relevant bracket underpinning the kinematic factors. Upon further expressing the gauge fields in terms of scalars, we can recover the famed original expression in \cite{Monteiro:2011pc} (see \cite{Bonezzi:2023pox} for details).




Let us pause to take stock of how the kinematic algebra arose.
Note that we had to
\vspace{-8 pt}
\begin{itemize}[noitemsep] 
\item find a convenient gauge choice, 
\item eliminate a non-propagating degree of freedom,
\end{itemize}
\vspace{-8 pt}
in order to reveal the presence of a strict kinematic algebra. 

Can we systematize the search for set-ups manifesting strict kinematic algebras? A promising tool are the so-called homotopy algebras, reviewed in detail in \cref{sec:Homotopy algebras and homotopy transfer}. Within this context, there exists an algorithmic way of constructing the kinematic algebra, which eliminates the guesswork \cite{Reiterer:2019dys}. The trade-off is that the outcome is not usually a strict algebra, i.e. the Jacobi Identity \eqref{Jacobi_amplitudes} is relaxed. Fortunately, the failure to satisfy \eqref{Jacobi_amplitudes} is precisely encoded in an object called the Jacobiator, also constructed in an algorithmic way from the data in the theory. This greatly simplifies the task of determining under which circumstances the kinematic algebra will be a strict one, by looking at how to set the Jacobiator to 0. 

Even though the double copy procedure can be performed even in the absence of a strict kinematic algebra \cite{Bonezzi:2022bse,Bonezzi:2022yuh,Bonezzi:2023ced,Bonezzi:2024dlv,Bonezzi:2023xhn,Bonezzi:2023lkx,Bonezzi:2023ciu}, its presence provides dramatic simplifications of the Yang-Mills theory, and consequently of the map to the resulting gravity theory. This has exploited for example in pushing the remit of the double copy to curved backgrounds \cite{Lipstein:2023pih,CarrilloGonzalez:2024sto}.

Later in the article we will show that there is a novel way of revealing strict kinematic algebras, which will be related to taking appropriate slices in the fall-of of fields close to null infinity.

\subsection{Review of physical fall-off at null infinity}
\label{Review of physical fall-off at null infinity}
We will work in the so-called flat Bondi coordinates $x^\mu=(r,u,z,\zb)$ 
in which the leading order part of the metric is:
\be\label{fl_B_metric}
ds^2=-2drdu+2r^2dzd\zb  \  .
\ee
In these coordinates, the celestial sphere becomes a complex plane spanned by $z$ and $\zb$. These are particularly suited for working in the self-dual sector, and they have a particularly simple map to the light-cone coordinates \eqref{LCcoord} 
\be\label{LCtoflB} 
x^-=rz\zb+u,\quad x^+=r,\quad w=rz,\quad \bar{w}=r\zb \ .
\ee 
The fall-off will be completely analogous to the more standard Bondi coordinates.

Here we will give a brief review of the physical motivation for the fall-offs that are generally assumed for the gauge potential at $\I$  \cite{Strominger:2013lka,Strominger:2017zoo}:
 \begin{align}\label{A_falls_no_u_gauge}
        A_r=\frac{A_r^{(-2)}}{r^2}+..., \qquad A_u=A^{(0)}_u+\frac{A_u^{(-1)}}{r}+..., \qquad A_z=A^{(0)}_z+..., \qquad A_{\bar{z}}=A^{(0)}_{\Bar{z}}+...
    \end{align}

In the above, we assume that the coefficients in the $r$-expansion are functions of $y^\alpha=(u,z,\zb)$. The fall-off above and the motivation below apply to both YM and Maxwell theory. By extension, it will also apply to the self-dual sector of these theories\footnote{Recall that \eqref{SD_eqn_original} encodes the general YM equations; this can be seen by contracting with the gauge-covariant derivative.}. We are seeking the most general fall-off that yields physically reasonable solutions. To start off, we recall that the electric field should fall off as 
\be 
E_r=\frac{E_r^{(-2)}}{r^2}+...
\ee 
The corresponding field strength component in Bondi coordinates is
\be\label{Fur_fall_off} 
F_{ur}=\frac{F_{ur}^{(-2)}}{r^2}+...
\ee 
This already tells us that $A_r$ starts at $\mathcal{O}(\tfrac{1}{r^2})$, so we can further conclude that
\be\label{Frz_fall_off} 
F_{rz}=\frac{F_{rz}^{(-2)}}{r^2}+...
\ee 
We then consider the energy flux, encoded in the $T_{uu}$ component of the energy momentum tensor
\be 
T_{\mu\nu}=\text{Tr}\left(F_{\mu\rho}F^{\rho}_{\ \nu}-\tfrac{1}{4}g_{\mu\nu}F_{\rho\sigma}F^{\rho\sigma} \right)
\ .
\ee 
Requiring that $T_{uu}$ falls of like 
\be 
T_{uu}=\frac{T_{uu}^{(-2)}}{r^2}+...
\ee 
which is needed in order to get a finite integrated flux, and using \eqref{Fur_fall_off} and \eqref{Frz_fall_off} we get
\be 
T_{uu}\approx\frac{F_{uz}F_{u\zb}}{r^2}+...
\ee
where the $...$ refer to terms with a sharper fall-off. We deduce that
\be \label{Fuz_fall_off}
F_{uz}=F_{uz}^{(0)}+...
\ee 
Finally, it is straightforward to see that the fall-offs in \eqref{A_falls_no_u_gauge} follow from \eqref{Fur_fall_off}, \eqref{Frz_fall_off} and \eqref{Fuz_fall_off} above. We remark that it is very common to set
\be \label{Au0_gauge_fix}
A^{(0)}_u=0 \ ,
\ee 
however this does not strictly follow from the physical considerations above. Generally, \eqref{Au0_gauge_fix} can be thought of as a partial gauge fixing; indeed it is often imposed in conjunction with Lorenz gauge $\nabla^\mu A_\mu=0$, or radial gauge $A_r=0$. In both these cases, the residual symmetry in these gauges is sufficient to impose \eqref{Au0_gauge_fix}. A useful consequence of \eqref{Au0_gauge_fix} is that the soft gluon mode $N_z$ can be written straightforwardly as a difference between the free data $A_z^{(0)}$ at future timelike infinity and spatial infinity. Specifically, we have
\be 
N_z=\lim_{\omega\to 0}\int_{-\infty}^{\infty}du\ F_{uz}^{(0)}\ e^{i\omega u}=
\int_{-\infty}^{\infty}du\ F_{uz}^{(0)} = A_z^{(0)}\Big|_{\I_+^+} - 
A_z^{(0)}\Big|_{\I_-^+}
\ee
where to get to the last line we used \eqref{Au0_gauge_fix}.

We pause here to comment on an implicit assumption in the fall-offs above: we have assumed the absence of log terms. More generally, the gauge field could be expanded as a polyhomogeneous function 
\be \label{A_st_exp}
A_\mu=\sum_{n,k\geq0}A_\mu^{(-n;k)}(y^\alpha)\frac{\text{log}^k r}{r^n} \ ,
\ee
The assumption that the log terms are absent (i.e. $A_\mu^{(-n;k)}(y^\alpha)=0$ for $k\neq 0$) is a well motivated one at tree-level, and is supported by the polynomial tree level fall-offs in the $u$ variable which relate to sub-leading effects \footnote{See e.g. \cite{Nagy:2024jua} for more details.}. In certain specific gauges, such as harmonic gauge  \cite{Campiglia:2021oqz}, or when including loop effects, it may be that one needs to relax the fall-off assumption to something of the form in \eqref{A_st_exp}. We leave the study of this more general scenario for future work. 

The fall-off in the gauge parameter following from the considerations above is 
\be 
\Lambda=\Lambda^{(0)}+...
\ee 
where $\Lambda^{(0)}$ is a function of $(z,\zb)$. The celebrated Weinberg soft theorem in gauge theory \cite{Weinberg1965} can be derived, at leading order, as a Ward identity \cite{Strominger:2013lka,He:2014cra} for the large gauge symmetry with parameter $\Lambda^{(0)}$, where the term large refers to the fact that $\Lambda^{(0)}$ does not vanish in the limit $r\to\infty$, unlike a standard gauge symmetry.

\section{Homotopy algebras}
\label{sec:Homotopy algebras and homotopy transfer}

An algebra is a vector space equipped with a bilinear map. We classify algebras according to the types of relations that the map obeys, which here we call \textit{defining relations}. For example, an associative algebra is a vector space equipped with a bilinear map that is associative, while a Lie algebra is a vector space equipped with an antisymmetric bilinear map that obeys the Jacobi identity. We call these algebras \textit{strict} because their defining relations are obeyed strictly. Homotopy algebras are generalizations of strict algebras in that the defining relations are relaxed: they are not obeyed strictly, but \textit{up to homotopy}. Mathematically, this means that there exist higher multiplicity maps (sometimes called higher brackets or higher products) that control the failure of the defining relations to hold strictly. 

In this section, we present the definition of the homotopy generalization of Lie algebras, called $L_{\infty}$ algebras, as well as the homotopy generalization of associative commutative algebras, called $C_{\infty}$ algebras, and finally, more exotic algebras called BV$_{\infty}^{\square}$ algebras. These algebras have been observed to be relevant in the mathematical description of the double copy program. In the following, we present these mathematical concepts in the context of physical theories, in particular, in the context of self-dual Yang-Mills theory.

\subsection{\texorpdfstring{$L_{\infty}$}{L infty} algebras}\label{Linfdefinitions}
To define an $L_{\infty}$ algebra, it will be useful to first define a linear structure called a \textit{cochain complex}. A cochain complex $(\mathcal{X}, B_{1})$ is a graded vector space $\mathcal{X}=\bigoplus_{i}X_{i}$ equipped with a nilpotent linear map $B_{1}:\mathcal{X}\to \mathcal{X}$ called a \textit{differential}. The integer label $i\in \mathbb{Z}$ is called the \textit{degree} of the vector space, and it is inherited by the elements of each space, i.e.,  for an element $x\in X_{i}$ we say that $x$ has degree $i$. In our conventions the differential has degree $+1$, namely $B_{1}:X_{i}\to X_{i+1}$, and nilpotency means $B_{1}(B_{1}(x))=0$. Diagrammatically, we can represent cochain complexes as follows:
\begin{equation}
\label{def:cochain_complex}
\begin{tikzcd}[row sep=10mm]
\ldots\arrow{r}{B_1}&X_{-1}\arrow{r}{B_1}&X_{0}\arrow{r}{B_1}&X_{1}\arrow{r}{B_1}&\ldots
\end{tikzcd}    \;.
\end{equation}
Let us remark that the differential can act differently on each space, as long as it is nilpotent. 

The cohomology of the differential in degree $i$ is defined as the equivalence class
\begin{equation}
\label{def:cohomology_of_B1}
    H^{i}=\frac{\text{Ker}\, B_{1}^{(i)}}{\text{Im}\, B_{1}^{(i-1)}}\;,
\end{equation}
where the superscripts in the differential denote the space where it is acting (we omit the superscripts in the following) and the kernel and image of the differential are subspaces of the cochain complex defined as
\begin{equation}
\begin{split}
    \text{Ker}\, B_{1}^{(i)}&=\{x\in X_{i}| B_{1}(x)=0\}\;,\\
    \text{Im}\, B_{1}^{(i-1)}&=\{x\in X_{i}|x = B_{1}(y)\; \text{with\;} y\in X_{i-1}\}\;.
\end{split}
\end{equation}
We say that an element of the kernel of the differential is closed, while an element of the image is exact. The cohomology is then the space of elements of $\mathcal{X}$ that are closed but not exact. A prominent example of a cochain complex is the de Rahm complex in differential geometry, where the graded vector space is the sum of the spaces of forms of all degrees and the differential is the de Rahm differential. In this case, the cohomology is related to the Euler number of the manifold where the differential forms are defined. In the $L_{\infty}$ formulation of field theories, the cohomology of the differential $H^{i}$ encodes the physics of the theory, as we shall see below when we formulate self-dual Yang-Mills theory in this framework.

An $L_{\infty}$ algebra is a cochain complex $(\mathcal{X}, B_{1})$ equipped with a set of multilinear maps (or $n$-brackets) $B_{n}:\mathcal{X}^{\otimes n}\to \mathcal{X}$ with $n\geq 1$ that obey a set of relations called generalized Jacobi relations. The maps $B_{n}$ are graded symmetric, which means that we can exchange the order of adjacent input elements of the maps at the cost of a sign determined by the degree of the elements, namely 
\begin{equation}
\label{def:graded_symmetry_Bn}
B_{n}(\ldots, x_{l}, x_{l+1},\ldots)=(-1)^{x_{l}x_{l+1}} B_{n}(\ldots, x_{l+1}, x_{l},\ldots)\;,
\end{equation}
where the $x$'s in the exponent denote the degree of the elements $x_{l}$ and $x_{l+1}$. The number of non-trivial generalized Jacobi relations depends on the number of non-trivial maps $B_{n}$ which, in principle, may be infinite. The first few generalized Jacobi relations are the nilpotency of the differential
\begin{equation}
\label{def:nilpotency_of_differential}
B_{1}(B_{1}(x))=0\;,
\end{equation}
the Leibniz rule of the differential with respect to the two-bracket $B_{2}$
\begin{equation}
\label{def:Leibniz_rule_B1_B2}
B_{1}(B_{2}(x_{1},x_{2}))+B_{2}(B_{1}(x_{1}),x_{2})+(-1)^{x_{1}}B_{2}(x_{1},B_{1}(x_{2}))=0\;,
\end{equation}
and the Jacobi identity \textit{up to homotopy}
\begin{equation}
\label{def:Jacobi_up_to_homotopy}
\begin{split}
&B_{2}(B_{2}(x_{1},x_{2}),x_{3})+(-1)^{x_{1}(x_{2}+x_{3})}B_{2}(B_{2}(x_{2},x_{3}),x_{1})+(-1)^{x_{3}(x_{1}+x_{2})}B_{2}(B_{2}(x_{3},x_{1}),x_{2})\\
+&B_{1}(B_{3}(x_{1},x_{2},x_{3}))+B_{3}(B_{1}(x_{1}), x_{2},x_{3})+(-1)^{x_{1}}B_{3}(x_{1},B_{1}(x_{2}),x_{3})\\
+&(-1)^{x_{1}+x_{2}}B_{3}(x_{1},x_{2},B_{1}(x_{3}))=0\;.
\end{split}
\end{equation}
The first line of this relation is the combination of terms that would vanish if $B_{2}$ obeyed the (graded) Jacobi identity (also called the Jacobiator). The relation states that in an $L_{\infty}$ algebra we allow the strict Jacobi identity to fail by terms involving the differential $B_{1}$ and a higher map $B_{3}$ or, in more technical terms, that the Jacobi identity is obeyed in the cohomology of the differential. For notational convenience and compactness, in the following, we write the Leibniz rule and the Jacobi identity up to homotopy as
\begin{equation}\label{eq:Leib and Jac}
    \begin{split}
    \text{Leib}(B_{1},B_{2})\equiv[B_{1},B_{2}]=0\;,\\
    \text{Jac}(B_{2})+[B_{1},B_{3}]=0\;,
    \end{split}
\end{equation}
respectively, with the Jacobiator $\text{Jac}(B_{2})$ defined as
\begin{equation}
    \begin{split}
    \text{Jac}(B_{2})(x_{1},x_{2},x_{3})=B_{2}(B_{2}(x_{1},x_{2}),x_{3})&+(-1)^{x_{1}(x_{2}+x_{3})}B_{2}(B_{2}(x_{2},x_{3}),x_{1})\\
    &+(-1)^{x_{3}(x_{1}+x_{2})}B_{2}(B_{2}(x_{3},x_{1}),x_{2})\;.
    \end{split}
\end{equation}
We also introduced the graded commutator of multilinear maps 
\begin{equation}
    [O_{1},O_{2}]=O_{1}O_{2}-(-1)^{O_{1}O_{2}}O_{2}O_{1}\;,
\end{equation}
which in the case of the homotopy Jacobi relation is given by 
\begin{equation}
\begin{split}
    [B_{1},B_{3}](x_{1},x_{2},x_{3})&= B_{1}(B_{3}(x_{1},x_{2},x_{3}))+B_{3}(B_{1}(x_{1}), x_{2},x_{3})\\
    &+(-1)^{x_{1}}B_{3}(x_{1},B_{1}(x_{2}),x_{3})
+(-1)^{x_{1}+x_{2}}B_{3}(x_{1},x_{2},B_{1}(x_{3}))\;.
\end{split}
\end{equation}
Let us emphasize that, in general, there are higher relations between the maps $B_{2}$, $B_{3}$, and $B_{4}$, and so on. These higher relations, however, will play no role in this paper, so we omit them.

$L_{\infty}$ algebras are relevant in physics because one can encode perturbative field theories as data of $L_{\infty}$ algebras. Generally speaking, the differential $B_{1}$ and the higher maps $B_{n}$ encode the free dynamics and interactions of the theory, respectively, whereas the generalized Jacobi relations encode the consistency of the theory, including gauge covariance, closure of the gauge algebra, and so on. In the following, we will describe the $L_{\infty}$ formulation of self-dual Yang-Mills theory. Let us recall that self-dual Yang-Mills theory consists of a gauge field (or one-form) $A=A_{\mu}^{a}\, dx^{\mu}\otimes t_{a}$ valued in a gauge algebra $\mathfrak{g}$ with generators $t_{a}$. In the following, for the sake of compactness, we shall use the formalism of differential forms for which we use the following conventions:
\begin{align*}
    \omega&=\frac{1}{r!}\omega_{\mu_1\dots\mu_r}\mathrm{d}x^{\mu_1}\wedge\dots\wedge\mathrm{d}x^{\mu_r}\;,\\
    \mathrm{d}\omega&=\frac{1}{r!}\partial_\nu\omega_{\mu_1\dots\mu_r}\mathrm{d}x^{\nu}\wedge\mathrm{d}x^{\mu_1}\wedge\dots\wedge\mathrm{d}x^{\mu_r}\;,\\
    *\omega&=\frac{\sqrt{|g|}}{r!(4-r)!}\omega_{\mu_1\dots\mu_r}{\epsilon^{\mu_1\dots\mu_r}}_{\nu_{r+1}\dots\nu_4}\mathrm{d}x^{\nu_{r+1}}\wedge\dots\wedge\mathrm{d}x^{\nu_4}\;.
\end{align*}

The dynamics of the gauge field are determined by the first-order equation
\begin{equation}\label{SDYMfeqs}
2\, P_{-}\, dA + P_{-}[A,A]=0\;,
\end{equation} 
where the projector onto anti-self-dual two-forms is defined as 
\begin{equation}
    P_{-}:=\frac{1}{2}(1-*)\;,
\end{equation}
while the bracket $[\cdot,\cdot]$ is a combination of the wedge product of forms and the structure constants of the Lie algebra of the color Lie group, namely, for two color-valued forms $\omega$ and $\eta$, the bracket is defined as
\begin{equation}
    [\omega,\eta] = \omega^{a}\wedge \eta^{b}\, f_{ab}{}^{c}\, t_{c}\;.
\end{equation}

The self-dual Yang-Mills field equations \eqref{SDYMfeqs} are covariant with respect to the non-abelian gauge transformation
\begin{equation}\label{gaugetfs}
\delta A=d\Lambda+[A,\Lambda]\;,
\end{equation}
generated by the $\mathfrak{g}$-valued zero-form $\Lambda = \Lambda^{a}t_{a}$.

In practice, we start formulating field theories in the framework of homotopy algebras by organizing the elements of the theory (gauge parameters, fields, and field equations) in a cochain complex. In the particular case at hand, we organize the elements of the theory in a cochain complex $(\mathcal{X}^{\rm{SDYM}}, B_{1})$ with $\mathcal{X}^{\rm{SDYM}}=\bigoplus_{i=-1}^{1}X_{i}$ as
\begin{equation}\label{SDYMXcomplex}
\begin{tikzcd}[row sep=2mm]
&X_{-1}\arrow{r}{B_{1}} & X_{0}\arrow{r}{B_{1}} & X_{1} \\
&\Lambda & A & E
\end{tikzcd}    \;,
\end{equation}
where the space $X_{-1}$ is the space of gauge parameters ($\mathfrak{g}$-valued zero-forms) $\Lambda=\Lambda^{a}t_{a}$, $X_{0}$ is the space of gauge fields ($\mathfrak{g}$-valued one-forms) $A=A_{\mu}^{a}\, dx^{\mu}\otimes t_{a}$, and $X_{1}$ is the space of field equations ($\mathfrak{g}$-valued anti-self-dual two-forms) $E = \tfrac{1}{2}E^{a}_{\mu\nu}\, dx^{\mu}\wedge dx^{\nu}\otimes t_{a}$. Next, to read off the action of the differential and the higher maps, we assume that the field equations can be written perturbatively as
\begin{equation}\label{SDYMfeqsalg}
B_{1}(A)+\frac{1}{2}B_{2}(A,A)=0\;,
\end{equation}
while the gauge transformations are assumed to be
\begin{equation}\label{gaugetfsalg}
\delta A = B_{1}(\Lambda)+B_{2}(A, \Lambda)\;.
\end{equation}
Comparing the field equations \eqref{SDYMfeqs} and gauge transformation \eqref{gaugetfs} with \eqref{SDYMfeqsalg} and \eqref{gaugetfsalg}, we conclude that the differential and the two-bracket act as 
\begin{align}\label{QB2}
B_{1}(A)&=2\, P_{-}dA\, \in X_{1}\; ,& B_{1}(\Lambda)&=d\Lambda\, \in X_{0}\nonumber\\
B_{2}(A^{1},A^{2})&=2\, P_{-}[A^{1},A^{2}]\, \in X_{1}\;, & B_{2}(A,\Lambda)&=[A,\Lambda]\, \in X_{0}\; .    
\end{align}
As we advertised previously, the differential $B_{1}$ encodes the free theory, that is, the linearized field equations and the linearized gauge transformation, while the bracket $B_{2}$ encodes the non-linearities in the field equations and gauge transformations. In addition to the above maps, there also exist the following maps that arise from the algebra of gauge transformations and gauge covariance of the field equations:
\begin{equation}\label{otherB2s}
B_{2}(\Lambda^{1},\Lambda^{2})=-[\Lambda^{1},\Lambda^{2}]\, \in X_{-1}\; ,\quad B_{2}(\Lambda,E)=-[\Lambda,E]\, \in X_{1}\; .
\end{equation}
Notice that in this case there are no higher maps than the two-bracket $B_{2}$.

Finally, one checks that the maps obey the generalized Jacobi relations mentioned above. Due to the absence of higher brackets, in self-dual Yang-Mills theory we only need to check the nilpotency of the differential, the Leibniz rule of the differential with respect to the two-bracket, and the strict graded Jacobi identity. These relations encode the consistency of the theory. As an example, let us comment on the nilpotency of the differential. Let us act twice with the differential on a gauge parameter:
\begin{equation}
\begin{split}
B_{1}(B_{1}(\Lambda))&=B_{1}(d\Lambda)\\
&=2\, P_{-}d(d\Lambda)\equiv 0\;,
\end{split}
\end{equation}
where in the last step we used $d^{2}=0$. Notice that the outermost $B_{1}$ acts on $B_{1}(\Lambda)\in X_{0}$ which is a field, which justifies the presence of the projectors $P_{-}$. The above simple computation is the statement of gauge invariance of the linearized theory. Consequently, it encodes information about the physical propagating degrees of freedom of the theory which are elements of the cohomology of the differential $B_{1}$. Indeed, the cohomology of $B_{1}$ in degree zero is
\begin{equation}\label{eq:phyisical interpretation}
H^{0}=\{A\in X_{0}\,|\,B_{1}(A)=0\; \text{and}\; A\neq B_{1}(\Lambda)\}\;,
\end{equation}
 \emph{which are field configurations that satisfy the linear field equations modulo linearized gauge transformations, which correspond to the theory's physical or propagating degrees of freedom}. In that sense, the generalized Jacobi relations encode the consistency of the theory. Proving the Leibniz rule and the graded Jacobi identity of $B_{2}$ amounts to checking that the theory is gauge covariant, that gauge transformations close, and so on. We will not explicitly show this in this paper.

Let us close this section by noting that since self-dual Yang-Mills theory has only up to bilinear terms in the field equations and gauge transformations, we only have up to bilinear brackets $B_{2}$, and no higher brackets. This algebraic structure is called a \textit{strict} $L_{\infty}$ algebra, or a \textit{differential graded Lie algebra} (dgLa).  In the full theory of Yang-Mills, however, the field equations have up to trilinear terms (coming from the quartic interactions in the action), and as a consequence, the theory also has a map $B_{3}$. In contrast, in the perturbative expansion of General Relativity, the field equations have infinite coupling terms. Thus the algebra that describes General Relativity has an infinite number of higher maps.

\subsection{\texorpdfstring{$C_{\infty}$}{C infty} algebras}
\label{sub:Cinf alg}

The double copy program in scattering amplitudes relies on \textit{color decomposition} or \textit{color stripping}: the fact that Yang-Mills scattering amplitudes\footnote{Or amplitudes of any other theory with color degrees of freedom.} factorize into color and kinematic information. Similarly, when formulating Yang-Mills theory in the framework of homotopy algebras, a color decomposition also occurs. More precisely, the $L_{\infty}$ algebra $\mathcal{X}^{\rm{YM}}$ that describes Yang-Mills theory factorizes as the tensor product\footnote{This factorization also holds for other theories with color degrees of freedom such as Chern-Simons theory and, as we show more concretely at the end of the section, self-dual Yang-Mills theory.} $\mathcal{X}^{\rm{YM}}=\mathcal{K}^{\rm{YM}}\otimes \mathfrak{g}$, where $\mathfrak{g}$ is a finite-dimensional Lie algebra that encodes color, while $\mathcal{K}^{\rm{YM}}$ is a homotopy algebra called $C_{\infty}$ algebra that generalizes associative commutative algebras. In this section, we define $C_{\infty}$ algebras and discuss the algebraic color decomposition and $C_{\infty}$ algebra of self-dual Yang-Mills theory.

$C_{\infty}$ algebras are homotopy generalizations of associative commutative algebras, where the associativity relation is relaxed to hold up to homotopy. More precisely, a $C_{\infty}$ algebra is a cochain complex $(\mathcal{K}, m_{1})$
\begin{equation}
\begin{tikzcd}[row sep=10mm]
\ldots\arrow{r}{m_1}&K_{0}\arrow{r}{m_1}&K_{1}\arrow{r}{m_1}&K_{2}\arrow{r}{m_1}&\ldots
\end{tikzcd}    \;,
\end{equation}
equipped with a set of multilinear products $m_{n}: \mathcal{K}^{\otimes n}\to \mathcal{K}$ that obey a set of homotopy associativity relations. In our conventions, the products $m_{n}$ have degree $2-n$ and have the symmetry property of \textit{vanishing on shuffles} which for the two-product means that it is graded symmetric, namely
\begin{equation}\label{eq:vanishing on shuffles}
m_{2}(u_{1},u_{2})=(-1)^{u_{1}u_{2}}m_{2}(u_{2},u_{1})\;,
\end{equation}
while for the higher products, there are more complicated relations that will not be relevant in this paper. The first few homotopy associativity relations are the nilpotency of the differential
\begin{equation}\label{eq:nilpotency of differential C_infty}
m_{1}(m_{1}(u))=0\;,
\end{equation}
the Leibniz rule of the differential with respect to the two-product $m_{2}$
\begin{equation}
m_{1}(m_{2}(u_{1},u_{2}))-m_{2}(m_{1}(u_{1}),u_{2})-(-1)^{u_{1}}m_{2}(u_{1},m_{1}(u_{2}))=0\;,
\end{equation}
and associativity up to homotopy
\begin{equation}\label{eq:associativity up to homotopy}
\begin{split}
&m_{2}(m_{2}(u_{1},u_{2}),u_{3})-m_{2}(u_{1},m_{2}(u_{2},u_{3}))=m_{1}(m_{3}(u_{1},u_{2},u_{3}))+m_{3}(m_{1}(u_{1}),u_{2},u_{3})\\
&+(-1)^{u_{1}}m_{3}(u_{1},m_{1}(u_{2}),u_{3})+(-1)^{u_{1}+u_{2}}m_{3}(u_{1},u_{2},m_{1}(u_{3}))\;.
\end{split}
\end{equation}
Similarly to the Jacobi identity up to homotopy in $L_{\infty}$ algebras, the homotopy associativity relation above tells us that the failure of the product $m_{2}$ to be associative is governed by a higher product $m_{3}$ and the differential $m_{1}$.

Let us now turn to the color decomposition of the dgLa underlying self-dual Yang-Mills theory. Recall that the elements of the vector space $\mathcal{X}^{\rm{SDYM}}$ are $\mathfrak{g}$-valued differential forms. Hence, a generic element of $\mathcal{X}^{\rm{SDYM}}$ can be written as
\begin{equation}
x = u^{a}\otimes t_{a}\;,
\end{equation}
where $u^{a}$ are either zero-, one-, or anti-self-dual two-forms, suggesting, analogously to Yang-Mills theory, the following factorization of the vector space
\begin{equation}\label{eq:factorisation}
\mathcal{X}^{\rm{SDYM}} = \mathcal{K}^{\rm{SDYM}}\otimes \mathfrak{g}\;,
\end{equation}
where $\mathcal{K}^{\rm{SDYM}}$ is an infinite-dimensional graded vector space encoding kinematics and spacetime dependence, and $\mathfrak{g}$ is a finite-dimensional vector space encoding color. Moreover, the maps that act on $\mathcal{X}^{\rm{SDYM}}$ exhibit such factorization. Indeed, given that the differential $B_{1}$ acts as a first-order differential operator, its action factorizes as
\begin{equation}
B_{1}(x)=m_{1}(u^{a})\otimes t_{a}\;.
\end{equation}
Similarly, the two-bracket $B_{2}$ factorizes as
\begin{equation}\label{eq:m2 from B2}
B_{2}(x_{1},x_{2})=(-1)^{x_{1}}m_{2}(u_{1}^{b},u_{2}^{c})\otimes f^{a}_{bc}t_{a}\;,
\end{equation}
where the sign is conventional and $f^{a}_{bc}$ are the structure constants of $\mathfrak{g}$. 

The vector space $\mathcal{K}^{\rm{SDYM}}=\bigoplus_{i=0}^{2}K_{i}$ equipped with the differential $m_{1}$ makes up a cochain complex
\begin{equation}\label{Kdiagram}
\begin{tikzcd}[row sep=2mm]
K_{0}\arrow{r}{m_{1}} & K_{1}\arrow{r}{m_{1}} & K_{2} \\
 \lambda& \cA & \mathcal{E}
\end{tikzcd}    \;,
\end{equation}
where we have dropped the color indices because they are irrelevant in this kinematic setting, and we use a different font for the elements of $\mathcal{K}$ relative to \eqref{SDYMXcomplex} to avoid confusion between colorful and color-less objects. Equipping this cochain complex with the two-product $m_{2}$ yields a strict $C_{\infty}$ algebra also called a \textit{differential graded commutative algebra}, or dgca, which can be easily shown by checking the homotopy associativity relations mentioned above using the explicit action of the maps on the elements of $\mathcal{K}^{\rm{SDYM}}$, which reads
\begin{align}\label{m2def}
    m_{1}(\lambda)&=d\lambda \in K_{1}\;, &\; \; m_{1}(\cA)&=2\, P_{-}d\cA\in K_{2}\nonumber\\
    m_{2}(\lambda_{1},\lambda_{2})&=\lambda_{1}\wedge \lambda_{2}\in K_{0}\; , &\;\; m_{2}(\lambda,\cA)&=\lambda\wedge \cA\in K_{1}\\
    m_{2}(\cA_{1},\cA_{2})&=2\, P_{-}\big(\cA_{1}\wedge\cA_{2}\big)\in K_{2}\; , &\;\; m_{2}(\lambda,\mathcal{E})&=\lambda\wedge\mathcal{E}\in K_{2}\;\nonumber .
\end{align}

\subsection{Kinematic algebras and BV\texorpdfstring{$_{\infty}^{\square}$}{BV infty box} algebras}
\label{sub:kin alg gen}

As discussed in \cref{ckduality}, a necessary condition for the double copy of scattering amplitudes is that color-kinematics duality be\ fulfilled. In that section, we discussed the duality in terms of algebraic relations that the kinematic numerators of Yang-Mills amplitudes have to obey. We further showed that, for the particular case of self-dual Yang-Mills theory in light-cone gauge, there is a \textit{kinematic algebra} underlying the relations between the numerators, namely the Lie algebra of area-preserving diffeomorphisms. Recently, in \cite{Bonezzi:2023pox}, using the framework of homotopy algebras, a more general, gauge-independent algebra for self-dual Yang-Mills theory was constructed which, upon imposing light-cone gauge and eliminating an unphysical degree of freedom, led to the Lie algebra of area-preserving diffeomorphisms. The gauge-independent algebraic structure constructed in \cite{Bonezzi:2023pox} is a so-called kinematic BV$_{\infty}^{\square}$ algebra. BV$_{\infty}^{\square}$ algebras were first introduced by Reiterer in \cite{Reiterer:2019dys} to construct a kinematic algebra for pure Yang-Mills theory and later used in \cite{Bonezzi:2022bse,Bonezzi:2022yuh,Bonezzi:2023ced,Bonezzi:2024dlv,Bonezzi:2023xhn,Bonezzi:2023lkx,Bonezzi:2023ciu} to construct gravity (and double field theory) using off-shell algebraic methods. The main goal of this section is to introduce BV$_{\infty}^{\square}$ algebras and explain how they can be constructed in the context of color-kinematics duality and the double-copy, especially in self-dual Yang-Mills theory.

Before formally introducing BV$_{\infty}^{\square}$ algebras, it will be illustrative to first recall the definition of a Poisson algebra. A Poisson algebra is a vector space $V$ equipped with a commutative associative two-product $\circ:V\times V\to V$ and a Lie bracket $\{\cdot, \cdot\}:V\times V \to V$ that obey the following compatibility relation:
\begin{equation}\label{compa}
\{f\circ g, h\} - \{g,h\}\circ f-\{f,h\}\circ g=0\;,\;\;\; f,g,h\in V.
\end{equation}
A common Poisson algebra in physics is the Poisson algebra of classical mechanics, where $V$ is the phase space, the associative commutative product $\circ$ is the product of functions of phase space, and the Lie bracket is the Poisson bracket of classical mechanics. In words, Poisson algebras consist of an associative, commutative subsector given by $\circ$, and a Lie subsector given by $\{\cdot, \cdot\}$. Both subsectors are united by the compatibility relation \eqref{compa}. BV$_{\infty}^{\square}$ algebras are homotopy generalizations of Poisson algebras, as we show in the following. 

To construct a kinematic BV$_{\infty}^{\square}$ algebra for self-dual Yang-Mills, we start with the dgca $(\mathcal{K}^{\rm{SDYM}}, m_{1}, m_{2})$ that we found after color stripping in the previous section, and introduce a nilpotent map $b:\mathcal{K}^{\rm{SDYM}}\to \mathcal{K}^{\rm{SDYM}}$ of degree $-1$. Diagrammatically, $b$ acts as
\begin{equation}
\begin{tikzcd}[row sep=4mm]
\ldots\arrow[r, "m_{1} "]&\arrow[l, bend left=50, "b"] K_{0}\arrow[r, "m_{1} "]&\arrow[l, bend left=50, "b"] K_{1}\arrow[r, "m_{1}"]&\arrow[l, bend left=50, "b"]  K_{2} \arrow[r, "m_{1} "]&\arrow[l, bend left=50, "b"] \ldots\;,
\end{tikzcd}    
\end{equation}
and (graded) commutes with the differential $m_{1}$ to the d'Alembert operator in Minkowski space\footnote{The explicit form of $\Box$ naturally depends on our choice of coordinates.}, namely
\begin{equation}
\label{def: box}
[b,m_{1}](u)\equiv m_{1}(b(u))+b(m_{1}(u)) = \Box u\;,\;\;\; u\in \mathcal{K}^{\rm{SDYM}}\;.
\end{equation}
The operator $b$ that satisfies the condition above is the adjoint de Rham operator
\begin{equation}
\label{def:differential b form notation}
b = d^{\dagger}:=-\star d\star\;.
\end{equation}
This operator, as one can explicitly check, does not obey the Leibniz rule with respect to $m_{2}$. We parametrize this failure by defining a \textit{kinematic two-bracket} 
\begin{equation}
\label{def:b2}
b_{2}(u_{1},u_{2}):=[b, m_{2}](u_{1},u_{2})\equiv bm_{2}(u_{1}, u_{2})-m_{2}(bu_{1},u_{2})-(-1)^{u_{1}}m_{2}(u_{1},bu_{2})\;,
\end{equation}
which acts explicitly as
\begin{equation}\label{b2_gen_exprs}
\begin{split}
    &b_{2}(\lambda,\cA)=d^{\dagger}(\cA\wedge\lambda)-d^{\dagger}\cA\wedge\lambda\in K_{0}\; ,\\
    &b_{2}(\cA_{1},\cA_{2})=2\, d^{\dagger} P_{-}(\cA_{1}\wedge \cA_{2})-2\, d^{\dagger}\cA_{[1}\wedge \cA_{2]}\in K_{1}\; ,\\
    &b_{2}(\lambda,\mathcal{E})=\star(\mathcal{E}\wedge d\lambda)\in K_{1}\; ,\\
    &b_{2}(\cA,\mathcal{E})=2\, P_{-}(\cA\wedge d^{\dagger}\mathcal{E})-d^{\dagger}\cA\wedge \mathcal{E}\in K_{2}\; ,
\end{split}
\end{equation}

A BV$_{\infty}^{\square}$ algebra is a collection $(\mathcal{K}, m_{1}, b, \Box, m_{2}, b_{2}, \theta_{3},\ldots)$ where the two-product $m_{2}$ and the two-bracket $b_{2}$ obey a generalization of the Poisson compatibility relation \eqref{compa}. More precisely, the maps obey the Poisson compatibility relation up to terms involving a trilinear map $\theta_{3}$, i.e.,
\begin{equation}\label{hompoi}
\begin{split}
    b_{2}(m_{2}(u_{1},u_{2}),u_{3})&-(-1)^{u_{1}(u_{2}+u_{3})}m_{2}(b_{2}(u_{2},u_{3}),u_{1})\\
    &-(-1)^{u_{3}(u_{1}+u_{2})}m_{2}(b_{2}(u_{3},u_{1}),u_{2})=[m_{1},\theta_{3}](u_{1},u_{2},u_{3})\;,
\end{split}
\end{equation}
with
\begin{equation}\label{commutator}
\begin{split}
    [m_{1},\theta_{3}](u_{1},u_{2},u_{3}):=m_{1}\theta_{3}(u_{1},u_{2},u_{3})-\theta_{3}(m_{1}u_{1},u_{2},u_{3})&-(-1)^{u_{1}}\theta_{3}(u_{1},m_{1}u_{2},u_{3})\\
    &-(-1)^{u_{1}+u_{2}}\theta_{3}(u_{1},u_{2},m_{1}u_{3})\; ,
\end{split}
\end{equation}
Notice that, upon identifying $b_{2}$ with $\{\cdot,\cdot\}$ and $m_{2}$ with $\circ$, the left-hand side of equation \eqref{hompoi} has the same combination of terms (albeit graded) as the left-hand side of equation \eqref{compa}. In the former, however, the right-hand side is non-vanishing and the trilinear map $\theta_{3}$ is a homotopy map that controls the failure of $m_{2}$ and $b_{2}$ to be Poisson compatible. Indeed, $\theta_{3}$ plays the same role as $B_{3}$ in $L_{\infty}$ algebras or $m_{3}$ in $C_{\infty}$ algebras. The homotopy map $\theta_{3}$ has the following two non-vanishing components:
\begin{equation}\label{theta3}
\begin{split}
\theta_{3}(\cA_{1},\cA_{2},\cA_{3})&=-\star \big(\cA_{1}\wedge \cA_{2}\wedge \cA_{3}\big)\in K_{1}\; ,\\
\theta_{3}(\mathcal{E} ,\cA_{1},\cA_{2})&=2\, P_{-}\Big\{ \star\big(\mathcal{E} \wedge\cA_{[1}\big)\wedge \cA_{2]}  \Big\}\in K_{2}\; .
\end{split}
\end{equation}
which were determined by explicit computation in \cite{Bonezzi:2023pox}. To higher order in inputs, a BV$_{\infty}^{\square}$ algebra has further compatibility relations that the non-linear maps obey, which means that the algebra may have higher homotopies, e.g., a quadri-linear map $\theta_{4}$, and so on. In this paper, however, we only consider this algebra up to trilinear maps because we are interested in finding a strict algebra with a vanishing $\theta_{3}$, i.e., we wish to find a regime where $m_{2}$ and $b_{2}$ are Poisson compatible and there are no higher homotopies $\theta_{n}$.

Let us remark that the general definition of a BV$_{\infty}^{\square}$ algebra allows one to have associativity up to homotopy, as is the case in pure Yang-Mills theory, where after color stripping one obtains a $C_{\infty}$ algebra with a non-trivial $m_{3}$. For a more detailed description of the construction of BV$_{\infty}^{\square}$ algebras, refer to \cite{Reiterer:2019dys, Bonezzi:2022bse}.

The failure of $m_{2}$ and $b_{2}$ to be Poisson compatible, and the fact that $b_{2}$ is derived from $m_{2}$, implies that the kinematic bracket $b_{2}$ does not obey the strict Jacobi identity. Instead, it obeys the Jacobi identity up to homotopy and \textit{modulo box}, i.e.,
\begin{equation}
\begin{split}
\text{Jac}(b_{2})(u_{1},u_{2},u_{3})+[m_{1} ,b_{3}](u_{1},u_{2},u_{3})+[\Box,\theta_3](u_1,u_2,u_3)=0\;,    
\end{split}    
\end{equation}
where the kinematic three-bracket $b_{3}$ is derived as $b_3=-[b,\theta_3]$ and the graded commutators that parametrize the failure of the Jacobi identity are given by
\begin{equation}
\begin{split}
[m_{1},b_{3}](u_{1},u_{2},u_{3}):=\; &m_{1}b_{3}(u_{1},u_{2},u_{3})+b_{3}(m_{1}u_{1},u_{2},u_{3})+(-1)^{u_{1}}b_{3}(u_{1},m_{1}u_{2},u_{3})\\
&+(-1)^{u_{1}+u_{2}}b_{3}(u_{1},u_{2},m_{1}u_{3})\;,\\
[\Box ,\theta_{3}](u_{1},u_{2},u_{3}):=\; &\Box \theta_{3}(u_{1},u_{2},u_{3})-\theta_{3}(\Box u_{1},u_{2},u_{3})-\theta_{3}(u_{1},\Box u_{2},u_{3})-\theta_{3}(u_{1},u_{2},\Box u_{3})\;,\\
[b,\theta_{3}](u_{1},u_{2},u_{3}):=\; &b\, \theta_{3}(u_{1},u_{2},u_{3})-\theta_{3}(b u_{1},u_{2},u_{3})-(-1)^{u_{1}}\theta_{3}(u_{1},b u_{2},u_{3})\\
    &-(-1)^{u_{1}+u_{2}}\theta_{3}(u_{1},u_{2},b u_{3})\; .
\end{split}    
\end{equation}
Hence, the gauge-invariant kinematic BV$_{\infty}^{\square}$ we constructed here does not have a strict Lie algebra subsector. Notice, however, that for a vanishing $\theta_{3}$, which in turn implies a vanishing $b_{3}$, $b_{2}$ would obey the strict Jacobi identity. Indeed, as shown in \cite{Bonezzi:2023pox}, if one imposes light-cone gauge and integrates out an unphysical degree of freedom, then the kinematic bracket $b_{2}$ reduces, in general, to the Schouten-Nijenhuis bracket of polyvectors, which for the particular case of two gauge fields is the Lie bracket of area-preserving diffeomorphisms
\be \label{SNbracket}
b_{2}(\cA_{1},\cA_{2})^{m}=\big[\cA_1,\cA_2\big]^m_{\rm SN}=\cA_1^n\del_n\cA_2^m-\cA_2^n\del_n\cA_1^m \ .
\ee
This tells us that, to significantly simplify the kinematic algebra and obtain a strict kinematic Lie algebra, one needs to find either particular gauge choices or regimes of the self-dual Yang-Mills.

\subsection{Summary and notation}


\begin{table}[h!]
\centering
\begin{tabularx}{\textwidth}{|l|X|X|}
\hline
\cellcolor[gray]{0.9}Algebra & \cellcolor[gray]{0.9}Definition & \cellcolor[gray]{0.9}Physics \\
\hline
$L_\infty$ & Cochain complex $(\mathcal{X}, B_{1})$ equipped with a set of graded symmetric \eqref{def:graded_symmetry_Bn} multilinear maps ($n$-brackets) $\{B_{n}:\mathcal{X}^{\otimes n}\to \mathcal{X}\}_{n\geq1}$ obeying generalized Jacobi relations, the first of which are \eqref{def:nilpotency_of_differential} and \eqref{eq:Leib and Jac}. Homotopy generalisation of \textit{Lie algebras}. If $B_n$ vanishes for all $n\geq3$ it is a differential graded Lie algebra (dgLa).& Encodes perturbative field theories: the differential $B_{1}$ and the higher maps $B_{n}$ describe the free dynamics and interactions, respectively, whereas the generalized Jacobi relations ensure the consistency of the theory (e.g. gauge covariance and  closure of the gauge algebra). See for instance equation \eqref{eq:phyisical interpretation}.\\
\hline
$C_\infty$ & Cochain complex $(\mathcal{K}, m_{1})$ equipped with a set of graded symmetric \eqref{eq:vanishing on shuffles} multilinear maps ($n$-products) $\{m_{n}:\mathcal{K}^{\otimes n}\to \mathcal{K}\}_{n\geq1}$ obeying homotopy associativity relations, the first of which are \eqref{eq:nilpotency of differential C_infty}-\eqref{eq:associativity up to homotopy}. Homotopy generalisation of \textit{associative commutative algebras}. If $m_n$ vanishes for all $n\geq3$ it is a differential graded commutative algebra (dgca). & Encodes the color-stripped version of theories with color degrees of freedom within the framework of homotopy algebras. For instance, given the $L_\infty$ algebra describing self-dual Yang-Mills, the corresponding $C_\infty$ algebra comes from the decomposition \eqref{eq:factorisation}, by equipping the graded vector space $\mathcal{K}^\mathrm{SDYM}$ with the two-product $m_2$ arising from \eqref{eq:m2 from B2}. \\
\hline
$\mathrm{BV}_\infty^\Box$ & Collection $(\mathcal{K}, m_{1}, b, \Box, m_{2}, b_{2}, \theta_{3},\ldots)$ where
\begin{itemize}[noitemsep, topsep=0pt, parsep=0pt, partopsep=0pt, label=\tiny$\bullet$]
    \item $(\mathcal{K},m_1,m_2,\dots)$ is a $C_\infty$ algebra,
    \item $b\colon\mathcal{K}\rightarrow\mathcal{K}$ is a nilpotent map of degree $-1$ satisfying the defining  relation \eqref{def: box} with the operator $\Box$,
    \item $b_2$ is the kinematic two-bracket \eqref{def:b2}, satisfying \eqref{hompoi},
    \item $\theta_3$ is the trilinear map appearing in \eqref{hompoi}.
\end{itemize}
Homotopy generalisation of \textit{Poisson algebras}. & \textit{Kinematic algebra} encoding the relations that the kinematic numerators in a theory with color degrees of freedom have to obey, in the context of the color-kinematics duality. For instance, for self-dual Yang-Mills it is constructed from the dgca $(\mathcal{K}^\mathrm{SDYM},m_1,m_2)$ by introducing the adjoint de Rham operator $b$ \eqref{def:differential b form notation} and the d'Alembert operator $\Box$.  \\
\hline
\end{tabularx}
\caption{Homotopy algebras defined in \cref{sec:Homotopy algebras and homotopy transfer}.}
\end{table}

\section{Homotopy algebra formulation of SDYM in Bondi coordinates}
\label{Homotopy algebra formulation of SDYM in Bondi coordinates}
We now wish to translate the results of the previous section in a language appropriate for taking the limit to null infinity. We choose to work with the flat Bondi coordinates 
\be 
x^\mu=(r,u,z,\zb)\equiv (r,y^\alpha) \ .
\ee 
due to the simplicity of the map from the light-cone coordinates, see \eqref{LCtoflB}. In the the flat Bondi coordinates, the celestial sphere is replaced by the complex plane via a stereographic projection. The flat Bondi metric, given in \eqref{fl_B_metric}, is reproduced below for convenience:
\begin{equation}\label{flat_Bdi}
    ds^2=-2drdu+2r^2dzd\Bar{z} \ .
\end{equation}

\subsection{\texorpdfstring{$L_\infty$}{L infty} algebra}
\label{subsec:algebra for SDYM in the bulk}
Let $\mathcal{M}$ be a manifold with global coordinates $x^\mu=(r,y^\alpha)$. We define $C_r^\infty(\mathcal{M})$ as the space of all \(C^\infty\) functions on the manifold $\mathcal{M}$ that admit an expansion in powers of \(r\):
\begin{equation}
    \label{def:C^infty_r}
    C_r^\infty(\mathcal{M}):=\{f\in C^\infty(\mathcal{M}) \ | \ f(r,y^\alpha)=\sum_{n\in\mathbb{Z}} r^n f^{(n)}(y^\alpha)\} \ .
\end{equation}
As in \cref{Linfdefinitions}, we define the following graded vector space:
    \begin{equation}
    \label{def:initial_graded_vector_space}
        \mathcal{X}:=\bigoplus_{i=-1}^{1}X_i \ .
    \end{equation}
where \(X_{-1},X_0\) and \(X_1\) are, respectively, the spaces of scalars \(\lambda\), covariant vectors \(A_\mu\) and skew-symmetric anti-self-dual \((0,2)\)-tensors \(E_{\mu\nu}\) over the manifold \(M\). Moreover, \(\forall\psi\in\chi\), \(\psi\) is \(\mathfrak{g}\)-valued and has the expansion
\begin{equation}
    \label{def:general_expansion_for_psi}
    \psi(r,u,z,\Bar{z})=\sum_{n=-\infty}^{+\infty}r^n\psi^{(n)}(u,z,\Bar{z}) \  .
\end{equation}
Note that at this point we haven't yet truncated, i.e. we consider all powers in $r$. Using the skew-symmetry and the anti-self-duality of elements in $X_1$, we find that a suitable basis is given by \(\mathcal{B}_{X_1}=\{\underline{e}^{i}_{\mu\nu}\}_{i=1,2,3}\), with
\begin{equation}
\label{def:basis_for_X1}
    \underline{e}^1_{\mu\nu}=\begin{pmatrix}
        0 & 1 & 0 & 0 \\
        -1 & 0 & 0 & 0 \\
        0 & 0 & 0 & r^2 \\
        0 & 0 & -r^2 & 0
    \end{pmatrix}, \qquad
    \underline{e}^2_{\mu\nu}=\begin{pmatrix}
        0 & 0 & 0 & 1 \\
        0 & 0 & 0 & 0 \\
        0 & 0 & 0 & 0 \\
        -1 & 0 & 0 & 0
    \end{pmatrix}, \qquad
    \underline{e}^3_{\mu\nu}=\begin{pmatrix}
        0 & 0 & 0 & 0 \\
        0 & 0 & 1 & 0 \\
        0 & -1 & 0 & 0 \\
        0 & 0 & 0 & 0
    \end{pmatrix}.
\end{equation}
Therefore, a generic element \(E_{\mu\nu}\in X_1\) is of the form
\begin{equation}
    \label{eq:E_full_generl_form}
    E_{\mu\nu}=\begin{pmatrix}
        0 & E_{ru} & 0 & E_{r\Bar{z}} \\
        -E_{ru} & 0 & E_{uz} & 0 \\
        0 & -E_{uz} & 0 & r^2E_{ru} \\
        -E_{r\Bar{z}} & 0 & -r^2E_{ru} & 0
    \end{pmatrix}.
\end{equation}
In other words, in order to determine completely a tensor in \(X_1\) it is sufficient to specify three of its components, namely  \(E_{ru}\), \(E_{r\Bar{z}}\) and \(E_{uz}\). The $B_1$ operator (\(B_1\colon\mathcal{X}\rightarrow\mathcal{X}\)) defined in \eqref{QB2}  becomes, in Bondi coordinates,
\begin{equation}
    \label{def:B_1 in components}
    \begin{aligned}
        B_1(\Lambda)_\mu&=\partial_\mu\Lambda\in X_0 \\
        B_1(A)_{\mu\nu}&=2\left(\partial_{[r}A_{u]}+r^{-2}\partial_{[z}A_{\Bar{z}]}\right)\underline{e}^1_{\mu\nu}+4\partial_{[r}A_{\Bar{z}]}\underline{e}^2_{\mu\nu}+4\partial_{[u}A_{z]}\underline{e}^3_{\mu\nu}\in X_1\ .
    \end{aligned}
\end{equation}
Similarly
\(B_2\colon\mathcal{X}\times\mathcal{X}\rightarrow\mathcal{X}\) can be written as follows:
\begin{equation}
    \label{def:B_2 in components}
    \begin{aligned}
        B_2(\Lambda_1,\Lambda_2)&=-[\Lambda_1,\Lambda_2]\in X_{-1}\\
        B_2(\Lambda,A)_\mu&=-[\Lambda,A_\mu]\in X_0\\
        B_2(\Lambda,E)_{\mu\nu}&=-[\Lambda,E_{\mu\nu}]\in X_1\\
        B_2(A_1,A_2)_{\mu\nu}&=2 \left([A_{1\,[r},A_{2\,u]}]+r^{-2}[A_{1\,[z},A_{2\,\Bar{z}]}] \right)\underline{e}^1_{\mu\nu} \\
        &\quad+4[A_{1\,[r},A_{2\,\Bar{z}]}]\underline{e}^2_{\mu\nu}+4[A_{1\,[u},A_{2\,z]}]\underline{e}^3_{\mu\nu}\in X_1 \ .
    \end{aligned}
\end{equation}
Then \((\mathcal{X},B_1,B_2)\) is an \(L_\infty\) algebra, with $B_1$ and $B_2$ satisfying the generalized Jacobi relations \eqref{def:nilpotency_of_differential}, \eqref{def:Leibniz_rule_B1_B2}, and \eqref{def:Jacobi_up_to_homotopy} with $B_3=0$.

\subsection{Color stripping and kinematic algebra}
\label{sec:color stripping and kinematic algebra}
We color-strip as in \cref{sub:Cinf alg}, to construct the differential graded commutative algebra \((\mathcal{K},m_1,m_2)\)\footnote{Recall that we use a different font relative to the previous section to denote the color-stripped scalars, vectors, and anti-self-dual 2-forms.}, with

\begin{equation}
    \label{def:m_1 in components}
    \begin{aligned}
        m_1(\lambda)_\mu&=\partial_\mu\lambda\in K_1 \\
        m_1(\mathcal{A})_{\mu\nu}&=2\left(\partial_{[r}\mathcal{A}_{u]}+r^{-2}\partial_{[z}\mathcal{A}_{\Bar{z}]}\right)\underline{e}^1_{\mu\nu}+4\partial_{[r}\mathcal{A}_{\Bar{z}]}\underline{e}^2_{\mu\nu}+4\partial_{[u}\mathcal{A}_{z]}\underline{e}^3_{\mu\nu}\in K_2
    \end{aligned}
\end{equation}
and
\begin{equation}
    \label{def:m_2 in components}
    \begin{aligned}
        m_2(\lambda_1,\lambda_2)&=\lambda_1\lambda_2\in K_{0}\\
        m_2(\lambda,\mathcal{A})_\mu&=\lambda \mathcal{A}_\mu\in K_1\\
        m_2(\lambda,\mathcal{E})_{\mu\nu}&=\lambda\mathcal{E}_{\mu\nu}\in K_2\\
        m_2(\mathcal{A}_1,\mathcal{A}_2)_{\mu\nu}&=2 \left(\mathcal{A}_{[1\,r}\mathcal{A}_{2]\,u}+r^{-2}\mathcal{A}_{[1\,z}\mathcal{A}_{2]\,\Bar{z}} \right)\underline{e}^1_{\mu\nu} \\
        &\quad +4\mathcal{A}_{[1\,r}\mathcal{A}_{2]\,\Bar{z}}\,\underline{e}^2_{\mu\nu}+4\mathcal{A}_{[1\,u}\mathcal{A}_{2]\,z}\,\underline{e}^3_{\mu\nu}\in K_2 \;.
    \end{aligned}
\end{equation}
The construction of the kinematic algebra
\begin{equation}
\label{def:kinematic algebra bulk}
    \mathfrak{Kin}=(\mathcal{K},m_1,b,\Box,m_2,b_2,\theta_3)
\end{equation}
proceeds as in \cref{sub:kin alg gen}. The main difference is that for the coordinate system \eqref{flat_Bdi}, not all Christoffel symbols vanish, so, for example, the action of $b$ is
\begin{equation}
    \label{def:differential b in components}
    \begin{aligned}
        b(\mathcal{A})&=-\nabla^\mu\mathcal{A}_\mu=2\left(\partial_{(r}\mathcal{A}_{u)}-r^{-2}\partial_{(z}\mathcal{A}_{\bar z)}+r^{-1}\mathcal{A}_u\right)\in K_0 \\
        b(\mathcal{E})_\mu&=\nabla^\nu\mathcal{E}_{\mu\nu}\in K_1 \;,
    \end{aligned}
\end{equation}
\black
while
\be
\Box\coloneqq dd^\dagger+d^\dagger d=-\nabla^2=-\partial^2+2r^{-1}\mathfrak{c} \;,
\ee
where we define
\begin{equation}
    \label{def:partial square}
    \nabla^2\coloneqq\nabla_\mu\nabla^\mu\;, \quad \partial^2\coloneqq\partial_\mu\partial^\mu=-2\partial_r\partial_u+2r^{-2}\partial_z\partial_{\bar z} \;,
\end{equation}
and the operator $\mathfrak{c}$ (that comes from the non-vanishing Christoffel symbols) acts as
\begin{equation}
    \begin{aligned}
        \mathfrak{c}(\Lambda)&=\partial_u\Lambda \\
        \mathfrak{c}(\mathcal{A})_\mu&=\partial_\mu\mathcal{A}_u-\delta_{\mu r}\left(2\partial_{[r}\mathcal{A}_{u]}-2r^{-2}\partial_{(z}\mathcal{A}_{\bar z)}+r^{-1}\mathcal{A}_u\right) \\
        \mathfrak{c}(\mathcal{E})_{\mu\nu}&=\left(\partial_u\mathcal{E}_{ru}-r^{-2}\partial_{\bar z}\mathcal{E}_{uz}\right)\underline{e}^1_{\mu\nu}+2\partial_{\bar z}\mathcal{E}_{ru}\underline{e}^2_{\mu\nu} \;.
    \end{aligned}
\end{equation}
Finally, we will need the explicit form of the kinematic bracket $b_2$ (between two gauge fields) and of the trilinear map $\theta_3$ in these coordinates. From \eqref{def:b2} and \eqref{theta3} we obtain, respectively,
\begin{equation}
\label{sec4:kin bracket}
    \begin{split}
        \frac{1}{2}b_2(\mathcal{A}_1,\mathcal{A}_2)_\mu&=\begin{pmatrix}
            \mathcal{A}_{1\,(r}\partial_{u)}\mathcal{A}_{2\,r} \\
            \mathcal{A}_{1\,(r}\partial_{u)}\mathcal{A}_{2\,u} \\
            \mathcal{A}_{1\, u}\partial_{r}\mathcal{A}_{2\,z}+\mathcal{A}_{1\, [r}\partial_{z]}\mathcal{A}_{2\,u}-\mathcal{A}_{1\, [u}\partial_{z]}\mathcal{A}_{2\,r} \\
            \mathcal{A}_{1\, r}\partial_{u}\mathcal{A}_{2\,\bar z}-\mathcal{A}_{1\, [r}\partial_{\bar z]}\mathcal{A}_{2\,u}+\mathcal{A}_{1\, [u}\partial_{\bar z]}\mathcal{A}_{2\,r}
        \end{pmatrix}-\frac{1}{r}\begin{pmatrix}
            0 \\
            0 \\
            \mathcal{A}_{1\,u}\mathcal{A}_{2\,z} \\
            \mathcal{A}_{1\,u}\mathcal{A}_{2\,\bar z}
        \end{pmatrix} \\
        &\quad -\frac{1}{r^2}\begin{pmatrix}
            \mathcal{A}_{1\,[r}\partial_{\bar z]}\mathcal{A}_{2\,z}-\mathcal{A}_{1\,[r}\partial_{z]}\mathcal{A}_{2\,\bar z}+\mathcal{A}_{1\,\bar z}\partial_{z}\mathcal{A}_{2\,r} \\
            -\mathcal{A}_{1\,[u}\partial_{\bar z]}\mathcal{A}_{2\,z}+\mathcal{A}_{1\,[u}\partial_{z]}\mathcal{A}_{2\,\bar z}+\mathcal{A}_{1\, z}\partial_{\bar z}\mathcal{A}_{2\,u} \\
            \mathcal{A}_{1\, (z}\partial_{\bar z)}\mathcal{A}_{2\,z} \\
            \mathcal{A}_{1\, (z}\partial_{\bar z)}\mathcal{A}_{2\,\bar z}
        \end{pmatrix}-(1\leftrightarrow2)
    \end{split}
\end{equation}
and
\begin{equation}
    \label{eq:theta3AAA explicit expression}
    \begin{split}
        \theta_{3}(\cA_{1},\cA_{2},\cA_{3})_\mu&=6\sqrt{\det g}\;\varepsilon_{\mu\nu\rho\sigma}g^{\nu\tau}g^{\rho\delta}g^{\sigma\epsilon}\mathcal{A}_{[1\,\tau}\mathcal{A}_{2\,\delta}\mathcal{A}_{3]\,\epsilon}=6\begin{pmatrix}
            r^{-2}\mathcal{A}_{[1\,r}\mathcal{A}_{2\,z}\mathcal{A}_{3]\,\bar z} \\
            -r^{-2}\mathcal{A}_{[1\,u}\mathcal{A}_{2\,z}\mathcal{A}_{3]\,\bar z} \\
            -\mathcal{A}_{[1\,r}\mathcal{A}_{2\,u}\mathcal{A}_{3]\,z} \\
            \mathcal{A}_{[1\,r}\mathcal{A}_{2\,u}\mathcal{A}_{3]\,\bar z}
        \end{pmatrix}\in K_1 \;,
    \end{split}
\end{equation}
\begin{equation}
    \label{eq:theta3EAA explicit expression}
    \begin{aligned}
        \theta_3(\mathcal{E},\mathcal{A}_1,\mathcal{A}_2)_{\mu\nu}&=2r^{-2}\left(\mathcal{E}_{uz}\mathcal{A}_{[1\,r}\mathcal{A}_{2]\,\bar z}-\mathcal{E}_{r\bar z}\mathcal{A}_{[1\,u}\mathcal{A}_{2]\,z}\right)\underline{e}^1_{\mu\nu} \\
        &\quad -\left[4\mathcal{E}_{ru}\mathcal{A}_{[1\,r}\mathcal{A}_{2]\,\bar z}-2\mathcal{E}_{r\bar z}\left(\mathcal{A}_{[1\,r}\mathcal{A}_{2]\,u}+r^{-2}\mathcal{A}_{[1\,z}\mathcal{A}_{2]\,\bar z}\right)\right]\underline{e}^2_{\mu\nu} \\        &\quad +\left[4\mathcal{E}_{ru}\mathcal{A}_{[1\,u}\mathcal{A}_{2]\,z}-2\mathcal{E}_{uz}\left(\mathcal{A}_{[1\,r}\mathcal{A}_{2]\,u}+r^{-2}\mathcal{A}_{[1\,z}\mathcal{A}_{2]\,\bar z}\right)\right]\underline{e}^3_{\mu\nu}\in K_2 \;,
    \end{aligned}
\end{equation}
for $\mathcal{A}_1,\mathcal{A}_2,\mathcal{A}_3\in K_1$ and $\mathcal{E}\in K_2$. One of our main goals in this article is to look for strict kinematic algebras close to null infinity, in other words study the conditions under which $\theta_3$ vanishes.

To connect with previous work, if one chooses the gauge $A_u=0$, which is equivalent to light-cone gauge in Bondi coordinates, and one partially solves the field equations which impose $A_{z}=0$ and $\partial_u\mathcal{A}_r=r^{-2}\partial_z\mathcal{A}_{\bar z}$, one finds the components of the strict kinematic bracket
\begin{equation}
\label{SN_bracket_in_Bondi_coords}
    \begin{split}
        b_2(\mathcal{A}_1,\mathcal{A}_2)_r&=2\mathcal{A}_{[1\,r}\partial_{u}\mathcal{A}_{2]\,r}-r^{-2}\left(-2\mathcal{A}_{[1\,r}\partial_{z}\mathcal{A}_{2]\,\bar z}+4\mathcal{A}_{[1\,\bar z}\partial_{z}\mathcal{A}_{2]\,r}\right) \\
        &=4\mathcal{A}_{[1\,r}\partial_{u}\mathcal{A}_{2]\,r}-4r^{-2}\mathcal{A}_{[1\,\bar z}\partial_{z}\mathcal{A}_{2]\,r}=-2[\mathcal{A}_1,\mathcal{A}_2]^\mathrm{SN}_r \\
        b_2(\mathcal{A}_1,\mathcal{A}_2)_u&=b_2(\mathcal{A}_1,\mathcal{A}_2)_z=0 \\
        b_2(\mathcal{A}_1,\mathcal{A}_2)_{\bar z}&=4\mathcal{A}_{[1\, r}\partial_{u}\mathcal{A}_{2]\,\bar z}-2\mathcal{A}_{[1\, \bar z}\partial_{u}\mathcal{A}_{2]\,r}-2r^{-2}\mathcal{A}_{[1\, \bar z}\partial_{z}\mathcal{A}_{2]\,\bar z} \\
        &=4\mathcal{A}_{[1\, r}\partial_{u}\mathcal{A}_{2]\,\bar z}-4r^{-2}\mathcal{A}_{[1\, \bar z}\partial_{z}\mathcal{A}_{2]\,\bar z}=-2[\mathcal{A}_1,\mathcal{A}_2]^\mathrm{SN}_{\bar z}\;,
    \end{split}
\end{equation}
which correspond to the components of the Schouten-Nijenhuis bracket \eqref{SNbracket} in Bondi coordinates.

\section{Physical fall-offs from $L_{\infty}$ algebras}
\label{sec:SDYM at null infinity via slice truncation}
In this section and the following one, we formulate self-dual Yang-Mills theory when approaching $\mathcal{I}$, i.e. when $r\rightarrow\infty$, using the homotopy algebra formalism again. That is to say, starting from the algebra $(\mathcal{X},B_1,B_2)$ presented in \cref{subsec:algebra for SDYM in the bulk}, we build a new $L_\infty$ algebra
\begin{equation}
    \label{def:general_expression_algebra_I}
    (\mathcal{X}^\mathcal{I},B_1^\mathcal{I},B_2^\mathcal{I})
\end{equation}
that encodes self-dual Yang-Mills in the regime close to null infinity.

 In practice, when doing computations close to null infinity, one expands gauge fields and parameters in powers of the radial coordinate $r$ and then truncates the expansion at a certain order. The tools that we introduce here will allow us to implement such truncations algebraically to construct $L_{\infty}$ algebras that describe the physics of Yang-Mills theory in certain regimes or slices close to null infinity, and consequently, their respective kinematic algebras.

\subsection{Requirements on the \texorpdfstring{$L_\infty$}{L infty} algebra near \texorpdfstring{$\mathcal{I}$}{I}}
\label{subsec:Requirements on L infty}
Before moving on to the explicit construction of the algebra near $\mathcal{I}$, let us focus on the properties that the space and the brackets in \eqref{def:general_expression_algebra_I} shall satisfy (in addition to the $L_\infty$ relations) as a consequence of the restriction of the initial theory close to the null infinity regime.
\begin{enumerate}
    \item The graded vector space $\mathcal{X}^{\mathcal{I}}$ contains the gauge parameters, gauge fields and field equations near $\mathcal{I}$, where they admit an expansion in (non-negative) powers of $1/r$. Thus, its elements have the same form as \eqref{def:initial_graded_vector_space}, except now the sum \eqref{def:general_expansion_for_psi} does not include the terms with $n>0$, that are divergent at null infinity (where $r\rightarrow\infty$):
    \begin{equation}
        \label{def:space_at_I}
        \mathcal{X}^\mathcal{I}\subseteq\mathcal{X}^{n\leq0}\subset\mathcal{X}
    \end{equation}
    where we defined $\mathcal{X}^{n\leq0}\coloneqq\left\{\psi\in\mathcal{X} \ | \ \psi^{(n)}=0 \ \mathrm{for} \ n>0\right\}$. Note that we allow $\mathcal{X}^\mathcal{I}$ to be a subspace of $\mathcal{X}^{n\leq0}$, to account for the possibility that the following requirements in this list impose further restrictions on $\mathcal{X}^\mathcal{I}$ (as is indeed the case). Moreover, for reasons that will be clarified shortly, it is useful to consider the projection
    \begin{equation}
        \label{def:capital_Pi}
        \Pi\colon\mathcal{X}\rightarrow\mathcal{X}^\mathcal{I}
    \end{equation}
    from the initial space to the one encoding the physics at null infinity; this map is well-defined by virtue of \eqref{def:space_at_I}. The terms in the expansion \eqref{def:general_expansion_for_psi} that diverge at null infinity are all contained in the projection to the subspace complementary to $\mathcal{X}^\mathcal{I}$,
    \begin{equation}
        \label{def:bar Pi}
        \Pi^\complement:=1_\mathcal{X}-\Pi \; .
    \end{equation}
    Notice that $\Pi^\complement\Pi=\Pi\Pi^\complement=0$, or equivalently $\mathrm{Im}\Pi^\complement=\mathrm{Ker}\Pi$. In other words, $\mathcal{X}=\mathrm{Im}\Pi\oplus\mathrm{Im}\Pi^\complement$, i.e. for all $\psi\in\mathcal{X}$ it is always possible to write a unique decomposition of the form $\psi=\Pi(\psi)+\Pi^\complement(\psi)$.
    \item Since our goal is to transfer the structure of the algebra $(\mathcal{X},B_1,B_2)$ in the bulk to the new $L_\infty$ algebra close to null infinity, the brackets of the latter shall be related to the ones of the former. Thus, the most natural way to define $B_1^\mathcal{I}$ and $B_2^\mathcal{I}$ is to restrict the domain of $B_1$ and $B_2$ to $\mathcal{X}^\mathcal{I}$ and $\mathcal{X}^\mathcal{I}\otimes\mathcal{X}^\mathcal{I}$, respectively, and then project their image via the map $\Pi$, namely
    \begin{equation}
        \label{def:B_i^I}
        B_i^\mathcal{I}\coloneqq\Pi B_i|_{(\mathcal{X}^\mathcal{I})^{\otimes i}} \;,
    \end{equation}
    for $i=1,2$; in the above and in the rest of the paper we use the notation
    \begin{equation}
        B_i|_{(\mathcal{X}^\mathcal{I})^{\otimes i}}\colon(\mathcal{X}^\mathcal{I})^{\otimes i}\to \mathcal{X}
    \end{equation}
    to denote the restriction of the domain of $B_i$ to the subspace $(\mathcal{X}^\mathcal{I})^{\otimes i}\subset\mathcal{X}^{\otimes i}$. In this way, both the domain and the codomain of the brackets are the correct ones.
    \item Recall that the physics of the theory in the bulk is encoded in the cohomology \eqref{def:cohomology_of_B1} of the differential $B_1$. To transfer this information to the algebra close to null infinity, we promote the projection $\Pi$ to be a \textit{cochain map}, namely a morphism $\Pi\colon(\mathcal{X},B_1)\rightarrow(\mathcal{X}^\mathcal{I},B_1^\mathcal{I})$ of cochain complexes\footnote{We remark that, \textit{a priori}, $(\mathcal{X}^\mathcal{I},B_1^\mathcal{I})$ is not guaranteed to be a cochain complex, since we have not imposed the nilpotency of the differential, yet. However, $(B_1^\mathcal{I})^2$ vanishes automatically from \eqref{def:B_i^I} and \eqref{def:Pi_as_quasi_isomorphism}, as we prove in \eqref{pf:nilpotency_of_B1I}, meaning that the map $\Pi$ satisfying condition \eqref{def:Pi_as_quasi_isomorphism} is indeed a morphism of cochain complexes.} such that
    \begin{equation}
        \label{def:Pi_as_quasi_isomorphism}
        \Pi B_1=B_1^\mathcal{I}\Pi
    \end{equation}
    or, equivalently, such that the following diagram commutes:
    \begin{equation}
        \label{diagram:physical_requirement}
        \begin{tikzcd}[row sep=10mm]
            X_i\arrow{d}{\Pi}\arrow{r}{B_1}&X_{i+1}\arrow{d}{\Pi}\\
            X_i^\mathcal{I}\arrow{r}{B_1^\mathcal{I}}&X_{i+1}^\mathcal{I}
        \end{tikzcd} \;,
    \end{equation}
    i.e., if, starting from the space $X_{i}$, we can follow either of the two possible routes to $ X_{i+1}^\mathcal{I}$ and obtain the same result.
    As we will explain in details in \cref{sec:cochain}, a consequence of the above requirement is that $\Pi$ induces morphisms $H^i(\mathcal{X})\rightarrow H^i(\mathcal{X}^\mathcal{I})$ on the cohomology groups defined in \eqref{def:cohomology_of_B1}.
    
    A different perspective that further explains the necessity of imposing \eqref{def:Pi_as_quasi_isomorphism} arises from the following situation. As explained above, $\forall\psi\in\mathcal{X}$, when going close to null infinity, only $\Pi(\psi)$ survives while $\Pi^\complement(\psi)$ is disregarded because it encloses the divergent part. Now, consider the action of the bracket $B_1$ on $\psi$ and then its projection close to $\mathcal{I}$:
    \begin{equation}
    \label{eq:projected_b1_psi}
        \Pi B_1(\psi)=\Pi B_1(\Pi(\psi)+\Pi^\complement(\psi))\;.
    \end{equation}
    In general, $\Pi B_1\Pi^\complement(\psi)\neq0$, meaning that the divergent terms in $\psi$ can become finite after applying the differential in the bulk, thus contributing to the projected $B_1(\psi)$ close to null infinity. Promoting $\Pi$ to a cochain map is equivalent to asking that this can never happen. In other words, we want the second term in the RHS of \eqref{eq:projected_b1_psi} to vanish for all $\psi\in\mathcal{X}$, namely
    \begin{equation}
        \label{eq:condition_on_Im_barPi}
        \mathrm{Im}\Pi^\complement\subseteq\mathrm{Ker}(\Pi B_1) \ .
    \end{equation}
    By virtue of \eqref{def:B_i^I} and the definition of $\Pi^\complement$, this condition is indeed equivalent to \eqref{def:Pi_as_quasi_isomorphism}.
    \item Finally, in order for \eqref{def:general_expression_algebra_I} to be a homotopy Lie algebra, we impose the $L_\infty$ relations on its brackets.  Note that our definition of the map $B_1^\mathcal{I}$ together with condition \eqref{def:Pi_as_quasi_isomorphism} and equation \eqref{def:nilpotency_of_differential} already imply the nilpotency of the differential at null infinity,
    \begin{equation}
    \label{pf:nilpotency_of_B1I}
        (B_{1}^\mathcal{I})^2=B_1^\mathcal{I}\Pi B_1|_{\mathcal{X}^\mathcal{I}}=\Pi B_1B_1|_{\mathcal{X}^\mathcal{I}}=0 \ .
    \end{equation}
    In principle, one should allow for the appearance of higher (than two-) brackets and higher homotopy relations in the homotopy algebra that describes the theory close to null infinity. Indeed, it is not obvious that truncating the theory by selecting powers of $r$ should preserve its original algebraic structure. However, the constraints that will be imposed on the projector $\Pi$ will lead to a dgLa close to null infinity, as detailed at the end of \cref{subsec:Conditions on the power selection sets}. So, for simplicity and brevity, let us for now assume that the theory is encoded by a dgLa in this regime, and only impose the Leibniz rule of $B_1^\mathcal{I}$ with respect to the two-bracket and the (strict) Jacobi identity for $B_2^\mathcal{I}$,
    \begin{align}
        \label{eq:Leib_Jac_at_I}
        \mathrm{Leib}\left(B_1^\mathcal{I},B_2^\mathcal{I}\right)=0 \quad \mathrm{and} \quad \mathrm{Jac}\left(B_2^\mathcal{I}\right)=0 \ .
    \end{align}

\end{enumerate}
A triplet $(\mathcal{X}^\mathcal{I},B_1^\mathcal{I},B_2^\mathcal{I})$ satisfying all the above requirements successfully encodes self-dual Yang-Mills theory in the null infinity region. To find such an algebra (or such algebras, as for now it is not clear whether the solution of the constraints listed here is unique), note that the graded vector space close to null infinity can be implicitly defined as the image of the projection introduced in \eqref{def:capital_Pi},
\begin{equation}
    \mathcal{X}^\mathcal{I}\coloneqq\mathrm{Im}\Pi \ .
\end{equation}
Recalling equation \eqref{def:B_i^I}, we point out that both the graded vector space and the brackets in $(\mathcal{X}^\mathcal{I},B_1^\mathcal{I},B_2^\mathcal{I})$ can be written in terms of the initial $L_\infty$ algebra $(\mathcal{X},B_1,B_2)$ and the map $\Pi$. Hence, the problem of determining the homotopy Lie algebra close to null infinity can be reformulated as the task of finding a projection $\Pi$ satisfying the system
\begin{subnumcases}{\label{eq:master_system}}
\mathrm{Im}\Pi\subseteq\mathcal{X}^{n\leq0} \label{eq:master_system_EQ1} \\
        \mathrm{Ker}\Pi\subseteq\mathrm{Ker}(\Pi B_1) \label{eq:master_system_EQ2}\\
        \mathrm{Leib}\left(\Pi B_1|_{\mathrm{Im}\,\Pi},\Pi B_2|_{\mathrm{Im}\,\Pi\otimes\mathrm{Im}\,\Pi}\right)=0 \label{eq:master_system_EQ3} \\
        \mathrm{Jac}\left(\Pi B_2|_{\mathrm{Im}\,\Pi\otimes\mathrm{Im}\,\Pi}\right)=0 \;, \label{eq:master_system_EQ4}
\end{subnumcases}
that is just a reformulation of conditions \eqref{def:space_at_I}, \eqref{eq:condition_on_Im_barPi}\footnote{Recall that $\mathrm{Im}\,\Pi^\complement=\mathrm{Ker}\,\Pi$.} and \eqref{eq:Leib_Jac_at_I} using the definition of $\mathcal{X}^\mathcal{I}$ and $B^\mathcal{I}_i$ in terms of $\Pi$. \eqref{eq:master_system} plays a central role in our analysis and will be referenced frequently throughout the paper; thus, for convenience, we shall refer to it as the \textit{master system}. In what follows, we prove the existence of an infinite family of solutions $\{\Pi_k\}_{k\in\mathbb{Z}^{\leq0}}$ for the system above, by explicit construction.

\subsection{Solutions of the master system}
\label{subsec:construction of solutions Pik}
Equation \eqref{eq:master_system_EQ1} states that the action of the projector $\Pi$ on an element $\psi\in\mathcal{X}$ must set to zero all the terms with positive powers of $r$ in its expansion - see the definition of the space $\mathcal{X}^{n\leq0}$ below equation \eqref{def:space_at_I}. Based on this observation, one can conjecture that the solution of the first equation of the master system is a map $\pi_{\mathrm{trial}}$ that acts on $\psi\in\mathcal{X}$ as
\begin{equation}
    \pi_{\mathrm{trial}}(\psi)\coloneqq\sum_{n\leq0}r^n\psi^{(n)}(y^a) \ .
\end{equation}
However, this truncation of the sum \eqref{def:general_expansion_for_psi} does not preserve, in general, the anti-self-duality of the elements in $X_1$ (recall that $X_1$ is the space of field equations, whose elements are anti-self-dual two-forms), namely
\begin{equation}
    \label{eq:subspace_failure}
    (1+*)\pi_{\mathrm{trial}}(E)_{\mu\nu}\neq0
\end{equation}
for $E_{\mu\nu}\in X_1$. The underlying reason is the presence of different powers of $r$ in different components of the basis element $\underline{e}^1_{\mu\nu}$ in \eqref{def:basis_for_X1}. Thus, truncating the $r$-expansion of $E_{\mu\nu}$ uniformly in all its components produces a tensor that does not have the general form \eqref{eq:E_full_generl_form}. Formally stated, $\pi_{\mathrm{trial}}$ cannot be a solution for $\Pi$ in \eqref{eq:master_system_EQ1} because $\mathrm{Im}\pi\not\subseteq\mathcal{X}^{n\leq0}$\footnote{Remember that, by definition, $\mathcal{X}^{n\leq0}$ is a subspace of $\mathcal{X}$. In particular, the tensors in $X_1^{n\leq0}\subset X_1$ are anti-self-dual.}. This consideration implies that the map solving the master system should act non-uniformly on the components of the vectors and tensors of $\mathcal{X}$. To find it, we must first construct a suitable truncation function that operates on individual components, which we do in the following.

\subsubsection{Projection \texorpdfstring{$\mathcal{P}_N$}{PN}}
We now construct a map that allows us to select specific terms in the $r$-expansion of each component of the objects in the bulk. Note that the components of the vectors/tensors in $\mathcal{X}$ are elements of the space $C^\infty_r(\mathcal{M})$ defined in \eqref{def:C^infty_r}, where $\mathcal{M}$ is the spacetime with Bondi coordinates $x^\mu=(r,y^\alpha)=(r,u,z,\bar z)$. On this space, given a set\footnote{Note that we do not require $N$ to be a discrete interval in $\mathbb{Z}$ (i.e. a set of consecutive integers). In fact, most of the proofs in \cref{sec:appendix_proofs} hold for arbitrary sets. However, we ultimately reduce to the case where all the sets consist of consecutive integers without gaps in order to extract the physical solution.} $N\subseteq\mathbb{Z}$, we define an endomorphism $\mathcal{P}_N$ that acts on functions $f\in C_r^\infty(\mathcal{M})$ as
\begin{equation}
    \label{def:projector_P_N}
    \mathcal{P}_{N}f(r,y^\alpha)\coloneqq\sum_{n\in N}r^{n}f^{(n)}(y^\alpha)
\end{equation}
To give a concrete example, for $N=\{-20,-2,0,3\}$ we have
\begin{equation}
    \mathcal{P}_Nf=r^3f^{(3)}+f^{(0)}+\frac{f^{(-2)}}{r^2}+\frac{f^{(-20)}}{r^{20}} \ .
\end{equation}
In the above equation and in what follows, we drop the explicit coordinate dependence of $f$ and $f^{(n)}$ for the sake of readability.

The map \eqref{def:projector_P_N} satisfies the properties \eqref{eq:P property} listed in \cref{sec:appendix_proofs}; in particular, the first of these relations implies that $\mathcal{P}_N$ is a projection $\forall N$, making it an ideal candidate to serve as a building block for the solution of the master system.

\subsubsection{Projection \texorpdfstring{$\pi_\mathcal{N}$}{piN}}
\label{subsec:projection piN}
Now that we have developed the machinery necessary to manipulate the $r$-expansion of the single components of the elements in $\mathcal{X}$, we can construct a function that truncates them non-uniformly, employing the projection defined above.

To do so, we need to specify how $\mathcal{P}$ acts on each component of a generic element of the space $\mathcal{X}$. In other words, we shall define a collection of sets
\begin{equation}
    \label{def:power selection sets}
    \mathcal{N}:=\{N_\Lambda,N_{A_\mu},N_{E_{\mu\nu}}\subseteq\mathbb{Z}\}
\end{equation}
that contains all the information about how to select the powers in the $r$-expansion of the components of $\psi$. For this reason, we refer to the sets in $\mathcal{N}$ as the \textit{power selection sets}. For example, the $u$ component of all elements in $X_0$ is truncated using $\mathcal{P}_{N_{A_u}}$, where $N_{A_u}\in\mathcal{N}$ is a set of integer numbers. This means that we truncate all the gauge parameters in $X_{-1}$ using $\mathcal{P}_{N_\Lambda}$\footnote{That is to say, given two gauge parameters $\Lambda_1,\Lambda_2$, we truncate them in the same way: $\mathcal{P}_{N_\Lambda}\Lambda_1$ and $\mathcal{P}_{N_\Lambda}\Lambda_2$, respectively.}, we truncate the $r$ component of all the gauge fields in $X_0$ using $\mathcal{P}_{N_{A_r}}$, and so on\footnote{The cardinality of \eqref{def:power selection sets} is $|\mathcal{N}|=1+4+16=21$.}. 

The next step is to introduce a function that simultaneously applies $\mathcal{P}$ to all the components of a given $\psi\in\mathcal{X}$, using the power selection sets in \eqref{def:power selection sets}. We observe that it is not guaranteed that after truncating the components of an element $E_{\mu\nu}\in X_1$, we obtain an anti-self-dual tensor, as explained below equation \eqref{eq:subspace_failure}. In other words, the result of the projection on an arbitrary equation of motion in $X_1$ does not belong, in general, to (a subspace of) $X_1$. For this reason, we start by introducing the following graded vector space
\begin{equation}
    \mathcal{C}:=\bigoplus_{i=-1}^1C_i=\overbrace{C_r^\infty(\mathcal{M})}^{C_{-1}}\oplus \overbrace{C_r^\infty(\mathcal{M})^4}^{C_0}\oplus\big(\overbrace{C_r^\infty(\mathcal{M})^4\otimes C_r^\infty(\mathcal{M})^4}^{C_1}\big) \;,
\end{equation}
and we define a morphism\footnote{A morphism between two graded vector spaces $\mathcal{V}=\bigoplus_{i\in I}V_i$ and $\mathcal{W}=\bigoplus_{i\in I}W_i$ is a linear map $f\colon\mathcal{V}\rightarrow\mathcal{W}$ of degree zero, i.e. $f(V_i)\subseteq W_i \ \forall i\in I$.} $\pi_\mathcal{N}\colon\mathcal{X}\rightarrow\mathcal{C}$ given by
\begin{equation}
    \label{def:better_truncation}
    \pi_\mathcal{N}(\psi)\coloneqq\begin{cases}
        \mathcal{P}_{N_\Lambda}(\Lambda) &\mathrm{if} \ \psi\in X_{-1} \\
        \mathcal{P}_{N_{A_\mu}}(A_\mu) &\mathrm{if} \ \psi\in X_0 \\
        \mathcal{P}_{N_{E_{\mu\nu}}}(E_{\mu\nu}) &\mathrm{if} \ \psi\in X_1
    \end{cases} \;.
\end{equation}
We remark that $\pi_\mathcal{N}$ is well defined since it applies the projection $\mathcal{P}_N$ to each component of the objects in $\mathcal{X}$, and those components are elements of $C^\infty_r(\mathcal{M})$. As an example, consider the following choice for the power selection sets for the gauge fields:
\begin{equation}
    \begin{aligned}
        N_{A_r}=\{-3,-2\}, \quad N_{A_u}=N_{A_{\bar z}}=\{-2,-1,0\}, \quad N_{A_z}=\{-1,0\},
    \end{aligned}
\end{equation}
then the image of an arbitrary gauge field $A_\mu\in X_0$ through $\pi_\mathcal{N}$ is
\begin{equation}
    (\pi_\mathcal{N}A)_\mu=\begin{pmatrix}
        \mathcal{P}_{N_{A_r}}A_r \\
        \mathcal{P}_{N_{A_u}}A_u \\
        \mathcal{P}_{N_{A_z}}A_z \\
        \mathcal{P}_{N_{A_{\bar z}}}A_{\bar z}
    \end{pmatrix}=\begin{pmatrix}
        \frac{A_r^{(-2)}}{r^2}+\frac{A_r^{(-3)}}{r^3} \\
        A_u^{(0)}+\frac{A_u^{(-1)}}{r}+\frac{A_u^{(-2)}}{r^2} \\
        A_z^{(0)}+\frac{A_z^{(-1)}}{r} \\
        A_{\bar z}^{(0)}+\frac{A_{\bar z}^{(-1)}}{r}+\frac{A_{\bar z}^{(-2)}}{r^2}
    \end{pmatrix}\in C_0 \ .
\end{equation}
Observe that $\pi_\mathcal{N}$ is not a projection (as $\mathcal{C}\not\subseteq\mathcal{X}$, since $C_1\not\subseteq X_1$), despite being defined in terms of the projection $\mathcal{P}$. Notice that, in contrast to pure Yang-Mills, the $z$ and $\bar z$ components have different truncations. This asymmetry arises from the self-duality condition.

At this point, it is finally possible to solve the first equation of the master system: $\Pi=\pi_\mathcal{N}$ is a solution for \eqref{eq:master_system_EQ1} if the power selection sets in $\mathcal{N}$ satisfy
\begin{subnumcases}{\label{def:pi_N_contraints_on_N}}
    N_I\subseteq\mathbb{Z}^{\leq0} \quad \forall I\in\{\Lambda,A_\mu,E_{\mu\nu}\} \label{def:pi_N_contraints_on_N 1} \\
    N_{E_{\mu\nu}}=N_{E_{\nu\mu}} \label{def:pi_N_contraints_on_N 2}\\
    N_{E_{z\bar z}}=\mathcal{T}_2N_{E_{ru}} \label{def:pi_N_contraints_on_N 3}
\end{subnumcases}
where $\mathbb{Z}^{\leq0}$ is the set of non-positive integers and $\mathcal{T}$ is the translation of sets, i.e. $\forall s\in\mathbb{Z}$ and $\forall U\subseteq\mathbb{Z}$
\begin{equation}
    \label{def:translation}
    \mathcal{T}_sU\coloneqq\{u+s, \ u\in U\}\subseteq\mathbb{Z} \ .
\end{equation}
The first condition of the system above ensures that no positive power of $r$ appears in the expansion of an element in $\mathrm{Im}\pi_{\mathcal{N}}$. Moreover, having in mind the general form \eqref{eq:E_full_generl_form} of the equations of motion, it is easy to see that the second condition above ensures that $\pi_\mathcal{N}(E)_{\mu\nu}$ are skew-symmetric tensors $\forall E_{\mu\nu}\in X_1$, while the third guarantees their anti-self-duality. Therefore, the morphism $\pi_\mathcal{N}$ with $\mathcal{N}$ satisfying \eqref{def:pi_N_contraints_on_N} is a projection\footnote{This justifies the title of the current subsection.} and its image is a subspace of $\mathcal{X}^{n\leq0}$, so it is a solution (in fact, the most general one) of the master system's first equation. 

The above considerations, combined with the fact that $E_{rz}$ always vanishes, indicate that it is sufficient to specify the sets $N_{E_{ru}}$, $N_{E_{r\bar z}}$ and $N_{E_{uz}}$ to completely define the action of $\pi_\mathcal{N}$ on an arbitrary equation of motion in $X_1$, which can thus be written as a linear combination of \eqref{def:basis_for_X1}, as required.

\subsubsection{An infinite family of \texorpdfstring{$L_\infty$}{L infty} algebras}
\label{subsec:an infinite family of Linf algebras}
We are now ready to present the solution of the master system \eqref{eq:master_system}, which relies on the formalism built above. The complete proof is detailed in \cref{sec:appendix_proofs}, where we show that the system is indeterminate and its physically relevant solution is a family of projections $\{\Pi_k\}_{k\in\mathbb{Z}^{\leq0}}$ of the form 
\begin{equation}
\label{res:master system solution}
    \Pi_k=\pi_{\mathcal{N}_k}\;, \quad \mathcal{N}_k=\{N^k_\Lambda,N^k_{A_\mu},N^k_{E_{\mu\nu}}\subseteq\mathbb{Z}^{\leq0}\}
\end{equation}
with
\begin{equation}
\label{res:master system solution sets}
    N^k_I=\begin{cases}
        \{k,\dots,0\} &\mathrm{if} \ I=A_z,E_{uz} \\
        \{k-1,\dots,0\} &\mathrm{if} \ I=\Lambda,A_u,A_{\bar z} \\
        \{k-2,\dots,-2\} &\mathrm{if} \ I=A_r,E_{ru},E_{r\bar z}
    \end{cases}\;, \quad \mathrm{for \ }k\leq0\;.
\end{equation}
Using these projections, we build an infinite family of $L_\infty$ algebras $\{(\mathcal{X}^{(k)},B_1^{(k)},B_2^{(k)})\}_{k\in\mathbb{Z}^{\leq0}}$ where
\begin{equation}
    \label{def:space and brackets for slices}
    \mathcal{X}^{(k)}:=\mathrm{Im}\Pi_k, \qquad B_i^{(k)}:=\Pi_k\circ B_i|_{(\mathrm{Im}\,\Pi_k)^{\otimes i}} \;.
\end{equation}
This will determine a series of natural ``slices" $\mathfrak{S}_k$, displayed in \cref{fig:generic slice}, on which we can define our theory and construct the kinematic algebra. Remarkably, the upper boundary, which gives the so-called fall-off of the fields, is identical for all slices and independent of $k$, and it coincides exactly with the physical fall-off conditions presented in \cite{Strominger:2013lka,Strominger:2017zoo} and reviewed in \autoref{Review of physical fall-off at null infinity}. Such fall-offs therefore arise from general properties of the master system (see \cref{fig:fall off} for an explanation of how each condition of the system prohibits a specific power in the fall-off). The lower boundary of the slice is simply the choice of precision in the expansion in $1/r$. In particular, we are interested in the fall-offs of the gauge fields. Below, we define two objects that encode the approximation (truncation from below) and the fall-offs (truncation from above) of each slice $\mathfrak{S}_k$:
\begin{equation}
    \label{def:truncation and falloff}
    \begin{aligned}
        t_\mu^k&\coloneqq\min N_{A_\mu}^{(k)}=(k-2,k-1,k,k-1) \qquad \text{(truncation from below)} \\
        f_\mu&\coloneqq\max N_{A_\mu}^{(k)}=(-2,0,0,0) \qquad\qquad\qquad\ \text{(fall-off)}\;.
    \end{aligned}
\end{equation}

Finally, we remark that the homotopy algebra construction automatically ensures that the truncated gauge fields entering our slices are exactly the ones needed to write the truncated equations of motion. This emerges directly when solving the master system \eqref{eq:master_system} as shown in \cref{sec:appendix_proofs}, i.e. we did not need to impose it as an additional condition.   

To wrap up this paragraph, we provide, as an illustrative example, the explicit form of the power selection sets that define the slice $\mathfrak{S}_0$:
\begin{equation}
    \label{res:solution N0 for Pi}
    N^0_I=\begin{cases}
        \{0\} &\mathrm{if} \ 
        I=A_z,E_{uz}\\
        \{-1,0\} &\mathrm{if} \ I=\Lambda,A_u,A_{\bar z} \\
        \{-2\} &\mathrm{if} \ I=A_r,E_{ru},E_{r\bar z}
    \end{cases}.
\end{equation}
Thus, the gauge parameters, gauge fields and e.o.m. in $\mathcal{X}^{(0)}\coloneqq\mathrm{Im}\Pi_0$ are, respectively,
\begin{equation}
    \label{eq:slice 0 elements}
    \Lambda=\Lambda^{(0)}+\frac{\Lambda^{(-1)}}{r}, \quad
    A_\mu=\begin{pmatrix}
        \frac{A_r^{(-2)}}{r^2} \\
        A_u^{(0)}+\frac{A_u^{(-1)}}{r} \\
        A_z^{(0)} \\
        A_{\bar z}^{(0)}+\frac{A_{\bar z}^{(-1)}}{r}
    \end{pmatrix}, \quad E_{\mu\nu}=\begin{pmatrix}
        0 & \frac{E_{ru}^{(-2)}}{r^2} & 0 & \frac{E_{r\bar z}^{(-2)}}{r^2} \\
        -\frac{E_{ru}^{(-2)}}{r^2} & 0 & E_{uz}^{(0)} & 0 \\
        0 & -E_{uz}^{(0)} & 0 & E_{ru}^{(-2)} \\
        -\frac{E_{r\bar z}^{(-2)}}{r^2} & 0 & -E_{ru}^{(-2)} & 0 \\
    \end{pmatrix}\;.
\end{equation}

\begin{figure}[h!]
    \centering
    \begin{overpic}[width=0.9\linewidth]{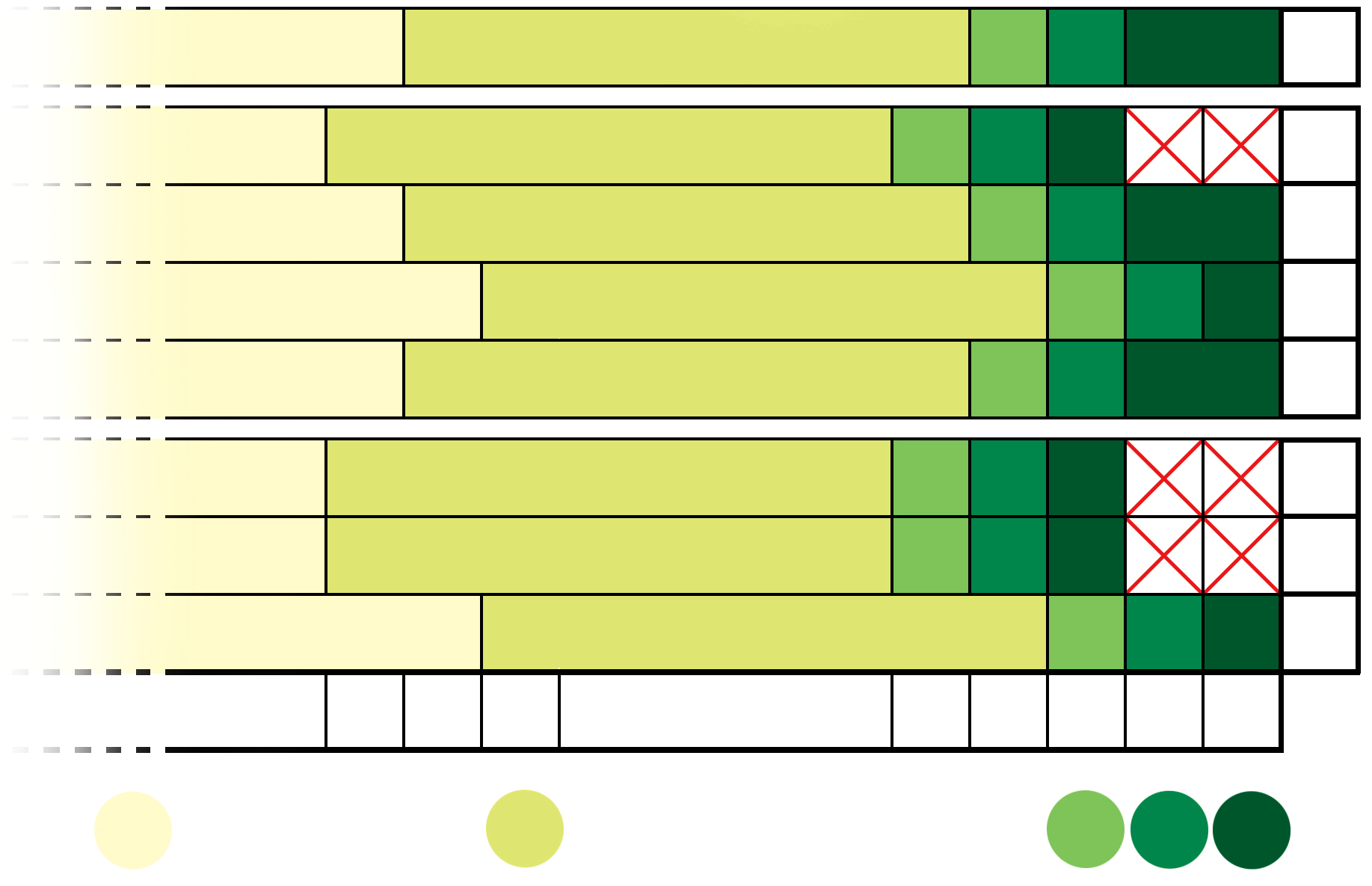} 
        \put(94,28.7){$E_{ru}$}
        \put(94,23.2){$E_{r\bar z}$}
        \put(94,17.6){$E_{uz}$}
        \put(94.7,53){$A_r$}
        \put(94.7,47.3){$A_u$}
        \put(94.7,41.7){$A_z$}
        \put(94.7,36){$A_{\bar z}$}
        \put(95.2,60){$\Lambda$}
        \put(89.6,3.3){\white$\mathfrak{S}_0$\black}
        \put(82.7,3.3){\white$\mathfrak{S}_{-1}$\black}
        \put(76.5,3.3){\white$\mathfrak{S}_{-2}$\black}
        \put(36.5,3.3){$\mathfrak{S}_k$}
        \put(8,3.3){$\mathfrak{S}_{-\infty}$}
        \put(24,11.8){$k$}
        \put(25.8,11.8){\footnotesize$-2$}
        \put(30,11.8){$k$}
        \put(31.8,11.8){\footnotesize$-1$}
        \put(37,11.8){$k$}
        \put(66,11.8){$-4$}
        \put(71.5,11.8){$-3$}
        \put(77,11.8){$-2$}
        \put(83,11.8){$-1$}
        \put(89.7,11.8){$0$}
        \put(51,11.8){$\dots$}
        \put(3,11.8){$-\infty$}
        \put(9,11.8){$\longleftarrow$}
    \end{overpic}
    \caption{The drawing illustrates the slices $\mathfrak{S}_k$ (where $k\in\mathbb{Z}^{\leq0}$) obtained by projecting the initial graded vector space $\mathcal{X}$ via the maps $\Pi_k$ in \eqref{res:master system solution}, that are the solutions of the master system \eqref{eq:master_system}. Each slice is represented in a different shade of green. The lighter slices always include the darker ones, in other words $\mathfrak{S}_{k-1}\subset\mathfrak{S}_k$ for all $k$. Moreover, all slices share the same maxima and therefore the same fall-offs.}
    \label{fig:generic slice}
\end{figure}

\begin{figure}[h!]
    \centering
    \vspace{1cm}
    \begin{overpic}[width=0.9\linewidth]{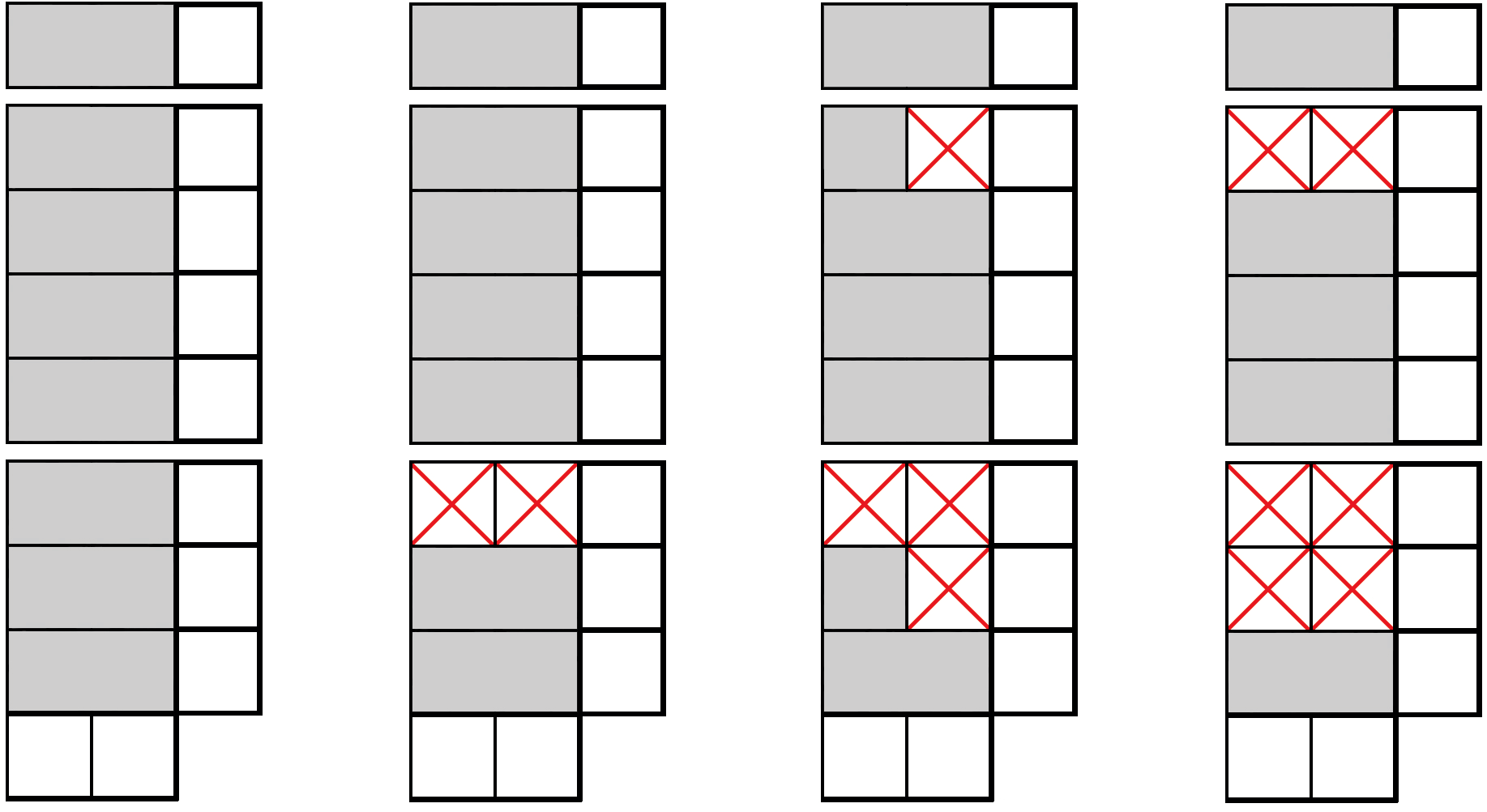} 
        \put(13.7,50.7){$\Lambda$}
        \put(13.1,44){$A_r$}
        \put(13.1,38.4){$A_u$}
        \put(13.1,32.8){$A_z$}
        \put(13.1,27){$A_{\bar z}$}
        \put(12.4,20){$E_{ru}$}
        \put(12.4,14.5){$E_{r\bar z}$}
        \put(12.4,8.8){$E_{uz}$}
        \put(40.8,50.7){$\Lambda$}
        \put(40.2,44){$A_r$}
        \put(40.2,38.4){$A_u$}
        \put(40.2,32.8){$A_z$}
        \put(40.2,27){$A_{\bar z}$}
        \put(39.5,20){$E_{ru}$}
        \put(39.5,14.5){$E_{r\bar z}$}
        \put(39.5,8.8){$E_{uz}$}
        \put(68.4,50.7){$\Lambda$}
        \put(67.8,44){$A_r$}
        \put(67.8,38.4){$A_u$}
        \put(67.8,32.8){$A_z$}
        \put(67.8,27){$A_{\bar z}$}
        \put(67.1,20){$E_{ru}$}
        \put(67.1,14.5){$E_{r\bar z}$}
        \put(67.1,8.8){$E_{uz}$}
        \put(95.6,50.7){$\Lambda$}
        \put(95,44){$A_r$}
        \put(95,38.4){$A_u$}
        \put(95,32.8){$A_z$}
        \put(95,27){$A_{\bar z}$}
        \put(94.3,20){$E_{ru}$}
        \put(94.3,14.5){$E_{r\bar z}$}
        \put(94.3,8.8){$E_{uz}$}
        \put(1.6,3){$-1$}
        \put(8.3,3){$0$}
        \put(28.7,3){$-1$}
        \put(35.4,3){$0$}
        \put(56.3,3){$-1$}
        \put(63,3){$0$}
        \put(83.5,3){$-1$}
        \put(90.2,3){$0$}
        \put(20,29.5){\huge{$\Rightarrow$}}
        \put(47.3,29.5){\huge{$\Rightarrow$}}
        \put(74.8,29.5){\huge{$\Rightarrow$}}
        \put(18.2,33.5){\eqref{eq:master_system_EQ1}}
        \put(45.5,33.5){\eqref{eq:master_system_EQ2}}
        \put(73,33.5){\eqref{eq:master_system_EQ3}}
        \put(20.5,60){No}
        \put(15.5,56.7){divergences}
        \put(42,60){Cochain map}
        \put(44.2,56.7){condition}
        \put(72.8,60){Leibniz}
        \put(74.7,56.7){rule}
        \begin{tikzpicture}[overlay, remember picture]
            \draw[->, black] (0.202\linewidth, 0.335\linewidth) -- (0.202\linewidth, 0.495\linewidth);
            \draw[->, black] (0.448\linewidth, 0.335\linewidth) -- (0.448\linewidth, 0.495\linewidth);
            \draw[->, black] (0.695\linewidth, 0.335\linewidth) -- (0.695\linewidth, 0.495\linewidth);
        \end{tikzpicture}
    \end{overpic}
    \caption{The fall-offs follow directly from the general properties of the master system: each condition in the system restricts specific powers in the fall-off, and the figure illustrates how these conditions determine the excluded powers, which are represented by red crosses.}
    \label{fig:fall off}
\end{figure}
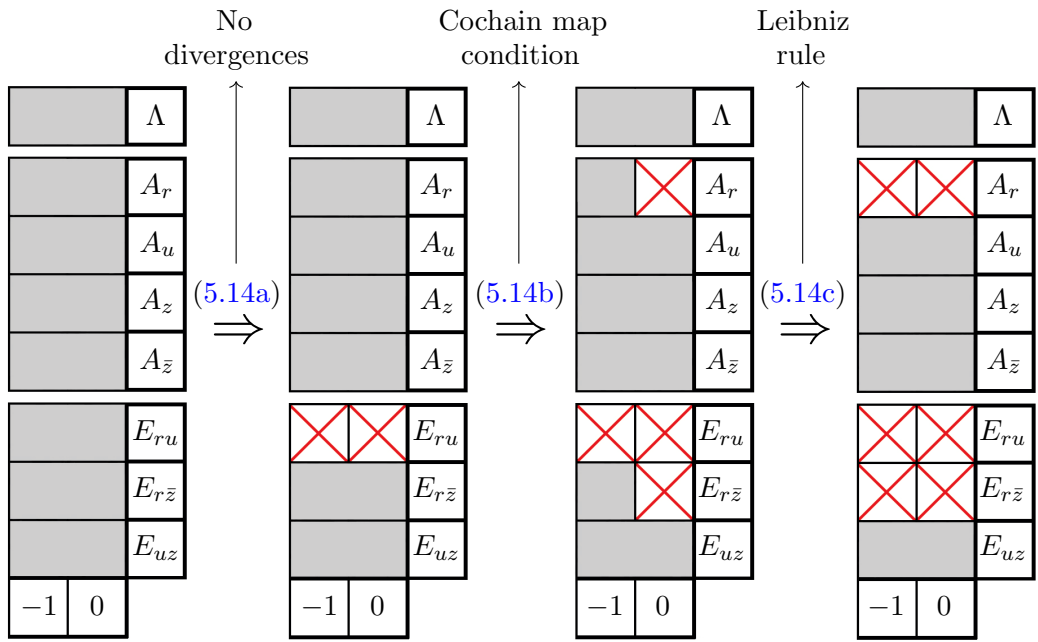

\newpage

\section{Strict kinematic algebras}
\label{sec:Strict kinematic algebras}
Let us take stock of the situation. Thanks to the method explained in \cref{sec:SDYM at null infinity via slice truncation}, we have found a way to define an infinite family of strict $L_\infty$ algebras (in fact, differential graded Lie algebras) that describe self-dual Yang Mills near null infinity. Each algebra $(\mathcal{X}^{(k)},B_1^{(k)},B_2^{(k)})$ is labelled by a non-positive integer $k$, that corresponds to the approximation we adopt when truncating ``from below" the $r$-expansion of the elements in $\mathcal{X}$, the graded vector space of the initial $L_\infty$ algebra encoding SDYM in the bulk.

\subsection{Kinematic algebras on slices}
\label{subsec:Kinematic algebras on slices}
At this point, we may proceed with the construction of the kinematic algebra on each slice $\mathfrak{S}_k$. We follow the procedure outlined in sections \ref{sub:Cinf alg}, \ref{sub:kin alg gen} and employed in \cref{sec:color stripping and kinematic algebra}, with the sole difference that we now use the projection $\Pi_{k}$ defined above to ensure that the kinematic algebra closes on the slice. First, by color-stripping, we obtain a differential graded commutative algebra $(\mathcal{K}^{(k)},m_1^{(k)},m_2^{(k)})$, where $m_i^{(k)}=\Pi_k m_i$ for $i=1,2$\footnote{\label{ftn:projection on kin space}Recall that $\Pi_k\colon\mathcal{X}\rightarrow\mathcal{X}^{(k)}$, where $\mathcal{X}=\mathcal{K}\otimes\mathfrak{g}$ and $\mathcal{X}^{(k)}=\mathcal{K}^{(k)}\otimes\mathfrak{g}$. Let $\Pi_k^\mathrm{kin}\colon\mathcal{K}\rightarrow\mathcal{K}^{(k)}$ be the projection that acts directly on the color-stripped graded vector space. Since the color-stripping procedure does not affect the $r$-expansion, we can identify $\Pi_k^\mathrm{kin}$ with $\Pi_k$. For this reason, by an abuse of notation, we suppress the superscript ``$\mathrm{kin}$" whenever the domain of the projection is clear from the context.}. Next, we equip the graded vector space $\mathcal{K}^{(k)}$ with the map $b^{(k)}\coloneqq\Pi_k b$, which is the projection on the slice $k$ of the differential $b$ introduced in \eqref{def:differential b in components}. As we show in the appendix, $(\mathcal{K}^{(k)},b^{(k)})$ is a chain complex (see \cref{rmk:Kk bk chain complex}). Finally, we define the box operator and the kinematic bracket, respectively, as
\begin{equation}
    \label{def: projected box and kin bracket}
    \Box^{(k)}\coloneqq[b^{(k)},m_1^{(k)}]\;, \quad b_2^{(k)}:=[b^{(k)}, m_2^{(k)}] \;,
\end{equation}
where the commutator in the second equation acts as in \eqref{def:b2}. As expected, the two-product $m_2^{(k)}$ and the kinematic bracket satisfy a generalised version of the Poisson compatibility relation. The failure to satisfy this relation strictly is governed by a trilinear map $\theta_3^{(k)}$. Putting all these ingredients together, we obtain a $\mathrm{BV}_\infty^{\Box^{(k)}}$ algebra
\begin{equation}
    \label{def:kinematic algebra for slice}
    \mathfrak{Kin}^{(k)}\coloneqq(\mathcal{K}^{(k)}, m_{1}^{(k)}, b^{(k)}, \Box^{(k)}, m_{2}^{(k)}, b_{2}^{(k)}, \theta_{3}^{(k)}) \;,
\end{equation}
which is the kinematic algebra on the slice $\mathfrak{S}_k$. In \cref{sec:appendix other proofs} we show that the box operator of this algebra is simply the projection on the slice of the d'Alembert operator, $\Box^{(k)}=\Pi_k\Box$ (see \cref{rmk:box k}). In particular, this implies that $\mathfrak{Kin}^{(k)}$ is a $\mathrm{BV}_\infty^{\Pi_k\Box}$ algebra. Similarly, one can show that $b_2^{(k)}=\Pi_kb_2$ and $\theta_3^{(k)}=\Pi_k\theta_3$. In light of these observations, a simple relation between the kinematic algebra on the slice and the one in the bulk can be written:
\begin{equation}
\label{projection_kin_alg}
    \mathfrak{Kin}^{(k)}=\Pi_k\circ\mathfrak{Kin}\;,
\end{equation}
cf. equations \eqref{def:kinematic algebra bulk} and \eqref{def:kinematic algebra for slice}. By this formal expression, we mean that every element (whether a space or a map) that defines $\mathfrak{Kin}^{(k)}$ is obtained by applying the projection $\Pi_k$ to the corresponding element of $\mathfrak{Kin}$.

As explained at the end of \cref{sub:kin alg gen}, the kinematic algebra contains a strict Lie algebra as a subsector iff it is itself strict. In other words, the kinematic bracket $b_2^{(k)}$ satisfies the strict Jacobi identity iff the ternary bracket $\theta_3^{(k)}$ vanishes. In order to study when this occurs, it is useful find the explicit expression for $\theta_3^{(k)}$. This is achieved by applying the projection $\Pi_k$ to equation \eqref{theta3}, and by employing the identities derived in the appendices\footnote{In particular, see \cref{lemma:properties of projector P} and equation \eqref{eq:projector on a product}.}. The components of the trilinear map applied to three gauge fields $\mathcal{A}_1,\mathcal{A}_2,\mathcal{A}_3\in K_1^{(k)}$ are given by
\begin{equation}
    \label{eq:theta3kAAA slice}
    \begin{aligned}
        \theta_3^{(k)}(\cA_{1},\cA_{2},\cA_{3})_r&=\sum_{\substack{m,n,l\in\{k-2,\dots,-2\} \\ p\in\{k-1,\dots,0\}}}6\delta_{m,n+l+p}r^m\mathcal{A}_{[1\,r}^{(n)}\mathcal{A}_{2\,z}^{(l+2)}\mathcal{A}_{3]\,\bar z}^{(p)} \\
        \theta_3^{(k)}(\cA_{1},\cA_{2},\cA_{3})_u&=-\sum_{\substack{m,n_1,n_3\in\{k-1,\dots,0\}\\ n_2\in\{k-2,\dots,-2\}}} 6\delta_{m,n_1+n_2+n_3}r^m\mathcal{A}_{[1\,u}^{(n_1)}\mathcal{A}_{2\,z}^{(n_2+2)}\mathcal{A}_{3]\,\bar z}^{(n_3)} \\
        \theta_3^{(k)}(\cA_{1},\cA_{2},\cA_{3})_z&=-\sum_{\substack{m,n_3\in\{k,\dots,0\} \\ n_1\in\{k-2,\dots,-2\}\\ n_2\in\{k-1,\dots,0\}}}6\delta_{m,n_1+n_2+n_3}r^m\mathcal{A}_{[1\,r}^{(n_1)}\mathcal{A}_{2\,u}^{(n_2)}\mathcal{A}_{3]\,z}^{(n_3)} \\
        \theta_3^{(k)}(\cA_{1},\cA_{2},\cA_{3})_{\bar z}&=\sum_{\substack{m,n_2,n_3\in\{k-1,\dots,0\} \\ n_1\in\{k-2,\dots,-2\}}}6\delta_{m,n_1+n_2+n_3}r^m\mathcal{A}_{[1\,r}^{(n_1)}\mathcal{A}_{2\,u}^{(n_2)}\mathcal{A}_{3]\,\bar z}^{(n_3)} \;,
    \end{aligned}
\end{equation}
while the components of $\theta_3^{(k)}$ on an equation of motion $\mathcal{E}\in K_2^{(k)}$ and two gauge fields are
\begin{equation}
    \label{eq:theta3EAA slice}
    \begin{aligned}
        \theta_3^{(k)}(\mathcal{E},\cA_{1},\cA_{2})_{ru}&=\sum_{\substack{m,n_1,n_2\in\{k-2,\dots,-2\}\\ n_3\in\{k-1,\dots,0\}}}2r^m\delta_{m,n_1+n_2+n_3}\Big(\mathcal{E}_{uz}^{(n_1+2)}\mathcal{A}_{[1\,r}^{(n_2)}\mathcal{A}_{2]\,\bar z}^{(n_3)}-\mathcal{E}_{r\bar z}^{(n_1)}\mathcal{A}_{[1\,u}^{(n_3)}\mathcal{A}_{2]\,z}^{(n_2+2)}\Big) \\
        \theta_3^{(k)}(\mathcal{E},\cA_{1},\cA_{2})_{r\bar z}&=\sum_{\substack{m,n_1,n_2\in\{k-2,\dots,-2\}\\ n_3\in\{k-1,\dots,0\}}}2r^m\delta_{m,n_1+n_2+n_3}\Big[\mathcal{E}_{r\bar z}^{(n_1)}\left(\mathcal{A}_{[1\,r}^{(n_2)}\mathcal{A}_{2]\,u}^{(n_3)}+\mathcal{A}_{[1\,z}^{(n_2+2)}\mathcal{A}_{2]\,\bar z}^{(n_3)}\right) \\
        &\qquad\qquad\qquad\qquad\qquad\qquad\qquad\qquad\ \raisebox{2ex}[0pt][0pt]{$-2\mathcal{E}_{ru}^{(n_1)}\mathcal{A}_{[1\,r}^{(n_2)}\mathcal{A}_{2]\,\bar z}^{(n_3)}\Big]$} \\
        \theta_3^{(k)}(\mathcal{E},\cA_{1},\cA_{2})_{uz}&=-\sum_{\substack{m,n_1\in\{k,\dots,0\}\\ n_2\in\{k-2,\dots,-2\}\\ n_3\in\{k-1,\dots,0\}}}2r^m\delta_{m,n_1+n_2+n_3}\Big[\mathcal{E}_{uz}^{(n_1)}\left(\mathcal{A}_{[1\,r}^{(n_2)}\mathcal{A}_{2]\,u}^{(n_3)}+\mathcal{A}_{[1\,z}^{(n_2+2)}\mathcal{A}_{2]\,\bar z}^{(n_3)}\right) \\
        &\qquad\qquad\qquad\qquad\qquad\qquad\qquad\quad\ \raisebox{4ex}[0pt][0pt]{$-2\mathcal{E}_{ru}^{(n_2)}\mathcal{A}_{[1\,u}^{(n_3)}\mathcal{A}_{2]\,z}^{(n_1)}\Big]\;.$}
    \end{aligned}
\end{equation}
Inspection of the above equations reveals that $\theta_3^{(k)}$ vanishes for all $\mathcal{A},\mathcal{E}$ if $k=0$, meaning that the kinematic algebra on the thinnest slice, i.e. $\mathfrak{Kin}^{(0)}$, is strict. However, this case is trivial and uninteresting, as the kinematic bracket between two fields vanishes upon imposing $A_u^{(0)}=0$, a standard gauge choice\footnote{The kinematic bracket between two gauge fields on the slice $k=0$ is $b_2^{(0)}(\mathcal{A}_1,\mathcal{A}_2)=\Pi_0b_2(\mathcal{A}_1,\mathcal{A}_2)$, where $\mathcal{A}_1,\mathcal{A}_2\in K_1^{(0)}$ are the color-stripped version of the fields given in \eqref{eq:slice 0 elements}. By explicit computation, one can check that the only non-zero component is $b_2^{(0)}(\mathcal{A}_1,\mathcal{A}_2)_{\bar z}=-4r^{-1}\mathcal{A}_{[1\,u}^{(0)}\mathcal{A}_{2]\,\bar z}^{(0)}$, which vanishes upon setting $A_u^{(0)}=0$.}.

In general, the equations above show that the kinematic algebra $\mathfrak{Kin}^{(k)}$ on a slice $\mathfrak{S}_k$ with $k$ non-vanishing is not strict. Since our goal is to construct kinematic brackets that strictly satisfies the Jacobi identity, we must seek alternative means of strictifying the $\mathrm{BV}_\infty^{\Pi_k\Box}$ algebra on $\mathfrak{S}_{k\neq0}$. Note that one could attempt to carefully impose relations among the components of the gauge fields for each $\mathfrak{S}_k$, so as to force $\theta_3^{(k)}$ to vanish. For example, if one were to require that all the components of $A_\mu$ be identical, the trilinear map in \eqref{eq:theta3kAAA slice} would vanish by symmetry. We leave this approach for future work.

In contrast, our strategy is to \textit{refine} $\mathfrak{S}_k$ by allowing the fall-offs (i.e., the maxima of the power selection sets) to vary, in such a way that the kinematic algebra defined on these refined slices becomes strict. This method enables us to construct the desired brackets without introducing additional relations among the components of the field, but merely by carefully selecting terms of their $r$-expansion\footnote{Together with a small precaution that we will explain below.}. In what follows, we develop this idea and show that for each slice, at least two distinct refinements\footnote{We wish to stress that, in the present context, the term \textit{refinement} is employed in a sense that diverges from its conventional meaning in topology.} can be defined on which the kinematic algebra is strictified.

\subsection{Strict kinematic algebras on refined slices}\label{Strict kin algebras on refined slices}
In what follows, as anticipated above, we allow the fall-off of the various components of the gauge fields to vary, by further projecting to zero some of the leading order coefficients of $A_\mu\in X_0^{(k)}$. This adjustment enables the construction of \textit{refined slices} $\mathfrak{S}_{k,f}\subseteq\mathfrak{S}_k$ on which a strict kinematic algebra can be defined (we add the subscript $f$ to indicate that these new slices depend on a modified fall-off). As a result, we obtain a kinematic bracket that satisfies the Jacobi identity, thereby forming a strict Lie algebra.

Before proceeding, a brief yet essential clarification: one might think that modifying the fall-offs (by projecting to zero particular coefficients) is not possible since, as shown in \cref{sec:SDYM at null infinity via slice truncation}, the projections $\Pi_k=\pi_{\mathcal{N}_k}$ (that define the slices $\mathfrak{S}_k$) are the only physical solutions to the master system, and the power selection sets in $\mathcal{N}_k$ have fixed maxima, i.e. fall-off. Any modification of these maxima would violate the $L_{\infty}$ relations, resulting in the loss of the $L_\infty$ algebras on $\mathfrak{S}_k$. This issue can be resolved by constructing the refined slices through an additional projection that not only further truncates $r$-expansion but also modifies the functional dependence of the gauge parameters, thus ensuring that the $L_{\infty}$ relations are obeyed in $\mathfrak{S}_{k,f}$, preserving the $L_{\infty}$ structure. Since this detail is not essential to the following discussion, we defer its analysis to \cref{subsec:kin k f}.

 The crucial point is that for each slice $\mathfrak{S}_k$, refined slices $\mathfrak{S}_{k,f}$ can be defined, each depending on a specific choice of the fall-off of the gauge fields. More precisely, $\mathfrak{S}_{k,f}$ are constructed following the same procedure described in \cref{subsec:an infinite family of Linf algebras} but using a modified version of $\Pi_k$. This new projection, which we call $\Pi_{k,f}$ and whose explicit definition is given in \cref{subsec:kin k f}, truncates the $r$-expansion of the gauge fields through power selection sets $N_{A_\mu}^{(k,f)}\subseteq N_{A_\mu}^{(k)}$ defined as:
\begin{equation}
\label{res:master system solution sets gauge field strict}
   N_{A_\mu}^{(k,f)}\coloneqq\begin{cases}
       \{k-2,\dots,f_r^k\} &\mu=r \\
       \{k-1,\dots,f_u^k\} &\mu=u \\
       \{k,\dots,f_z^k\} &\mu=z \\
       \{k-1,\dots,f_{\bar z}^k\} &\mu=\bar z
    \end{cases}
\end{equation}
where
\begin{equation}
    \label{eq:def inequalities for f}
    f_\mu^k\coloneqq\max N_{A_\mu}^{(k,f)}\in\mathbb{Z} \quad \text{s.th.} \quad \begin{cases}
        k-2\leq f^k_r\leq-2 \\
        k-1\leq f^k_u,f^k_{\bar z}\leq0 \\
        k\leq f^k_z\leq0
    \end{cases}
\end{equation}
are the modified fall-off. The inequalities above are necessary to ensure that the sets in \eqref{res:master system solution sets gauge field strict} are well-defined, and that they are subsets of $N_{A_\mu}^{(k)}$. We point out that the refined slices are obtained by modifying the fall-off but not the truncation of the $r$-expansion from ``below", which is not affected by the refinement procedure. In other words, $\min N_{A_\mu}^{(k,f)}\equiv\min N_{A_\mu}^{(k)}=t_\mu^k$ for all $k$ and for all $f$, cf. equations \eqref{res:master system solution sets gauge field strict} and \eqref{def:truncation and falloff}. Indeed, if we choose the fall-offs to be the same as those in \eqref{def:truncation and falloff}, then the resulting refined slices coincide with the initial ones, as expected
\begin{equation}
    \label{eq:refined slice reduces to initial one}
    \mathfrak{S}_{k,f}\equiv \mathfrak{S}_k \ \Leftrightarrow \ f_{\mu}^k\equiv f_\mu=(-2,0,0,0)\;.
\end{equation}

In a completely analogous way to what was done for the slices $\mathfrak{S}_k$, we can construct a kinematic algebra on each refined slice $\mathfrak{S}_{k,f}$, which we denote by $\mathfrak{Kin}^{(k,f)}$. The steps remain the same as before: we use the projection $\Pi_{k,f}$ to define a subspace $\mathcal{X}^{(k,f)}\subseteq\mathcal{X}$, on which we build an $L_\infty$ algebra by projecting the brackets $B_1,B_2$, then apply color-stripping, and finally obtain the algebra $\mathfrak{Kin}^{(k,f)}=\Pi_{k,f}\circ\mathfrak{Kin}$. Such an algebra, for a generic choice of fall-off $f_\mu^k$, is a homotopy algebra (to be precise, it is a $\mathrm{BV}_\infty^{\Box^{(k,f)}}$ algebra) and, as such, is not strict. However, crucially, there exist specific choices of fall-offs that strictify the kinematic algebra on the corresponding refined slice $\mathfrak{S}_{k,f}$. Specifically:
\begin{equation}
    \label{res:strict kin algebras on refined slices}
    \mathfrak{Kin}^{(k,f)} \text{\ is strict} \quad \Leftrightarrow \quad f_\mu^k \text{\ satisfy}\quad \begin{cases}
        f_r^k+f_u^k+f_z^k<k \\
        f_r^k+f_u^k+f_{\bar z}^k<k-1 \\
        f_r^k+f_z^k+f_{\bar z}^k<k \\
        f_u^k+f_z^k+f_{\bar z}^k<k+1 \\
        f_r^k+f_u^k<k \\
        f_r^k+f_{\bar z}^k<k \\
        f_u^k+f_z^k<k+2 \\
        f_z^k+f_{\bar z}^k<k+2 \;.
    \end{cases}
\end{equation}
This statement, which we prove in the appendix (see \cref{res:strict kin algebra}), provides a general criterion for finding (refined) slices on which a kinematic Lie bracket can be constructed. The first four inequalities correspond to the requirement that $\theta_3^{(k,f)}(\mathcal{A}_1,\mathcal{A}_2,\mathcal{A}_3)$ vanishes for any three fields $\mathcal{A}_i$, while the last four ensure that $\theta_3^{(k,f)}(\mathcal{E},\mathcal{A}_1,\mathcal{A}_2)=0$.

First of all, we observe that the fall-off $f_\mu$ in \eqref{def:truncation and falloff} is a solution of the system above if and only if $k$ vanishes. This fact together with equation \eqref{eq:refined slice reduces to initial one} imply that the kinematic algebra $\mathfrak{Kin}^{(k)}$ on the (unrefined) slice $\mathfrak{S}_k$ is strict only when $k=0$, in full agreement with what we found in \cref{subsec:Kinematic algebras on slices}.

In general, given a slice $\mathfrak{S}_k$ with $k\leq-1$, there exist multiple possible refined slices $\mathfrak{S}_{k,f}$ with $f_\mu^k$ satisfying the system above. We seek to minimise the number of constraints imposed on the fields, so as to find the most general refined slices manifesting a strict kinematic algebra. Of course any subslice of a refined slice (together with appropriate restrictions on the gauge parameters) will also be a solution.

The expansion for the gauge field corresponding to these solutions is
\begin{equation}
\label{exp with t and f}
    A_\mu=\sum_{n=t_\mu^k}^{f_{i,\mu}^k}r^nA_\mu^{(n)}
\end{equation}

Below, we present the solutions for the minimal non-trivial slice $k=-1$, and the general slices $k\leq-2$, respectively. 


\subsubsection{Strict kinematic algebra for \texorpdfstring{$\mathfrak{S}_{-1}$}{ slice -1}}
\label{subsec:strict kin algebra for k=-1}

For $k=-1$, we obtain 3 different slices, summarised in the table below:
\begin{table}[H]
    \centering
        \renewcommand{\arraystretch}{1.5}
        \begin{tabular}{c|c|c|c|c|}
        \cline{2-5}
        \multicolumn{1}{c|}{} & \multicolumn{4}{c|}{$k=-1$} \\
        \cline{2-5}
        \multicolumn{1}{c|}{} & $t_\mu^{-1}$ & $f_{1,\mu}^{-1}$ & $f_{2,\mu}^{-1}$ & $f_{3,\mu}^{-1}$ \\
        \hline
        \multicolumn{1}{|c|}{$r$} & $-3$ & $-2$ & $-2$ & $-3$ \\
        \multicolumn{1}{|c|}{$u$} & $-2$ & $-1$ & $0$ & $0$ \\
        \multicolumn{1}{|c|}{$z$} & $-1$ & $0$ & $0$ & $-1$ \\
        \multicolumn{1}{|c|}{$\bar z$} & $-2$ & $0$ & $-1$ & 0 \\
        \hline
        \end{tabular}
        \caption{The three refined slices $\mathfrak{S}_{-1,f_i}$ that strictify $\mathfrak{Kin}^{(-1)}$.}
    \label{table: solutions for slice -1}
\end{table}

Each of the solution above arises as a simple (partial) gauge choice on the slice\footnote{It is hence possible to reach these via projections which satisfy the cochain map condition \eqref{def:Pi_as_quasi_isomorphism}. We will not present these here explicitly, as our projected fields can be obtained by a simple field redefinition of the ones arising from imposing the co-chain map condition.}. We collect the explicit kinematic brackets below.

\begin{itemize}
\item Slice  $\mathfrak{S}_{(-1,f_1)}$
\be 
 f_{1,\mu}^{-1}=(-2,-1,0,0) \quad \Rightarrow \quad A_\mu=\begin{pmatrix}
        \frac{A_r^{(-2)}}{r^2}+\frac{A_r^{(-3)}}{r^3} \\
        \frac{A_u^{(-1)}}{r}+\frac{A_u^{(-2)}}{r^2} \\
        A_z^{(0)}+\frac{A_z^{(-1)}}{r} \\
        A_{\bar z}^{(0)}+\frac{A_{\bar z}^{(-1)}}{r}+\frac{A_{\bar z}^{(-2)}}{r^2}
    \end{pmatrix}\in X_0^{(-1,f_1)} \ ,
\ee
with the non-vanishing components of the kinematic bracket given by
\begin{equation}
        \begin{aligned}
            b_2(\mathcal{A}_1,\mathcal{A}_2)_u&=2r^{-2}\partial_u\left(\mathcal{A}_{[1\,z}^{(0)}\mathcal{A}_{2]\,\bar z}^{(0)}\right) \\
            b_2(\mathcal{A}_1,\mathcal{A}_2)_{\bar z}&=2r^{-2}\Big(\partial_u\mathcal{A}_{[1\,r}^{(-2)}\mathcal{A}_{2]\,\bar z}^{(0)}+2\mathcal{A}_{[1\,r}^{(-2)}\partial_u\mathcal{A}_{2]\,\bar z}^{(0)}-\mathcal{A}_{[1\,u}^{(-1)}\mathcal{A}_{2]\,\bar z}^{(0)} \\
            &\quad\quad\quad\ \ -\mathcal{A}_{[1\,\bar z}^{(0)}\partial_z\mathcal{A}_{2]\,\bar z}^{(0)}-\mathcal{A}_{[1\,z}^{(0)}\partial_{\bar z}\mathcal{A}_{2]\,\bar z}^{(0)}\Big) \ . \\
        \end{aligned}
    \end{equation}    
\item Slice  $\mathfrak{S}_{(-1,f_2)}$
\be 
 f_{2,\mu}^{-1}=(-2,0,0,-1) \quad \Rightarrow \quad A_\mu=\begin{pmatrix}
        \frac{A_r^{(-2)}}{r^2}+\frac{A_r^{(-3)}}{r^3} \\
        A_u^{(0)}+\frac{A_u^{(-1)}}{r}+\frac{A_u^{(-2)}}{r^2} \\
        A_z^{(0)}+\frac{A_z^{(-1)}}{r} \\
        \frac{A_{\bar z}^{(-1)}}{r}+\frac{A_{\bar z}^{(-2)}}{r^2}
    \end{pmatrix}\in X_0^{(-1,f_2)}  \ ,
\ee
with kinematic bracket given by
\begin{equation}
        \begin{aligned}
            b_2(\mathcal{A}_1,\mathcal{A}_2)_r&=4r^{-3}\mathcal{A}_{[1\,r}^{(-2)}\mathcal{A}_{2]\,u}^{(0)} \\
            b_2(\mathcal{A}_1,\mathcal{A}_2)_u&=2r^{-2}\left(\mathcal{A}_{[1\,r}^{(-2)}\partial_u\mathcal{A}_{2]\,u}^{(0)}-\mathcal{A}_{[1\,u}^{(0)}\mathcal{A}_{2]\,u}^{(-1)}+2\partial_z\mathcal{A}_{[1\,u}^{(0)}\mathcal{A}_{2]\,z}^{(0)}+\mathcal{A}_{[1\,u}^{(0)}\partial_z\mathcal{A}_{2]\,z}^{(0)}\right) \\
            b_2(\mathcal{A}_1,\mathcal{A}_2)_z&=-4r^{-1}\mathcal{A}_{[1\,u}^{(0)}\mathcal{A}_{2]\,z}^{(0)} \\
            b_2(\mathcal{A}_1,\mathcal{A}_2)_{\bar z}&=-2r^{-2}\left[2\mathcal{A}_{[1\,u}^{(0)}\mathcal{A}_{2]\,\bar z}^{(-1)}+\partial_{\bar z}\left(\mathcal{A}_{[1\,r}^{(-2)}\mathcal{A}_{2]\,u}^{(0)}\right)\right] \ . \\
        \end{aligned}
    \end{equation}
    Note that if we also set $\mathcal{A}_u^{(0)}=0$, then $b_2(\mathcal{A}_1,\mathcal{A}_2)=0$.
\item Slice  $\mathfrak{S}_{(-1,f_3)}$
\be 
 f_{3,\mu}^{-1}=(-3,0,-1,0) \quad \Rightarrow \quad A_\mu=\begin{pmatrix}
        \frac{A_r^{(-3)}}{r^3} \\
        A_u^{(0)}+\frac{A_u^{(-1)}}{r}+\frac{A_u^{(-2)}}{r^2} \\
        \frac{A_z^{(-1)}}{r} \\
        A_{\bar z}^{(0)}+\frac{A_{\bar z}^{(-1)}}{r}+\frac{A_{\bar z}^{(-2)}}{r^2}
    \end{pmatrix}\in X_0^{(-1,f_3)} \ ,
\ee
    with the non-vanishing components of the kinematic bracket given by
    \begin{equation}
        \label{eq:bs for slice f3}
        b_2(\mathcal{A}_1,\mathcal{A}_2)_\mu=2r^{-2}\mathcal{A}_{[1\,\mu}^{(0)}\left(\mathcal{A}_{2]\,u}^{(-1)}-\partial_z\mathcal{A}_{2]\,\bar z}^{(0)}\right) \;, \quad \mu=u,\bar z
    \end{equation}
    Note that if we also set $\mathcal{A}_u^{(0)}=0$, then the only non vanishing component is $\bar z$.
\end{itemize}


\subsubsection{Strict kinematic algebra for \texorpdfstring{$\mathfrak{S}_{k\leq-2}$}{ slice from -2}}
Refined slices $\mathfrak{S}_{k,f_i}$ with $k\leq-2$ are collected in the table below:
\begin{table}[H]
    \centering
    \renewcommand{\arraystretch}{1.5}
    \begin{tabular}{|c|c|c|c|c|c|c|c|}
    \cline{2-8}
    \multicolumn{1}{c|}{} & \multicolumn{7}{c|}{$\forall k\leq-2$} \\
    \cline{2-8}
    \hhline{~|~~~|----}
    \multicolumn{1}{c|}{} & $t_\mu^k$ & $f_{4,\mu}^k$ & $f_{5,\mu}^k$ & \cellcolor[gray]{0.9}$f_{6,\mu}^k$ & \cellcolor[gray]{0.9}$f_{7,\mu}^k$ & \cellcolor[gray]{0.9}$f_{8,\mu}^k$ & \cellcolor[gray]{0.9}$f_{9,\mu}^k$ \\
    \hhline{|-|---|----}  
    $r$ & $k-2$ & $-2$ & $k-2$ & \cellcolor[gray]{0.9}$k-1$ & \cellcolor[gray]{0.9}$-2$ & \cellcolor[gray]{0.9}$-2$ & \cellcolor[gray]{0.9}$k-1$ \\
    $u$ & $k-1$ & $k+1$ & $0$ & \cellcolor[gray]{0.9}$0$ & \cellcolor[gray]{0.9}$k$ & \cellcolor[gray]{0.9}$0$ & \cellcolor[gray]{0.9}$k$ \\
    $z$ & $k$ & $0$ & $k$ & \cellcolor[gray]{0.9}$0$ & \cellcolor[gray]{0.9}$k+1$ & \cellcolor[gray]{0.9}$k+1$ & \cellcolor[gray]{0.9}$0$ \\
    $\bar z$ & $k-1$ & $k+1$ & $0$ & \cellcolor[gray]{0.9}$k$ & \cellcolor[gray]{0.9}$0$ & \cellcolor[gray]{0.9}$k$ & \cellcolor[gray]{0.9}$0$ \\
    \hline
    \end{tabular}
    \caption{Refined slices with $k\leq -2$. Slices $\mathfrak{S}_{k,f_4}$ and $\mathfrak{S}_{k,f_5}$ fully strictify  $\mathfrak{Kin}^{(k)}$, by setting all $\theta_3$'s in \eqref{eq:theta3AAA explicit expression} and \eqref{eq:theta3EAA explicit expression} to $0$. Slices $\mathfrak{S}_{k,f_6}$ to $\mathfrak{S}_{k,f_9}$ are obtained from the weaker requirement that only  $\theta_3(\cA_{1},\cA_{2},\cA_{3})$  in \eqref{eq:theta3AAA explicit expression} vanishes.
} 
    \label{table: refined slices for k leq -2}
\end{table}

\renewcommand{\arraystretch}{1}

We have obtained six infinite families of solutions, valid for any $k\leq -2$. In contrast with the previous section, these slices require, in addition to some gauge choice, the elimination of certain unphysical degrees of freedom, or the restriction of the solution space\footnote{And hence they might not necessarily by obtainable by a projection satisfying the co-chain map condition.}.  Additionally, in each of the slices with $k\leq-2$, we found that restrictions were needed on two of the components, resulting in $\binom{4}{2}=6$ families of solutions.   

The kinematic algebra brackets can be obtained by plugging in the expansion for the gauge field given in \eqref{exp with t and f}, for the explicit values of $t_\mu^k$ and $f_{i,\mu}$, with $i\in\{4,5,6,7,8,9\}$ given in the table, into the expression in \eqref{b2_gen_exprs}, together with the projection in \eqref{res:master system solution strict}.

 Solutions $f_{4,\mu}^k$ and $f_{5,\mu}^k$ automatically strictify the kinematic algebra, as they satisfy all eight conditions in \eqref{res:strict kin algebras on refined slices}. For example, the non-vanishing components of the kinematic bracket between two gauge fields for the refined slice $\mathfrak{S}_{k,f_5}$ are
\begin{equation}
    b_2^{(k,f_5)}(\mathcal{A}_1,\mathcal{A}_2)_\mu=2r^{k-1}\mathcal{A}_{[1\,\mu}^{(0)}\left[(k+2)A_{2]\,u}^{(k)}-\partial_z\mathcal{A}_{2]\,\bar z}^{(k+1)}\right]\;,\quad \mu=u,\bar z\;.
\end{equation}
Notice that the above expression with $k=-1$ reduces to \eqref{eq:bs for slice f3}, the kinematic bracket of the slice $\mathfrak{S}_{-1,f_3}$.

The families of solutions in the shaded columns of the table, i.e. $\mathfrak{S}_{k,f_6}$ to $\mathfrak{S}_{k,f_9}$, only satisfy the first four conditions in \eqref{res:strict kin algebras on refined slices}. We have checked that they strictify the kinematic algebra upon imposing further constrains arising from the non-dynamical components of the equations of motion. Specifically, they turn out to be a subset of the components for which the quadratic part of the e.o.m. vanishes. 

Interestingly, for slice $\mathfrak{S}_{k,f_7}$ we find that these conditions are exactly required by the co-chain map condition, which amounts to imposing the following constraints on the projected gauge fields:
\begin{equation}\label{slice_7_conds}
    \begin{cases}
        \partial_uA_z^{(k+1)}=0 \\
        \partial_zA_{\bar z}^{(n+2)}=\partial_uA_r^{(n)} &n\in\{k,\dots,-2\}
    \end{cases}.
\end{equation}

Let us now recall the Schouten-Nijenhuis 
bracket previously obtained in \cite{Bonezzi:2023pox}, and which reduces to the Poisson bracket of Monteiro and O'Connell \cite{Monteiro:2011pc} upon writing the gauge field in terms of a scalar. This was reviewed in \autoref{sub:kin alg gen} and presented in Bondi coordinates in \autoref{sec:color stripping and kinematic algebra}. A natural question is whether this can arise as some kind of limit of one of the solutions in \autoref{table: refined slices for k leq -2}. Indeed this is the case. By taking the $k\to-\infty$ limit of slice 7 we have:
\be 
\lim_{k\rightarrow-\infty}\mathfrak{S}_{k,f_7} \ \Rightarrow \ A_u=A_z=0  \ .
\ee 
In this limit, the first condition in \eqref{slice_7_conds} is trivially satisfied, while the second one reduces to the  Lorenz gauge condition\footnote{Notice that the Lorenz condition in Bondi coordinates is given by $b(A)=0$, where $b$ is the operator in \eqref{def:differential b in components}.} discussed in \cite{Bonezzi:2023pox}, and hence the the kinematic bracket reduces to the Schouten-Nijenhuis bracket in the 2-dimensional space spanned by $r$ and $\bar{z}$, see \eqref{SN_bracket_in_Bondi_coords}. 

It turns out that the limit $k\to-\infty$ of slice 6 gives a very similar result: 
\be 
\lim_{k\rightarrow-\infty}\mathfrak{S}_{k,f_6} \ \Rightarrow \ A_r=A_{\bar z}=0 \ ,
\ee 
for which the kinematic bracket is just the Schouten-Nijenhuis in the 2-dimensional space spanned by $u$ and $z$. Indeed, this just results from an alternative gauge choice of the solution in \cite{Bonezzi:2023pox}, which leads to the elimination of a different unphysical degree of freedom.

It would be very interesting to study the other families of solutions in more detail, and see whether they correspond to any useful subspaces of the theory, perhaps for the study of subleading infrared phenomena. We leave this for future work.

\section{Fall-offs and quasi-isomorphism}
\label{sec:algebra_via_homotopy_transfer}
In the previous sections, we have shown how to construct $L_{\infty}$ that encode SDYM in different slices with different powers of the radial coordinate $r$. More precisely, we have defined projection maps that truncate the $r$ expansion of all the elements of the $L_{\infty}$ algebra that describes SDYM. Naturally, performing such truncations or projections can lead to a loss of physical information. Indeed, as we mentioned earlier and will explain in detail in the following subsection, such projections induce maps from the cohomology of the original theory to the cohomology of a truncated version of the theory. In general, these induced maps are not invertible (isomorphic), which implies that the cohomologies of the original and truncated theories are not the same, indicating the loss of physical information. 

To understand this from a mathematical perspective and to implement truncations that preserve all the physical information of the theory, we begin this section by introducing the notion \textit{quasi-isomorphisms}, which are morphisms between chain complexes that induce isomorphisms between their cohomologies. Following that initial formal discussion, we introduce maps we use to construct $L_{\infty}$ algebras with isomorphic cohomologies.

\subsection{Cochain maps and quasi-isomorphisms}\label{sec:cochain}
Intuitively, we want to identify subspaces of $\mathcal{X}^{\rm{SDYM}}$ that can describe physics close to null infinity. Therefore, we need to understand how to relate the physical information of the whole space $\mathcal{X}^{\rm{SDYM}}$ with that of the smaller subspaces, which we shall denote by $\bar{\mathcal{X}}$, that we wish to consider. In the $L_{\infty}$ formulation of field theories, the physical information of the theory is encoded in the cohomology of the differential $B_{1}$. Indeed, recalling the cochain complex of the theory
\begin{equation}
\begin{tikzcd}[row sep=2mm]
&X_{-1}\arrow{r}{B_{1}} & X_{0}\arrow{r}{B_{1}} & X_{1} \\
&\Lambda & A & E
\end{tikzcd}    \;,
\end{equation}
we see that the cohomology of $B_{1}$ corresponds to field configurations that satisfy the linear field equations modulo linearized gauge transformations, which correspond to the theory's physical or propagating degrees of freedom. Consequently, to understand the relation between the physics of the full theory and that of the subspaces, we need to understand the relation between the cohomologies of the cochain complexes $(B_{1}, \mathcal{X}^{\rm{SDYM}})$ and $(\bar{B}_{1}, \bar{\mathcal{X}})$ with $\bar{\mathcal{X}}\subset \mathcal{X}^{\rm{SDYM}}$, where $\bar B_{1}$ is the differential that describes the theory in the desired regime, i.e, at the desired power of the radial coordinate $r$. In the subsequent general discussion, we follow \cite{Arvanitakis:2020rrk}. 

The first step in establishing the relation between the physics of the full theory and that of the subspaces is to introduce a projection map $\pi:\mathcal{X}^{\rm{SDYM}}\to \bar{\mathcal{X}}$ which acts diagrammatically as
\begin{equation}
    \label{homotopy transfer diagram}
    \begin{tikzcd}[row sep=10mm]
        X_{-1}\arrow{d}{\pi}\arrow{r}{B_1}&X_{0}\arrow{d}{\pi}\arrow{r}{B_1}&X_{1}\arrow{d}{\pi}\\\
         \bar X_{-1}\arrow{r}{\bar B_1}& \bar X_0\arrow{r}{\bar B_1}& \bar X_1\;.
    \end{tikzcd}
\end{equation}

If the projection map $\pi$ obeys the following relation:
\begin{equation}\label{cochcond}
    \pi \circ B_{1} = \bar B_{1}\circ \pi\;,
\end{equation}
then we say that the $\pi$ is a \textit{cochain map}. Diagrammatically, we say that $\pi$ is a cochain map if the following diagram commutes
\begin{equation}
    \label{homotopy transfer diagram}
    \begin{tikzcd}[row sep=10mm]
        X_{i}\arrow{d}{\pi}\arrow{r}{B_1}&X_{i+1}\arrow{d}{\pi}\\
         \bar X_{i}\arrow{r}{\bar B_1}& \bar X_{i+1}\;,
    \end{tikzcd}
\end{equation}
for all $i$. In words, the diagram commutes if, starting from the space $X_{i}$, we can follow either of the two possible routes to $\bar X_{i+1}$ and obtain the same result. Naturally, cochain maps induce maps on the cohomologies $H^{i}$ and $\bar H^{i}$ of the differentials $B_{1}$ and $\bar B_{1}$, respectively, as we shall prove in the following. Consider an element $x\in H^{i}$, that is, $x\sim x+B_{1}(y)$, where $\sim$ indicates that they are equivalent in $H^{i}$. We denote the whole equivalence class as $[x]$. The projection map acts on an arbitrary representative of $H^{i}$ $[x+B_{1}(y)]$ as 
\begin{equation}
\begin{split}
\bar H^{i}\ni [\pi(x+B_{1}(y))]& =[\pi(x)+ \pi (B_{1}(y))]\\
& = [\pi(x)+\bar B_{1}(\pi(y))] \\
& = [\pi(x)]\;,
\end{split}
\end{equation}
where in the second line we used the cochain map condition \eqref{cochcond} and to get to the last line we used the fact that, in the cohomology $\bar H^{i}$, $\pi(x) \sim \pi(x)+\bar B_{1}(\pi(y))$. The above computation proves that, indeed, $\pi$ induces a well-defined, representative independent map on cohomologies, which, in general, is non-invertible. Physically, $\pi$ takes an arbitrary element of the full theory in $\mathcal{X}^{\rm{SDYM}}$ and projects it to the regime we wish to consider. As we shall define in detail in the following sections, for our purposes, $\pi$ truncates powers of the radial coordinate $r$ to obtain expansions in this coordinate of the form \eqref{A_falls_no_u_gauge}.

In this paper, we also consider a map $\iota:\bar{\mathcal{X}}\to \mathcal{X}^{\rm{SDYM}}$ called the \textit{inclusion map} that acts diagrammatically as
\begin{equation}
    \label{homotopy transfer diagram}
    \begin{tikzcd}[row sep=10mm]
        X_{-1}\arrow{r}{B_1}&X_{0}\arrow{r}{B_1}&X_{1}\\\
         \bar X_{-1}\arrow{u}{\iota}\arrow{r}{\bar B_1}& \bar X_0\arrow{u}{\iota}\arrow{r}{\bar B_1}& \bar X_1\arrow{u}{\iota}\;,
    \end{tikzcd}
\end{equation}
which tells us what the elements of $\bar{\mathcal{X}}$ look like in the original, bigger space $\mathcal{X}^{\rm{SDYM}}$. To induce a map in cohomologies, analogously to $\pi$, $\iota$ shall obey the cochain map condition
\begin{equation}
B_{1}\circ\iota = \iota \circ \bar B_{1}\;.
\end{equation}
Thus, the inclusion map $\iota$ induces a map inverse to that induced by the projection map $\pi$ in cohomologies, i.e.,
\begin{equation}
\begin{split}
H^{i}\ni [\iota(\bar x+\bar B_{1}(\bar y))]& =[\iota(\bar x)+ \iota (\bar B_{1}(\bar y))]\\
& = [\iota(\bar x)+ B_{1}(\iota(\bar y))] \\
& = [\iota(\bar x)]\;,
\end{split}
\end{equation}
with $\bar x, \bar y \in \bar{\mathcal{X}}$. In other words, there exists an isomorphism on cohomologies and we say that there exists a \textit{quasi-isomorphism} between the two chain complexes $(B_{1}, \mathcal{X}^{\rm{SDYM}})$ and $(\bar{B}_{1}, \bar{\mathcal{X}})$. Physically, this indicates that both chain complexes contain the same physical information.

\subsection{Cochain map conditions}
The goal is to construct a new $L_\infty$-algebra $(\bar{\mathcal{X}},\bar B_1,\bar B_2)$ that encodes the physics of the initial theory near $\mathcal{I}$. To this end, we shall define two maps $\pi:\mathcal{X}\rightarrow\bar{\mathcal{X}}$ and $\iota:\bar{\mathcal{X}}\rightarrow\mathcal{X}$ such that the diagram
\begin{equation}
    \label{homotopy transfer diagram}
    \begin{tikzcd}[row sep=10mm]
        0\arrow{r}{B_1}&X_{-1}\arrow{d}{\pi}\arrow{r}{B_1}&X_{0}\arrow{d}{\pi}\arrow{r}{B_1}&X_{1}\arrow{d}{\pi}\arrow{r}{B_1}&0\\
        0\arrow{r}{\bar B_1}& \bar X_{-1}\arrow[shift left=2mm]{u}{\iota}\arrow{r}{\bar B_1}& \bar X_0\arrow[shift left=2mm]{u}{\iota}\arrow{r}{\bar B_1}& \bar X_1\arrow[shift left=2mm]{u}{\iota}\arrow{r}{\bar B_1}&0
    \end{tikzcd}
\end{equation}
commutes. In other words, we require $\pi$ and $\iota$ to be cochain maps. In general, the first two brackets on $\bar{\mathcal{X}}$ are defined in terms of the brackets on $\mathcal{X}$ via the cochain maps as
\begin{equation}
    \label{def:bar B_i}
    \bar B_i:=\pi B_i \iota, \quad i=1,2 \;.
\end{equation}
We observe that
\begin{equation}
    \mathrm{\eqref{homotopy transfer diagram} \ commutes} \ \Leftrightarrow \ 
    \begin{cases}
        \bar B_1\pi=\pi B_1 \\
        B_1\iota=\iota\bar B_1
    \end{cases} \ \Leftrightarrow \ \begin{cases}
        \mathrm{Im}(1_\mathcal{X}-\iota\pi)\subseteq\mathrm{Ker}(\pi B_1) \\
        \mathrm{Im}(B_1\iota)\subseteq\mathrm{Ker}(1_\mathcal{X}-\iota\pi)
    \end{cases} \;,
\end{equation}
where in the second equivalence we used the definition of $\bar B_1$. We want $\bar{\mathcal{X}}$ to satisfy an equation similar to \eqref{def:space_at_I}, namely
\begin{equation}
    \label{eq:condition 1 for bar X}    \bar{\mathcal{X}}\subseteq\mathcal{X}^{n\leq0}\subset\mathcal{X} \;.
\end{equation}
In particular, we require $\bar{\mathcal{X}}$ to be a subspace of $\mathcal{X}$. Hence, the maps $\pi,\iota$ shall be, respectively, the projection to $\bar{\mathcal{X}}$ and the canonical inclusion:
\begin{equation}
    \begin{aligned}
        \pi\colon\mathcal{X}&\rightarrow\bar{\mathcal{X}} \\
        \psi&\mapsto\pi(\psi)\;,
    \end{aligned} \qquad
    \begin{aligned}
        \iota\colon\bar{\mathcal{X}}&\hookrightarrow\mathcal{X} \\
        \bar\psi&\mapsto\bar\psi\;.
    \end{aligned}
\end{equation}
Notice that for an arbitrary $\psi\in\mathcal{X}$, we can write $\iota\pi(\psi)=\pi(\psi)$, where both sides of this equality are viewed as elements of $\mathcal{X}$ due to \eqref{eq:condition 1 for bar X}. In particular,
\begin{equation}
    1_\mathcal{X}-\iota\pi=1_\mathcal{X}-\pi=:\pi^\complement \;,
\end{equation}
where $\pi^\complement$ is the projection to the subspace complementary to $\bar{\mathcal{X}}$ (see discussion in \cref{subsec:Requirements on L infty}). Moreover, $\mathcal{X}=\mathrm{Im}\pi\oplus\mathrm{Im}\pi^\complement$, where $\mathrm{Im}\pi=\bar{\mathcal{X}}$. Of course, $\mathrm{Im}\pi^\complement=\mathrm{Ker}\pi$ (and this identity is symmetric under $\pi\leftrightarrow\pi^\complement$). This implies that \eqref{homotopy transfer diagram} commutes iff
\begin{subnumcases}{\label{eq:condition subalgebra}}
    \mathrm{Im}\pi^\complement\subseteq\mathrm{Ker}(\pi B_1) \label{eq:condition subalgebra 1} \\
    \mathrm{Im}(B_1\iota)\subseteq\mathrm{Ker}\pi^\complement=\mathrm{Im}\pi=\bar{\mathcal{X}} \;. \label{eq:condition subalgebra 2}
\end{subnumcases}
In particular, the second inclusion above implies that the subspace $\bar{\mathcal{X}}$ is closed under the action of the differential $B_1$. Thus, $\bar B_1=B_1|_{\bar{\mathcal{X}}}$.

In \cref{sec:SDYM at null infinity via slice truncation} we showed that \eqref{eq:condition 1 for bar X} is solved by
\begin{equation}
    \label{def:homotopy transfer target space}
    \bar{\mathcal{X}}:=\{\psi\in\mathrm{Im}\pi_\mathcal{N} \ | \ \mathcal{N} \ \mathrm{satisfies \ \eqref{def:pi_N_contraints_on_N}}\} \;,
\end{equation}
where $\mathcal{N}$ is the collection of power selection sets introduced in \eqref{def:power selection sets} and $\pi_\mathcal{N}$ is the morphism defined in \eqref{def:better_truncation}. In other words, $\pi=\pi_\mathcal{N}$. In what follows, we will prove that there is only one choice of sets in $\mathcal{N}$ such that \eqref{eq:condition subalgebra} with $\pi=\pi_\mathcal{N}$ is satisfied. (For this reason, $\bar{\mathcal{X}}$ in \eqref{def:homotopy transfer target space} does not carry the label $\mathcal{N}$.)

\subsection{Fall-off from cochain map conditions}
\label{subsec:falloff_from_HT}
\begin{result}
    Let $D$ be the diagram in \eqref{homotopy transfer diagram} with $\bar{\mathcal{X}}$ as in \eqref{def:homotopy transfer target space}, projection $\pi=\pi_\mathcal{N}$ and canonical inclusion $\iota$. Then $D$ commutes iff the power selection sets in $\mathcal{N}$ are
    \begin{equation}
    \label{res:physical fall-off}
    \begin{aligned}
        N_\Lambda&=N_{A_\alpha}=N_{E_{uz}}=\mathbb{Z}^{\leq0} \\
        N_{A_r}&=N_{E_{ru}}=N_{E_{r\bar z}}=\mathbb{Z}^{<0}\setminus\{-1\}.
    \end{aligned}
    \end{equation}
\end{result}
\begin{proof}
    From equation \eqref{eq:condition subalgebra} it follows that
    \begin{equation}
        \label{pf:conditions D commutes}
        D \ \mathrm{commutes} \ \Leftrightarrow \ \begin{cases}
            \mathrm{Im}\pi_\mathcal{N}^\complement\subseteq\mathrm{Ker}(\pi_\mathcal{N} B_1) \\
            \mathrm{Im}(B_1\iota)\subseteq\mathrm{Ker}\pi_\mathcal{N}^\complement
        \end{cases}\Leftrightarrow \ \begin{cases}
            \Psi_\mathcal{N}(\psi)\coloneqq\pi_\mathcal{N} B_1\pi_\mathcal{N}^\complement(\psi)=0 &\forall\psi\in\mathcal{X} \\
            \Phi_\mathcal{N}(\bar\psi)\coloneqq\pi_\mathcal{N}^\complement B_1\pi_\mathcal{N}(\bar\psi)=0 &\forall\bar\psi\in\bar{\mathcal{X}} \;,
        \end{cases}
    \end{equation}
    where $\pi_\mathcal{N}^\complement\coloneqq 1_\mathcal{X}-\pi_\mathcal{N}$ is the projection to the complement of $\bar{\mathcal{X}}$ and we used $\iota(\bar\psi)=\bar\psi=\pi_\mathcal{N}(\bar\psi)$. We want to find the conditions that the power selection sets in $\mathcal{N}$ must obey so that the second system above holds. For the first equation of this system, the answer is provided by lemma \ref{lemma:conditions on N from commuting diagram} of \cref{sec:appendix_proofs}. There, we prove that $\Psi_\mathcal{N}(\psi)$ vanishes for all $\psi\in\mathcal{X}$ iff the sets in $\mathcal{N}$ satisfy the inclusions listed in \eqref{eq:conditions on N from commuting diagram}. To solve the second equation in \eqref{pf:conditions D commutes}, we observe that $\Phi_\mathcal{N}$ can be obtained from $\Psi_\mathcal{N}$ by switching $\pi_\mathcal{N}\leftrightarrow\pi_\mathcal{N}^\complement$. Moreover, from \eqref{def:better_truncation} and \eqref{pf:bar piN} it follows that $\pi_\mathcal{N}^\complement=\pi_{\mathcal{N}^\complement}$, where $\mathcal{N}^\complement$ is the collection of sets containing the complements in $\mathbb{Z}$ of the sets in $\mathcal{N}$. Therefore, we deduce that $\Phi_\mathcal{N}=\Psi_{\mathcal{N}^\complement}$. So, the system in \eqref{pf:conditions D commutes} is solved by a projector $\pi_\mathcal{N}$ with $\mathcal{N}$ such that both $\mathcal{N}$ and $\mathcal{N}^\complement$ satisfy the inclusions in \eqref{eq:conditions on N from commuting diagram}. Such a collection $\mathcal{N}$ consists of sets satisfying the system:
    \begin{subnumcases}{\label{pf:D commutes sets equalities}}
        N_{A_r}\setminus\{-1\}=\mathcal{T}_{-1}N_\Lambda\setminus\{-1\} \label{pf:D commutes sets equalities 1} \\
        N_{A_\alpha}=N_\Lambda \label{pf:D commutes sets equalities 2} \\
        N_{E_{ru}}=N_{A_r}=\mathcal{T}_{-1}N_{A_u}=\mathcal{T}_{-2}N_{A_z}=\mathcal{T}_{-2}N_{A_{\bar z}} \label{pf:D commutes sets equalities 3} \\
        N_{E_{r\bar z}}=N_{A_r} \\
        N_{E_{r\bar z}}\setminus\{-1\}=\mathcal{T}_{-1}N_{A_{\bar z}}\setminus\{-1\} \\
        N_{E_{uz}}=N_{A_u}=N_{A_z} \;,
    \end{subnumcases}
    where $\mathcal{T}$ is the translation defined in \eqref{def:translation}. To find the above equations, we used the following fact: given $M,N,L\subseteq\mathbb{Z}$ and $m\in\mathbb{Z}$, one can prove that
    \begin{enumerate}[label=\textit{\roman*})]
        \item if $M\setminus\{m\}\subseteq N$ and $M^\complement\setminus\{m\}\subseteq N^\complement$, then $M\setminus\{m\}=N\setminus\{m\}$;
        \item if $L\subseteq M\cap N$ and $L^\complement\subseteq M^\complement\cap N^\complement$, then $L=M=N$.
    \end{enumerate}
    At this point, it is easy to see that \eqref{res:physical fall-off} is the only possible solution for both \eqref{def:pi_N_contraints_on_N} (which is our hypothesis) and the system \eqref{pf:D commutes sets equalities}.
\end{proof}

Equation \eqref{res:physical fall-off} provides a proof of the fall-offs given in \eqref{A_falls_no_u_gauge}; thus, we are effectively proving the same result already established in \cref{sec:SDYM at null infinity via slice truncation}, but in a more concise and elegant way. Moreover, by comparing the power selection sets found solving the master system in \cref{sec:SDYM at null infinity via slice truncation} and the ones found here, we observe that
\begin{equation}
    \bar{\mathcal{X}}=\lim_{k\rightarrow-\infty}\mathcal{X}^{(k)} \;,
\end{equation}
cf. equations \eqref{res:physical fall-off} and \eqref{res:master system solution sets}.

\subsection{\texorpdfstring{$(\bar{\mathcal{X}},\bar B_1,\bar B_2)$}{X bar} is an \texorpdfstring{$L_\infty$}{L infty} sub-algebra}
\begin{lemma}
    $\bar B_2=B_2|_{\bar{\mathcal{X}}\otimes\bar{\mathcal{X}}}$
\end{lemma}
\begin{proof}
    From the definition of $\bar B_2$, see \eqref{def:bar B_i}, it follows that the result we want to prove is equivalent to the statement that $\bar{\mathcal{X}}$ is closed under the action of $B_2$\footnote{We observe that this is equivalent to saying that $B_2$ commutes with the inclusion $\iota$, since $\mathrm{Im}[B_2(\iota\otimes\iota)]\subseteq\mathrm{Ker}\pi_\mathcal{N}^\complement=\mathrm{Ker}(1_\mathcal{X}-\iota\pi_\mathcal{N})$ iff $B_2(\iota\otimes\iota)=\iota\bar B_2$.}, i.e.
    \begin{equation}
        \mathrm{Im}[B_2(\iota\otimes\iota)]\subseteq \bar{\mathcal{X}}=\mathrm{Im}\pi_\mathcal{N}=\mathrm{Ker}\pi_\mathcal{N}^\complement \;.
    \end{equation}
    In other words, it suffices to show that $\Xi\coloneqq\pi_\mathcal{N}^\complement B_2(\iota\otimes\iota)=\pi_\mathcal{N}^\complement B_2(\pi_\mathcal{N}\otimes\pi_\mathcal{N})$ vanishes on $\bar{\mathcal{X}}\otimes\bar{\mathcal{X}}$. This is always true, as one can show using the explicit form of $B_2$ and the properties of the projectors\footnote{In particular, use \eqref{eq:P property 7}.}. 
    For instance, for arbitrary $\bar\Lambda_1,\bar\Lambda_2\in\bar X_{-1}$,
    \begin{equation}
        \Xi(\bar\Lambda_1,\bar\Lambda_2)=-\mathcal{P}_{\mathbb{Z}^{>0}} [\mathcal{P}_{\mathbb{Z}^{\leq0}}\bar\Lambda_1,\mathcal{P}_{\mathbb{Z}^{\leq0}}\bar\Lambda_2]=-\mathcal{P}_{\mathbb{Z}^{>0}\cap\mathbb{Z}^{\leq0}} [\bar\Lambda_1,\bar\Lambda_2]=0 \;.
    \end{equation}
    Analogously, one can prove that $\Xi(\psi_1,\psi_2)$ vanishes for all other choices of $\psi_1,\psi_2\in\bar{\mathcal{X}}$.
\end{proof}

This implies that $(\bar{\mathcal{X}},\bar B_1,\bar B_2)$ is an $L_\infty$ sub-algebra of $(\mathcal{X}, B_1, B_2)$. In appendix \ref{app:Homotopy transfer}, we provide a proof of this result based on the concept of \textit{homotopy transfer} which consists of the transfer of algebraic information between different chain complexes. Moreover, we show how the kinematic algebra for these choices of fall-offs can be constructed.

\section{Conclusions}
\label{sec: Conclusions}

 In this article, we have employed the power of homotopy algebras to construct robust truncations of self-dual Yang-Mills in a $\frac{1}{r}$ expansion near null infinity. We have found that the physical fall-off of gauge fields can be \emph{derived} from the requirement that slices preserve the $L_{\infty}$ relations of the original theory. We then proceeded to construct novel strict kinematic algebras on various slices near null infinity.

In our search for a strict kinematic algebra, we restricted ourselves to axial-type (partial) gauge constrains in the current article. We anticipate that there are a number of other options for revealing the strict kinematic algebra, via the vanishing of the $\theta$ terms in \eqref{eq:theta3kAAA slice} and \eqref{eq:theta3EAA slice}, which exactly parametrise the failure of the algebra to be strict. We leave these for future work.

The natural next step is to extend our results to self-dual gravity. In this framework, it would be intriguing to explore whether the fall-off behaviour of the metric components, stemming from asymptotic flatness, can similarly be derived from a homotopy algebraic perspective. Additionally, incorporating the BMS symmetry group into this analysis presents an exciting avenue for future work. The presence of the strict kinematic algebras identified in this article also suggests that the double copy will proceed straightforwardly\footnote{For some previous work on the double copy at null infinity, see \cite{Campiglia:2021srh,Adamo:2021dfg,Nagy:2022xxs,Ferrero:2024eva,Mao:2021kxq,Godazgar:2021iae}}.

So far, our analysis has centred on the self-dual sub-sectors of field theories, where the cochain complex exhibits a simpler structure, and the $L_\infty$ algebra is limited to brackets no higher than $B_2$. Extending this framework to encompass full Yang-Mills and gravity will require working with a more intricate structure, including the presence of higher-order brackets. Nevertheless, we anticipate that the master equation in \eqref{eq:master_system} will admit a natural generalization of the following form:
\begin{subnumcases}{\label{eq:master_system_YM}}
\mathrm{Im}\Pi\subseteq\mathcal{X}^{n\leq0} \label{eq:master_system_EQ1_YM} \\
        \mathrm{Ker}\Pi\subseteq\mathrm{Ker}(\Pi B_1) \label{eq:master_system_EQ2_YM}\\
        \mathrm{Leib}\left(\Pi B_1|_{\mathrm{Im}\,\Pi},\Pi B_2|_{\mathrm{Im}\,\Pi\otimes\mathrm{Im}\,\Pi}\right)=0 \label{eq:master_system_EQ3_YM} \\
        \mathrm{Jac}\left(\Pi B_2|_{\mathrm{Im}\,\Pi\otimes\mathrm{Im}\,\Pi}\right)+[\Pi B_{1}|_{\mathrm{Im}\,\Pi}, \Pi B_3|_{\mathrm{Im}\,\Pi\otimes\mathrm{Im}\, \Pi\otimes\mathrm{Im}\,\Pi}]=0\label{eq:master_system_EQ4_YM}\\
        \Pi B_2|_{\mathrm{Im}\,\Pi\otimes\mathrm{Im}\,\Pi} \Pi B_3|_{\mathrm{Im}\,\Pi\otimes\mathrm{Im}\, \Pi\otimes\mathrm{Im}\,\Pi}+\Pi B_3|_{\mathrm{Im}\,\Pi\otimes\mathrm{Im}\, \Pi\otimes\mathrm{Im}\,\Pi} \Pi B_2|_{\mathrm{Im}\,\Pi\otimes\mathrm{Im}\,\Pi} = 0\;,
\end{subnumcases}
where in the second to last relation we included a term containing the three-bracket $B_{3}$ that captures the quartic interactions of Yang-Mills, and the last relation schematically represents a higher homotopy compatibility relation between the two-bracket $B_{2}$ and the higher bracket $B_{3}$. 

The study of infrared phenomena and their relations to asymptotic symmetries has more recently been extended to include subleading effects \cite{Campiglia:2015qka,Campiglia:2020qvc,Campiglia:2021oqz,Strominger:2013jfa,Lysov:2014csa,Donnelly:2016auv,Speranza:2017gxd,Freidel:2020ayo,Freidel:2020svx,Freidel:2020xyx,Freidel:2021dxw,Ciambelli:2021nmv,Freidel:2021cjp,Freidel:2021dfs,Ciambelli:2021vnn,Nagy:2024dme,Nagy:2025hip,Nagy:2022xxs,Geiller:2024bgf,Geiller:2022vto,Peraza:2023ivy,Campiglia:2018dyi}. This has unveiled infinite symmetry algebras \cite{Pope:1989ew,Bakas:1989xu,Fairlie:1990wv,Pope:1991ig,Pope:1991zka,Strominger:2021mtt,Guevara:2021abz,Strominger:2021lvk,Himwich:2021dau,Jiang:2021ovh,Boyer:1985aj,Park:1989fz,Park:1989vq,Adamo:2021lrv,Monteiro:2022lwm,Bu:2022iak,Bittleston:2023bzp,Bittleston:2024rqe,Taylor:2023ajd,Kmec:2024nmu,Nagy:2024dme,Lipstein:2023pih,Geiller:2024bgf,Freidel:2021ytz,CarrilloGonzalez:2024sto}, which appear in the OPE's of the CFT on the boundary in celestial holography. On the YM side, the relevant algebra is the so-called S-algebra (see \cite{Freidel:2023gue} and \cite{Hu:2023geb} for its phase space realisation). It would be interesting to include the overleading large gauge parameters (which are known to lead to subleading effects, see e.g. \cite{Campiglia:2021srh} and \cite{Nagy:2022xxs} for SDYM) in our analysis to see if the homotopy algebra gives rise to some known deformations of the S-algebra. On the gravity side, the corresponding infinite algebra is the $w_{1+\infty}$ algebra; here a potentially useful starting point is \cite{Monteiro:2022lwm}, where the  $w_{1+\infty}$ algebra was derived from the kinematic algebra in its original presentation as an area-preserving diffeomorphism in the self-dual sector \cite{Monteiro:2011pc}. Finally, another long-term goal to explore how to use the homotopy-transfer-based framework developed in the context of AdS/CFT in \cite{Chiaffrino:2023wxk} in this setting.

 \section*{Acknowledgments} 

We thank Roberto Bonezzi, Christoph Chiaffrino, and Olaf Hohm for enlightening discussions and previous collaborations on the topic. 

\noindent
F.D.J was funded by the
Deutsche Forschungsgemeinschaft (DFG, German Research Foundation) ``Rethinking Quantum Field Theory", Projektnummer 417533893/GRK2575. G.P. is funded by STFC Doctoral Studentship 2023. S.N. is supported in part by STFC consolidated grant T000708.

\appendix

\section{Construction of the master system's solutions}
\label{sec:appendix_proofs}
In this appendix, we provide a detailed proof of the solutions \eqref{res:master system solution} for the master system \eqref{eq:master_system}, using the formalism introduced in \cref{subsec:construction of solutions Pik}. We start by listing some useful properties of the projection $\mathcal{P}_N$, then we formulate the conditions that the power selection sets in \eqref{def:power selection sets} must satisfy in order for $\pi_\mathcal{N}$ to solve the system. Finally, we construct the infinite family of physically relevant solutions $\{\mathcal{N}_k\}_{k\in\mathbb{Z}^{\leq0}}$ for these conditions, which in turn allows us to derive the $L_\infty$ algebras encoding self-dual YM theory near null infinity.

\subsection{Properties of \texorpdfstring{$\mathcal{P}_N$}{PN}}
In this subsection, we provide some useful properties of the projection $\mathcal{P}_N$ defined in \eqref{def:projector_P_N}, which will be frequently referenced in the remainder of this appendix.

We begin by introducing some notation. Let $M,N\subseteq\mathbb{Z}$. We define the complement of $M$ in $\mathbb{Z}$ and the addition of the sets $M$ and $N$, respectively, as
\begin{equation}
    \label{def:complement and sum of sets}
    M^\complement\coloneqq\mathbb{Z}\setminus M\;, \quad M+N\coloneqq\{m+n \ | \ m\in M,n\in N\}\subseteq\mathbb{Z} \;.
\end{equation}
Moreover, we adopt the following notation for the projector $\mathcal{P}_N$: if $N$ is a singleton $\{n\}$, we simply write $\mathcal{P}_n:=\mathcal{P}_{\{n\}}$. With this in mind, and recalling that we work with Bondi coordinates $x^\mu=(r,y^\alpha)=(r,u,z,\bar z)$, we can prove the properties below.
\begin{lemma}[Properties of projection]
    \label{lemma:properties of projector P}
    Let $\mathcal{P}_N$ be the map
    \begin{equation}
        \mathcal{P}_{N}f(r,y^\alpha):=\sum_{n\in N}r^{n}f^{(n)}(y^\alpha) \tag{\ref{def:projector_P_N}}
    \end{equation}
    and $\mathcal{T}$ the translation in \eqref{def:translation}. Then $\forall s\in\mathbb{Z}$, $\forall L,M,N\subseteq\mathbb{Z}$ and $\forall f,g\in C^\infty_r(\mathcal{M})$ the following relations hold:
    \begin{subequations}
        \label{eq:P property}
        \begin{align}
            \label{eq:P property 1}
            \mathcal{P}_M\mathcal{P}_N&=\mathcal{P}_{M\cap N} \\
            \label{eq:P property 2}
            \mathcal{P}_M-\mathcal{P}_N&=\mathcal{P}_{M\setminus N} \quad \mathrm{if \ } N\subseteq M \\
            \label{eq:P property 3}
            \mathcal{P}_M(r^sf)&=r^s\mathcal{P}_{\mathcal{T}_{-s}M}f \\
            \label{eq:P property 4}
            \mathcal{P}_M\partial_r=\mathcal{P}_{M\setminus\{-1\}}\partial_r&=\partial_r\mathcal{P}_{\mathcal{T}_1M\setminus\{0\}}=\partial_r\mathcal{P}_{\mathcal{T}_1M} \\
            \label{eq:P property 5}
            \mathcal{P}_M\partial_\alpha&=\partial_\alpha\mathcal{P}_M \\
            \label{eq:P property 6}
            \mathcal{P}_M(fg)&=\sum_{m\in M}\sum_{n\in\mathbb{Z}}\mathcal{P}_{m-n}(f)\mathcal{P}_n(g) \\
            \label{eq:P property 7}
            \mathcal{P}_L(\mathcal{P}_M(f)\mathcal{P}_N(g))&=\begin{cases}
                \mathcal{P}_M(f)\mathcal{P}_N(g) \quad \mathrm{if \ } L\supseteq M+N \\
                0 \quad \quad \quad \quad \quad \ \ \ \mathrm{if \ } L\subseteq (M+N)^\complement
            \end{cases}
        \end{align}
    \end{subequations}
    In particular, $\mathcal{P}_N$ is a projector and the commutator $[\mathcal{P}_M,\mathcal{P}_N]$ vanishes for all $M,N$.
\end{lemma}
\begin{proof}
    $\mathcal{P}_M$ can be rewritten as
    \begin{align}
        \label{id:trivial identity for projector}
        \mathcal{P}_M=\sum_{m\in M}\mathcal{P}_m \;,
    \end{align}
    where $\mathcal{P}_m$ acts as
    \begin{equation}
        \label{pf:definition Pn}
        \mathcal{P}_m(f)\coloneqq\sum_{n\in\mathbb{Z}}\delta_{mn}r^nf^{(n)}=r^mf^{(m)}\;.
    \end{equation}
    By virtue of this trivial identity, we can prove the above lemma in the simple case when $L,M$ and $N$ are singletons; the general result can then be recovered using \eqref{id:trivial identity for projector}. $\forall l,m,n,s\in\mathbb{Z}$ and $\forall f,g\in C^\infty_r(\mathcal{M})$ the following identities hold:
    \begin{subequations}\label{eq:P_m properties}
        \begin{align}
            \mathcal{P}_m\mathcal{P}_nf&=\mathcal{P}_m(r^nf^{(n)})=\delta_{mn}r^nf^{(n)}=\delta_{mn}\mathcal{P}_nf \\
            \mathcal{P}_m(r^sf)&=\sum_{n\in\mathbb{Z}}\delta_{m,n+s}r^{n+s}f^{(n)}=r^{m}f^{(m-s)}=r^s\mathcal{P}_{m-s}f \label{eq:P_m properties 2}\\
            \mathcal{P}_m\partial_rf&=\sum_{n\in\mathbb{Z}}n\delta_{m,n-1}r^{n-1}f^{(n)}=(m+1)r^mf^{(m+1)} \label{eq:P_m properties 3}\\
            &=\partial_r(r^{m+1}f^{(m+1)})=\partial_r\mathcal{P}_{m+1}f \notag \\        \mathcal{P}_m(fg)&=\sum_{l,n\in\mathbb{Z}}\delta_{m,l+n}r^{l+n}f^{(l)}g^{(n)}=r^m\sum_{n\in\mathbb{Z}}f^{(m-n)}g^{(n)} \\
            &=\sum_{n\in\mathbb{Z}}\mathcal{P}_{m-n}(f)\mathcal{P}_n(g) \notag \\
            \mathcal{P}_l[\mathcal{P}_m(f)\mathcal{P}_n(g)]&=\delta_{l,m+n}r^{m+n}f^{(m)}g^{(n)}=\delta_{l,m+n}\mathcal{P}_m(f)\mathcal{P}_n(g) \;.
        \end{align}
    \end{subequations}
    The first equation above implies that $\mathcal{P}_m$ is a projector (since $\mathcal{P}_m^2=\mathcal{P}_m$), while from the third it is clear that
    \begin{equation}
        \label{eq:P-1_vanishes}
        \mathcal{P}_{-1}(\partial_rf)=0 \quad \forall f\in C^\infty_r(\mathcal{M}) \ .
    \end{equation}
    Note that this equation holds because we do not include logarithms in the definition of $C^\infty_r(\mathcal{M})$, see the expansion \eqref{def:C^infty_r}. It is easy to see that $\mathcal{P}_m$ commutes with $\partial_\alpha$. Equations \eqref{eq:P property 1}-\eqref{eq:P property 6} follow straightforwardly. For instance:
    \begin{equation}
        \mathcal{P}_M\mathcal{P}_N=\sum_{m\in M}\sum_{n\in N}\mathcal{P}_m\mathcal{P}_n=\sum_{m\in M}\sum_{n\in N}\mathcal{P}_m\delta_{mn}=\sum_{m\in M\cap N}\mathcal{P}_m=\mathcal{P}_{M\cap N} \ ,
    \end{equation}
    thus $\mathcal{P}_M$ is a projector.
\end{proof}

\subsection{Conditions on the power selection sets}
\label{subsec:Conditions on the power selection sets}
As discussed in \cref{subsec:projection piN}, the morphism $\pi_\mathcal{N}$ with power selection sets satisfying \eqref{def:pi_N_contraints_on_N} is the most general solution for $\Pi$ for the first equation of the master system. Below, we find the additional conditions that the sets in $\mathcal{N}$ must satisfy in order for $\pi_\mathcal{N}$ to solve all the other equations of the master system, too. In this way, we end up with a system of equations for the power selection sets. A projection $\pi_\mathcal{N}$ with $\mathcal{N}$ solving such a system will be the desired solution for $\Pi$ in \eqref{eq:master_system}.

\subsubsection{Cochain map condition}
We start by looking at the second equation of the master system \eqref{eq:master_system}, the one encoding the request that $\Pi$ is a cochain map.
\begin{lemma}[Cochain map condition]
    \label{lemma:conditions on N from commuting diagram}
    Let $\pi_\mathcal{N}$ be the morphism defined in \eqref{def:better_truncation} with the sets in $\mathcal{N}$ obeying condition \eqref{def:pi_N_contraints_on_N}. Then $\pi_\mathcal{N}$ is a solution for $\Pi$ in equation \eqref{eq:master_system_EQ2} iff the corresponding power selection sets satisfy all the inclusions\footnote{$\mathcal{T}$ is the translation defined in \eqref{def:translation}.}
    \begin{subequations}
        \label{eq:conditions on N from commuting diagram}
        \begin{align}
            \label{eq:condition on N 1}
            N_{A_r}\setminus\{-1\}&\subseteq \mathcal{T}_{-1}N_\Lambda \\
            \label{eq:condition on N 2}
            N_{A_\alpha}&\subseteq N_\Lambda \\
            \label{eq:condition on N 3}
            N_{E_{ru}}&\subseteq N_{A_r}\cap\mathcal{T}_{-1}N_{A_u}\cap\mathcal{T}_{-2}\left(N_{A_z}\cap N_{A_{\bar z}}\right) \\
            \label{eq:condition on N 4}
            N_{E_{r\bar z}}&\subseteq N_{A_r} \\
            \label{eq:condition on N 4.1}
            N_{E_{r\bar z}}\setminus\{-1\}&\subseteq \mathcal{T}_{-1}N_{A_{\bar z}} \\
            \label{eq:condition on N 5}
            N_{E_{uz}}&\subseteq N_{A_u}\cap N_{A_z} \; .
        \end{align}
    \end{subequations}
\end{lemma}
\begin{proof}
    We want to find the constraints that the sets in $\mathcal{N}$ shall satisfy in order to have $\mathrm{Ker}\pi_\mathcal{N}\subseteq\mathrm{Ker}(\pi_\mathcal{N} B_1)$. Following the same approach discussed in \cref{subsec:Requirements on L infty}\footnote{Cfr. equations \eqref{def:bar Pi} and \eqref{eq:condition_on_Im_barPi}.}, it is convenient to rewrite this expression as $\mathrm{Im}\pi_\mathcal{N}^\complement\subseteq\mathrm{Ker}(\pi_\mathcal{N}B_1)$, where $\pi_\mathcal{N}^\complement:=1_\mathcal{X}-\pi_\mathcal{N}$ is the projection to the complement of $\mathrm{Im}\pi_\mathcal{N}$ in $\mathcal{X}$\footnote{Recall that, by hypothesis, $\mathrm{Im}\pi_\mathcal{N}\subseteq\mathcal{X}^{n\leq0}\subset\mathcal{X}$, see equations \eqref{def:space_at_I} and \eqref{eq:master_system_EQ1}.}. This condition is equivalent to asking that
    \begin{equation}
        \label{pf:psi_N}
        \Psi_\mathcal{N}(\psi):=\pi_\mathcal{N}B_1\pi_\mathcal{N}^\complement(\psi)=0 \quad \forall\psi\in\mathcal{X} \;,
    \end{equation}
    where we introduce the notation $\Psi_\mathcal{N}$ to increase readability in what follows. From equations \eqref{def:better_truncation} and \eqref{eq:P property 2} we find that, $\forall\psi\in\mathcal{X}$,
    \begin{equation}
        \label{pf:bar piN}
        \pi_\mathcal{N}^\complement(\psi)=\psi-\pi_\mathcal{N}(\psi)=\begin{cases}
            (\mathcal{P}_\mathbb{Z}-\mathcal{P}_{N_\Lambda})\Lambda=\mathcal{P}_{N_\Lambda^\complement}\Lambda &\mathrm{if} \ \psi\in X_{-1} \\
            (\mathcal{P}_\mathbb{Z}-\mathcal{P}_{N_{A_\mu}})A_\mu=\mathcal{P}_{N_{A_\mu}^\complement}A_\mu &\mathrm{if} \ \psi\in X_0 \\
            (\mathcal{P}_\mathbb{Z}-\mathcal{P}_{N_{E_{\mu\nu}}})E_{\mu\nu}=\mathcal{P}_{N_{E_{\mu\nu}}^\complement}E_{\mu\nu} &\mathrm{if} \ \psi\in X_1
        \end{cases} \;.
    \end{equation}
    where $N_I^\complement:=\mathbb{Z}\setminus N_I$ is the complement of the set $N_I$, with $I\in\{\Lambda,A_\mu,E_{\mu\nu}\}$.
    
    This observation allows us to rewrite $\Psi_\mathcal{N}(\psi)$ in terms of the projection $\mathcal{P}$, so that we can manipulate it using the properties listed in \eqref{eq:P property}. If $\psi=\Lambda\in X_{-1}$, then $\Psi_\mathcal{N}(\Lambda)$ is an element of $X_0$ ($\Psi_\mathcal{N}$ has degree $+1$, since it contains the differential) and its components are\footnote{Recall equations \eqref{def:better_truncation}, \eqref{def:B_1 in components} and \eqref{pf:bar piN} for $\pi_\mathcal{N}$, $B_1$ and $\pi_\mathcal{N}^\complement$, respectively. Moreover, use properties \eqref{eq:P property 1}, \eqref{eq:P property 4} and \eqref{eq:P property 5}.}
    \begin{equation}
        \label{pf:PsiN Lambda}
        \begin{aligned}
            \Psi_\mathcal{N}(\Lambda)_r&=\mathcal{P}_{N_{A_r}}\partial_r\mathcal{P}_{N_\Lambda^\complement}\Lambda=\mathcal{P}_{N_{A_r}\setminus\{-1\}\cap \mathcal{T}_{-1}N_\Lambda^\complement}\partial_r\Lambda \\
            \Psi_\mathcal{N}(\Lambda)_\alpha&=\mathcal{P}_{N_{A_\alpha}}\partial_\alpha\mathcal{P}_{N_\Lambda^\complement}\Lambda=\mathcal{P}_{N_{A_\alpha}\cap N_\Lambda^\complement}\partial_\alpha\Lambda \;.
        \end{aligned}
    \end{equation}
    To study when these expressions reduce to zero, note that given any subsets $U,V\subseteq\mathbb{Z}$, the intersection $U\cap V^\complement$ vanishes iff $U\subseteq V$. Moreover, the translation $\mathcal{T}_sU$ of the set $U$ by an arbitrary integer $s$ satisfies the relation $(\mathcal{T}_sU)^\complement=\mathcal{T}_s(U^\complement)$. Thus, the components in \eqref{pf:PsiN Lambda} vanish $\forall\Lambda\in X_{-1}$ iff the inclusions \eqref{eq:condition on N 1} and \eqref{eq:condition on N 2} are satisfied. A similar argument applies when $\psi=A_\mu\in X_0$. In this case, $\Psi_\mathcal{N}(A)$ is an equation of motion in $X_1$ and its independent components are given by
    \begin{equation}
        \begin{aligned}
            \Psi_\mathcal{N}(A)_{ru}&=2\mathcal{P}_{N_{E_{ru}}}\big[\partial_{[r}\big(\mathcal{P}_{N_A^\complement}A\big)_{u]}+r^{-2}\partial_{[z}\big(\mathcal{P}_{N_A^\complement}A\big)_{\Bar{z}]}\big] \\
            &=\mathcal{P}_{N_{E_{ru}}}\big[\mathcal{P}_{\mathcal{T}_{-1}N_{A_u}^\complement}\partial_{r}A_{u}-\mathcal{P}_{N_{A_r}^\complement}\partial_{u}A_{r}+\mathcal{P}_{\mathcal{T}_{-2}N_{A_{\Bar{z}}}^\complement}(r^{-2}\partial_{z}A_{\Bar{z}})-\mathcal{P}_{\mathcal{T}_{-2}N_{A_{z}}^\complement}(r^{-2}\partial_{\bar{z}}A_z)\big] \\
            \Psi_\mathcal{N}(A)_{r\bar{z}}&=4\mathcal{P}_{N_{E_{r\bar{z}}}}\partial_{[r}\big(\mathcal{P}_{N_A^\complement}A\big)_{\bar{z}]}=2\mathcal{P}_{N_{E_{r\bar{z}}}}\big(\mathcal{P}_{\mathcal{T}_{-1}N_{A_{\bar{z}}}^\complement\setminus\{-1\}}\partial_{r}A_{\bar{z}}-\mathcal{P}_{N_{A_r}^\complement}\partial_{\bar{z}}A_{r}\big) \\
            \Psi_\mathcal{N}(A)_{uz}&=4\mathcal{P}_{N_{E_{uz}}}\partial_{[u}\big(\mathcal{P}_{N_A^\complement}A\big)_{z]}=2\mathcal{P}_{N_{E_{uz}}}\big(\mathcal{P}_{N_{A_z}^\complement}\partial_{u}A_z-\mathcal{P}_{N_{A_u}^\complement}\partial_zA_{u}\big) \;.
        \end{aligned}
    \end{equation}
    Since we want these expressions to vanish for arbitrary $A_\mu$, we require that each term above vanishes independently (because different terms in a given component of $\Psi_\mathcal{N}(A)_{\mu\nu}$ depend on different components of $A_\mu$). For instance,
    \begin{equation}
        \Psi_\mathcal{N}(A)_{uz}=0 \quad \forall A_\mu \ \Leftrightarrow \ \begin{cases}
            \mathcal{P}_{N_{E_{uz}}\cap N_{A_z}^\complement}\partial_{u}A_z=0 &\forall A_z \\
            \mathcal{P}_{N_{E_{uz}}\cap N_{A_u}^\complement}\partial_zA_{u}=0 &\forall A_u
        \end{cases} \ \Leftrightarrow \ N_{E_{uz}}\subseteq N_{A_u}\cap N_{A_z} \;,
    \end{equation}
    which gives equation \eqref{eq:condition on N 5}. Similarly, one can prove that $\Psi_\mathcal{N}(A)_{ru}$ and $\Psi_\mathcal{N}(A)_{r\bar z}$ vanish for all $A_\mu$ iff equations \eqref{eq:condition on N 3}, \eqref{eq:condition on N 4} and \eqref{eq:condition on N 4.1} hold. Finally, by degree, $\Psi_\mathcal{N}(E)=0$ for all $E\in X_1$, so no additional information on $\mathcal{N}$ can be extracted from this condition.
\end{proof}

\subsubsection{Leibniz rule condition}
The next step is to determine the additional conditions that the power selection sets in $\mathcal{N}$ must satisfy to ensure that $\pi_\mathcal{N}$ is a solution of the Leibniz rule in \eqref{eq:master_system_EQ3}.
\begin{lemma}[Leibniz rule condition]
    \label{lemma:conditions on N from Leib}
    Let $\pi_\mathcal{N}$ be the morphism defined in \eqref{def:better_truncation} with the sets in $\mathcal{N}$ obeying conditions \eqref{def:pi_N_contraints_on_N} and \eqref{eq:conditions on N from commuting diagram}. Then $\pi_\mathcal{N}$ is a solution for $\Pi$ in equation \eqref{eq:master_system_EQ3} iff the corresponding power selection sets satisfy all the following equations:\footnote{As a reminder, we reference the notation defined in \eqref{def:complement and sum of sets}.}
    \begin{subequations}
        \label{eq:conditions on N from Leibniz}
        \begin{align}
            \label{eq:condition on N 6}
            N_{A_r}\setminus\{-1\}\cap\big[\mathcal{T}_{-1}N_\Lambda\setminus (N_{A_r}\cup\{-1\})+N_\Lambda\big]&=\emptyset \\
            \label{eq:condition on N 7}
            N_{A_\alpha}\cap\big(N_\Lambda\setminus N_{A_\alpha}+N_\Lambda\big)&=\emptyset \\
            \label{eq:condition on N 8}
            N_{E_{ru}}\cap\big[\big(N_{A_r}\cup \mathcal{T}_{-1}(N_{A_u}\setminus\{0\})\cup \mathcal{T}_{-2}(N_{A_z}\cup N_{A_{\bar{z}}})\big)\setminus N_{E_{ru}}+N_\Lambda\big]&=\emptyset \\
            \label{eq:condition on N 9}
            N_{E_{r\bar z}}\setminus\{-1\}\cap \big[\big(N_{A_r}\cup\mathcal{T}_{-1}(N_{A_{\bar z}}\setminus\{0\})\big)\setminus N_{E_{r\bar z}}+N_\Lambda\big]&=\emptyset \\
            \label{eq:condition on N 10}
            N_{E_{uz}}\cap\big[\big(N_{A_u}\cup N_{A_z}\big)\setminus N_{E_{uz}}+N_\Lambda\big]&=\emptyset \;,
        \end{align}
    \end{subequations}
    together with the conditions
    \begin{subequations}
        \label{eq:conditions on N from Leibniz last ones}
        \begin{align}
            \label{eq:condition on N 11}
            N_{A_r}\ni-1\ &\Rightarrow\ N_\Lambda\ni0 \\
            \label{eq:condition on N 12}
            N_{E_{r\bar z}}\ni-1\ &\Rightarrow\ N_{A_{\bar z}}\ni0 \;.
        \end{align}
    \end{subequations}
\end{lemma}
\begin{proof}
    The goal is to determine the conditions on the sets in the collection $\mathcal{N}$ so that
    \begin{equation}
        \mathrm{Leib}_\mathcal{N}(\psi_1,\psi_2)=0 \quad \forall\psi_1,\psi_2\in\mathrm{Im}\pi_\mathcal{N} \;,
    \end{equation}
    where, for brevity, we defined $\mathrm{Leib}_\mathcal{N}:=\mathrm{Leib}\left(\pi_\mathcal{N} B_1|_{\mathrm{Im}\pi_\mathcal{N}},\pi_\mathcal{N} B_2|_{\mathrm{Im}\pi_\mathcal{N}\otimes\mathrm{Im}\pi_\mathcal{N}}\right)$. By degree, the above equation is a trivial identity when the sum of the degree of $\psi_1$ and $\psi_2$ is non-negative. Moreover, the Leibniz rule is graded symmetric\footnote{From equations \eqref{def:graded_symmetry_Bn} and \eqref{def:Leibniz_rule_B1_B2} it follows that $\mathrm{Leib}_\mathcal{N}(\psi_1,\psi_2)=(-1)^{\psi_1\psi_2}\mathrm{Leib}_\mathcal{N}(\psi_2,\psi_1)$.}. Therefore, the problem reduces to the study of the two cases with $\psi_1\in X_{-1}$ and $\psi_2$ belonging to either $X_{-1}$ or $X_0$.
    
    To begin, we apply $\mathrm{Leib}_\mathcal{N}$ to two gauge parameters $\Lambda_1,\Lambda_2\in\mathrm{Im}\pi_\mathcal{N}\cap X_{-1}$, obtaining the following gauge field:
    \begin{equation}
        \label{pf:LeibN lambdas all components}
        \begin{aligned}
            \mathrm{Leib}_\mathcal{N}(\Lambda_1,\Lambda_2)_\mu&=\pi_\mathcal{N}B_1\pi_\mathcal{N}B_2(\Lambda_1,\Lambda_2)_\mu+\pi_\mathcal{N}B_2(\pi_\mathcal{N}B_1(\Lambda_1)_\mu,\Lambda_2)-\pi_\mathcal{N}B_2(\Lambda_1,\pi_\mathcal{N}B_1(\Lambda_2)_\mu) \\
            &=\mathcal{P}_{N_{A_\mu}}\big(\partial_\mu \pi_\mathcal{N}B_2(\Lambda_1,\Lambda_2)+[\pi_\mathcal{N}B_1(\Lambda_1)_\mu,\Lambda_2]+[\Lambda_1,\pi_\mathcal{N}B_1(\Lambda_2)_\mu]\big) \\
            &=\mathcal{P}_{N_{A_\mu}}\big(-\partial_\mu\mathcal{P}_{N_\Lambda}[\Lambda_1,\Lambda_2]+2[\mathcal{P}_{N_{A_\mu}}\partial_\mu\Lambda_{[1},\Lambda_{2]}]\big) \;,
        \end{aligned}
    \end{equation}
    where we used the expressions \eqref{def:B_1 in components} and \eqref{def:B_2 in components} for $B_1$ and $B_2$, respectively, and the definition of $\pi_\mathcal{N}$. At this point, it is worth noting the following fact: since $\Lambda_i$ (with $i=1,2$) is an element of $\mathrm{Im}\pi_\mathcal{N}$ and $\pi_\mathcal{N}$ is idempotent, it is always possible to write
    \begin{equation}
        \label{pf:extract projector}
        \Lambda_i=\mathcal{P}_{N_\Lambda}\Lambda_i \;.
    \end{equation}
    By applying this observation, together with the properties in \eqref{eq:P property}, the $r$-component of the expression above takes the form 
    \begin{equation}
        \begin{aligned}
            \mathrm{Leib}_\mathcal{N}(\Lambda_1,\Lambda_2)_r&=-\mathcal{P}_{N_{A_r}\setminus\{-1\}}\partial_r\mathcal{P}_{N_\Lambda}[\Lambda_1,\Lambda_2]+2\mathcal{P}_{N_{A_r}}[\mathcal{P}_{N_{A_r}\setminus\{-1\}}\partial_r\Lambda_{[1},\Lambda_{2]}] \\
            \label{pf:L_2 lambda r component}
            &=-2\mathcal{P}_{N_{A_r}\setminus\{-1\}\cap \mathcal{T}_{-1}N_\Lambda}[\mathcal{P}_{\mathcal{T}_{-1}N_\Lambda\setminus\{-1\}}\partial_r\Lambda_{[1},\Lambda_{2]}] \\
            &\quad +2\mathcal{P}_{N_{A_r}}[\mathcal{P}_{N_{A_r}\setminus\{-1\}}\partial_r\Lambda_{[1},\Lambda_{2]}] \;.
        \end{aligned}
    \end{equation}
    The projector $\mathcal{P}_{N_{A_r}}$ in the last line can be written as $\mathcal{P}_{N_{A_r}\setminus\{-1\}}+\delta_{A_r,-1}\mathcal{P}_{-1}$, where
    \begin{equation}
        \label{pf:delta A_r}
        \delta_{A_r,-1}:=\begin{cases}
            1 &\mathrm{if} \ N_{A_r}\ni-1 \\
            0 &\mathrm{otherwise}
        \end{cases}
    \end{equation}
    Moreover, equation \eqref{eq:condition on N 1} (which holds by hypothesis) implies $N_{A_r}\setminus\{-1\}\cap \mathcal{T}_{-1}N_\Lambda=N_{A_r}\setminus\{-1\}$. Thus the above expression becomes
    \begin{equation}
        \begin{aligned}
            \mathrm{Leib}_\mathcal{N}(\Lambda_1,\Lambda_2)_r&=-2\mathcal{P}_{N_{A_r}\setminus\{-1\}}[(\mathcal{P}_{\mathcal{T}_{-1}N_\Lambda\setminus\{-1\}}-\mathcal{P}_{N_{A_r}\setminus\{-1\}})\partial_r\Lambda_{[1},\Lambda_{2]}] \\
            &\quad +2\delta_{A_r,-1}\mathcal{P}_{-1}[\mathcal{P}_{N_{A_r}\setminus\{-1\}}\partial_r\Lambda_{[1},\Lambda_{2]}] \\
            \label{pf:LeibNr}
            &=-2\mathcal{P}_{N_{A_r}\setminus\{-1\}}[\mathcal{P}_{\mathcal{T}_{-1}N_\Lambda\setminus(N_{A_r}\cup\{-1\})}\partial_r\Lambda_{[1},\mathcal{P}_{N_\Lambda}\Lambda_{2]}] \\
            &\quad +2\delta_{A_r,-1}\mathcal{P}_{-1}[\mathcal{P}_{N_{A_r}\setminus\{-1\}}\partial_r\Lambda_{[1},\mathcal{P}_{N_\Lambda}\Lambda_{2]}] \;,
        \end{aligned}
    \end{equation}
    where in the second equality we used \eqref{eq:P property 2}. The last line vanishes for all $\Lambda_1,\Lambda_2$. We prove this in two steps. First, observe that the hypotheses of this lemma imply that:
    \begin{equation}
        \label{eq:0_not_in_NAr}
        N_{A_r}\not\ni0 \;.
    \end{equation}
    To show why this is true, we argue by contradiction. Suppose $N_{A_r}\ni0$, then from \eqref{eq:condition on N 1} we conclude that $1\in N_\Lambda$. However, this contradicts \eqref{def:pi_N_contraints_on_N 1}, since all the sets we are considering here must not contain positive integers. Thus, $N_{A_r}\not\ni0$. The second step is to note that $\max N_\Lambda\leq0$ and $\max N_{A_r}<0$, thanks to \eqref{def:pi_N_contraints_on_N 1} and \eqref{eq:0_not_in_NAr}, hence\footnote{Here we use the fact that for all sets $U,V\subseteq\mathbb{Z}$, the maximum of the sum of the two sets is the sum of their maxima, i.e. $\max (U+V)=\max U+\max V$.}
    \begin{equation}
        \max(N_{A_r}\setminus\{-1\}+N_\Lambda)<-1 \ \Rightarrow\ -1\not\in N_{A_r}\setminus\{-1\}+N_\Lambda \;,
    \end{equation}
    which implies that the second term in the RHS of \eqref{pf:LeibNr} is zero by virtue of \eqref{eq:P property 7}. Employing this property once more, we observe that the other term vanishes for all $\Lambda_1,\Lambda_2$ iff equation \eqref{eq:condition on N 6} is satisfied. Now, consider the other components of \eqref{pf:LeibN lambdas all components}, namely $\mu=\alpha$. By following the same line of reasoning as above, we find\footnote{Use properties \eqref{eq:P property 1}, \eqref{eq:P property 2}, \eqref{eq:P property 5}, \eqref{eq:P property 7}, \eqref{eq:condition on N 2} and equation \eqref{pf:extract projector}.}
    \begin{equation}
        \begin{aligned}
            \mathrm{Leib}_\mathcal{N}(\Lambda_1,\Lambda_2)_\alpha&=\mathcal{P}_{N_{A_\alpha}}\big(-\partial_\alpha\mathcal{P}_{N_\Lambda}[\Lambda_1,\Lambda_2]+2[\mathcal{P}_{N_{A_\alpha}}\partial_\alpha\Lambda_{[1},\Lambda_{2]}]\big) \\
            &=-2\mathcal{P}_{N_{A_\alpha}\cap N_\Lambda}[\mathcal{P}_{N_\Lambda}\partial_\alpha\Lambda_{[1},\Lambda_{2]}]+2\mathcal{P}_{N_{A_\alpha}}[\mathcal{P}_{N_{A_\alpha}}\partial_\alpha\Lambda_{[1},\Lambda_{2]}] \\
            &=-2\mathcal{P}_{N_{A_\alpha}}[\mathcal{P}_{N_\Lambda\setminus N_{A_\alpha}}\partial_\alpha\Lambda_{[1},\mathcal{P}_{N_\Lambda}\Lambda_{2]}] \ ,
        \end{aligned}
    \end{equation}
    that is zero for all $\Lambda_1,\Lambda_2\in\mathrm{Im}\pi_\mathcal{N}\cap X_{-1}$ iff $N_\Lambda$ and $N_{A_\alpha}$ satisfy equation \eqref{eq:condition on N 7}.

    We now turn to the case where we apply the Leibniz rule to a gauge parameter $\Lambda\in\mathrm{Im}\pi_\mathcal{N}\cap X_{-1}$ and a gauge field $A_\mu\in\mathrm{Im}\pi_\mathcal{N}\cap X_0$:
    \begin{align}
        \label{pf:LeibN lambda A all components}
        \mathrm{Leib}_\mathcal{N}(\Lambda,A)_{\mu\nu}&=\pi_\mathcal{N}B_1\pi_\mathcal{N}B_2(\Lambda,A)_{\mu\nu}+\pi_\mathcal{N}B_2(\pi_\mathcal{N}B_1(\Lambda),A)_{\mu\nu}-\pi_\mathcal{N}B_2(\Lambda,\pi_\mathcal{N}B_1(A))_{\mu\nu} \;.
    \end{align}
    The line of reasoning is the same as before, although the expressions are more involved. Let us analyse the three independent components of \eqref{pf:LeibN lambda A all components} one by one, starting from the $ru$-component:
    \begin{equation}
        \label{pf:LeibN_LambdA_alpha_ru}
        \begin{aligned}
            \mathrm{Leib}_\mathcal{N}(\Lambda,A)_{ru}&=2\mathcal{P}_{N_{E_{ru}}}\big(\partial_{[r}\pi_\mathcal{N}B_2(\Lambda,A)_{u]}+r^{-2}\partial_{[z}\pi_\mathcal{N}B_2(\Lambda,A)_{\Bar{z}]}+[\pi_\mathcal{N}B_1(\Lambda)_{[r},A_{u]}] \\
            &\quad +r^{-2}[\pi_\mathcal{N}B_1(\Lambda)_{[z},A_{\Bar{z}]}]+\tfrac{1}{2}[\Lambda,\pi_\mathcal{N}B_1(A)_{ru}]\big) \\
            &=2\mathcal{P}_{N_{E_{ru}}}\big(-\partial_{[r}\pi_\mathcal{N}([\Lambda,A])_{u]}-r^{-2}\partial_{[z}\pi_\mathcal{N}([\Lambda,A])_{\Bar{z}]}+[\pi_\mathcal{N}(\partial\Lambda)_{[r},A_{u]}] \\
            &\quad +r^{-2}[\pi_\mathcal{N}(\partial\Lambda)_{[z},A_{\Bar{z}]}]+[\Lambda,\mathcal{P}_{N_{E_{ru}}}(\partial_{[r}A_{u]}+r^{-2}\partial_{[z}A_{\Bar{z}]})]\big) \;.
        \end{aligned}
    \end{equation}
    Let $T_{1,ru}$ be the sum of the terms above that contain $A_r$ and $A_u$. At this point, we apply the same trick as before: we take advantage of the identities $\Lambda=\mathcal{P}_{N_\Lambda}\Lambda$ and $A_\mu=\mathcal{P}_{N_{A_\mu}}A_\mu$, together with the equations in \eqref{eq:P property}, to find 
    \begin{equation}
        \begin{aligned}
            T_{1,ru}&:=2\mathcal{P}_{N_{E_{ru}}}\left(-\partial_{[r}\pi_N([\Lambda,A])_{u]}+[\pi_\mathcal{N}(\partial\Lambda)_{[r},A_{u]}]+[\Lambda,\mathcal{P}_{N_{E_{ru}}}\partial_{[r}A_{u]}]\right) \\
            &=\mathcal{P}_{N_{E_{ru}}}\big\{\mathcal{P}_{N_{A_r}}\left([\mathcal{P}_{N_\Lambda}\partial_u\Lambda,A_r]+[\Lambda,\mathcal{P}_{N_{A_r}}\partial_{u}A_r]\right) \\
            &\quad -\mathcal{P}_{\mathcal{T}_{-1}N_{A_u}\setminus\{-1\}}\big([\mathcal{P}_{\mathcal{T}_{-1}N_\Lambda\setminus\{-1\}}\partial_r\Lambda,A_u]+[\Lambda,\mathcal{P}_{\mathcal{T}_{-1}N_{A_u}\setminus\{-1\}}\partial_rA_u]\big) \\
            &\quad +[\mathcal{P}_{N(A_{r})\setminus\{-1\}}\partial_{r}\Lambda,A_{u}]-[\mathcal{P}_{N_{A_u}}\partial_{u}\Lambda,A_{r}]+2[\Lambda,\mathcal{P}_{N_{E_{ru}}}\partial_{[r}A_{u]}]\big\} \\
            &=\mathcal{P}_{N_{E_{ru}}}\big([\mathcal{P}_{N_\Lambda}\partial_u\Lambda,A_r]+[\Lambda,\mathcal{P}_{N_{A_r}}\partial_{u}A_r]-[\mathcal{P}_{\mathcal{T}_{-1}N_\Lambda\setminus\{-1\}}\partial_r\Lambda,A_u] \\
            &\quad -[\Lambda,\mathcal{P}_{\mathcal{T}_{-1}N_{A_u}\setminus\{-1\}}\partial_rA_u]+[\mathcal{P}_{N(A_{r})\setminus\{-1\}}\partial_{r}\Lambda,A_{u}]-[\mathcal{P}_{N_{A_u}}\partial_{u}\Lambda,A_{r}] \\
            &\quad +2[\Lambda,\mathcal{P}_{N_{E_{ru}}}\partial_{[r}A_{u]}]\big) \ ,
        \end{aligned}
    \end{equation}
    where in the last equality we make use of \eqref{eq:condition on N 1}-\eqref{eq:condition on N 3} to simplify the projectors\footnote{In particular, equation \eqref{eq:condition on N 3} implies that $N_{E_{ru}}\subseteq\mathcal{T}_{-1}N_{A_u}\setminus\{-1\}$, since from \eqref{def:pi_N_contraints_on_N 1} together with \eqref{def:pi_N_contraints_on_N 3} it follows that $N_{E_{ru}}\not\ni-1$.}. Employing property \eqref{eq:P property 2}, we obtain
    \begin{equation}
        \label{pf:T1ru}
        \begin{aligned}
            T_{1,ru}&=\mathcal{P}_{N_{E_{ru}}}\big([\mathcal{P}_{N_\Lambda\setminus N_{A_u}}\partial_u\Lambda,A_r]-[\mathcal{P}_{\mathcal{T}_{-1}N_\Lambda\setminus(N_{A_r}\cup\{-1\})}\partial_r\Lambda,A_u] \\
            &\quad +[\Lambda,\mathcal{P}_{N_{A_r}\setminus N_{E_{ru}}}\partial_{u}A_r]-[\Lambda,\mathcal{P}_{\mathcal{T}_{-1}N_{A_u}\setminus (N_{E_{ru}}\cup\{-1\})}\partial_rA_u]\big) \ .
        \end{aligned}
    \end{equation}
    Let us focus on the first term on the RHS. We want to analyse it using the properties of the projector $\mathcal{P}$. To this end, define the set
    \begin{equation}
        \label{S_def}S:=N_{E_{ru}}\cap\left(N_\Lambda\setminus N_{A_u}+N_{A_r}\right)\subseteq\mathbb{Z}^{\leq0} \;.
    \end{equation}
    By virtue of the second line in equation \eqref{eq:P property 7}, we know that the first term in the RHS of \eqref{pf:T1ru} vanishes for all $\Lambda,A_\mu$ when the set $S$ above is empty. To determine when $S=\emptyset$, consider the following key observations.
    \begin{enumerate}[label=\roman*)]
        \item From equations \eqref{eq:condition on N 3} and \eqref{eq:condition on N 7} we deduce that
        \begin{equation}
            \label{pf:key obs 1}
            N_{E_{ru}}\subseteq\mathcal{T}_{-1}N_{A_u}\subseteq\mathcal{T}_{-1}(N_\Lambda\setminus N_{A_u}+N_\Lambda)^\complement \;.
        \end{equation}
        \item Equation \eqref{eq:condition on N 1} implies
        \begin{equation}
            \label{pf:key obs 2}
            N_{A_r}=\begin{cases}
                N_{A_r}\setminus\{-1\}\cup\{-1\}\subseteq\mathcal{T}_{-1}N_\Lambda\cup\{-1\} &\mathrm{if} \ N_{A_r}\ni-1 \\
                N_{A_r}\setminus\{-1\}\subseteq\mathcal{T}_{-1}N_\Lambda &\mathrm{if} \ N_{A_r}\not\ni-1 \;.
            \end{cases}
        \end{equation}
        \item Let $U,V\subseteq\mathbb{Z}$ and $s\in\mathbb{Z}$. It is always true that
        \begin{equation}
            \label{eq:translation distribution on sums}
            \mathcal{T}_s(U+V)=\mathcal{T}_sU+V=U+\mathcal{T}_sV \;.
        \end{equation}
        We emphasize that this identity means that the addition of sets of integers defined in \eqref{def:complement and sum of sets} and the usual notion of union of sets behave differently with the translation \eqref{def:translation}\footnote{Recall that $\mathcal{T}_s(U\cup V)=\mathcal{T}_sU\cup\mathcal{T}_sV$.}. Moreover, the identity element for the addition of sets is the singleton containing the identity element of the addition in $\mathbb{Z}$, that is to say $U+\{0\}=\{0\}+U=U$ for all $U$. Together, these two facts suggest a useful relation:
        \begin{equation}
            \label{pf:key obs 3}
            \begin{aligned}
                U+(\mathcal{T}_{-1}V\cup\{-1\})=\mathcal{T}_{-1}(U+(V\cup\{0\}))&=\mathcal{T}_{-1}[(U+V)\cup U] \;.
            \end{aligned}
        \end{equation}
    \end{enumerate}
    We combine equation \eqref{S_def} with \eqref{pf:key obs 1}-\eqref{pf:key obs 3} and we distinguish two different cases:
    \begin{enumerate}
        \item if $N_{A_r}\ni-1$, then
        \begin{equation}
            \label{pf:finding condition on 0 in N_Lambda}
            \begin{aligned}
                S&\subseteq \mathcal{T}_{-1}\left(N_\Lambda\setminus N_{A_u}+N_\Lambda\right)^\complement\cap\mathcal{T}_{-1}\left[(N_\Lambda\setminus N_{A_u}+N_\Lambda)\cup N_\Lambda\setminus N_{A_u}\right] \\
                &=\mathcal{T}_{-1}\big\{[(N_\Lambda\setminus N_{A_u}+N_\Lambda)^\complement\cap(N_\Lambda\setminus N_{A_u}+N_\Lambda)] \\
                &\quad\quad\quad\cup [(N_\Lambda\setminus N_{A_u}+N_\Lambda)^\complement\cap N_\Lambda\setminus N_{A_u}]\big\}=\emptyset \quad \Leftrightarrow \quad N_\Lambda\ni0 \;,
            \end{aligned}
        \end{equation}
        where we used the fact that for $U,V\subseteq\mathbb{Z}^{\leq0}$ we have
        \begin{equation}
            \label{pf:U+Vcompl cup U}
            (U+V)^\complement\cap U=\emptyset \quad\Leftrightarrow\quad U\subseteq U+V \quad\Leftrightarrow\quad V\ni0 \;.
        \end{equation}
        \item If $N_{A_r}\not\ni-1$, then
        \begin{equation}
            \begin{aligned}
                S&\subseteq \mathcal{T}_{-1}\left(N_\Lambda\setminus N_{A_u}+N_\Lambda\right)^\complement\cap\mathcal{T}_{-1}(N_\Lambda\setminus N_{A_u}+N_\Lambda)=\emptyset \;.
            \end{aligned}
        \end{equation}
    \end{enumerate}
    Hence, we find that the first term in \eqref{pf:T1ru} vanishes $\forall\Lambda,A_\mu$ if we impose condition \eqref{eq:condition on N 11}, which comes from the first case above. The second term in \eqref{pf:T1ru} is zero as well, as one can see by employing a similar argument, which results in the following chain of inclusions:
    \begin{equation}
        \begin{aligned}
            N_{E_{ru}}&\cap\left[\mathcal{T}_{-1}N_\Lambda\setminus(N_{A_r}\cup\{-1\})+N_{A_u}\right] \\
            &\subseteq N_{A_r}\setminus\{-1\}\cap\left[\mathcal{T}_{-1}N_\Lambda\setminus(N_{A_r}\cup\{-1\})+N_\Lambda\right] \\
            &\subseteq \left[\mathcal{T}_{-1}N_\Lambda\setminus (N_{A_r}\cup\{-1\})+N_\Lambda\right]^\complement\cap\left[\mathcal{T}_{-1}N_\Lambda\setminus(N_{A_r}\cup\{-1\})+N_\Lambda\right]=\emptyset \;.
        \end{aligned}
    \end{equation}
    Conversely, the two terms in the second line of \eqref{pf:T1ru} do not vanish for arbitrary $\Lambda,A_\mu$ unless
    \begin{equation}
        \label{pf:NEru1}
        \begin{cases}
            N_{E_{ru}}\cap(N_{A_r}\setminus N_{E_{ru}}+N_\Lambda)=\emptyset \\
            N_{E_{ru}}\cap\big[\mathcal{T}_{-1}N_{A_u}\setminus(N_{E_{ru}}\cup\{-1\})+N_\Lambda\big]=\emptyset \;.
        \end{cases}
    \end{equation}
    At this point, we focus on the remaining terms in \eqref{pf:LeibN_LambdA_alpha_ru}, denoting their sum as $T_{2,ru}$. After some algebra (completely analogous to what we do above for $T_{1,ru}$, with the only difference that now there is an additional coefficient $r^{-2}$ that needs to be treated carefully, via the identity \eqref{eq:P property 3}), we rearrange the terms in $T_{2,ru}$ as
    \begin{equation}
        \begin{aligned}
            T_{2,ru}&:=2\mathcal{P}_{N_{E_{ru}}}\big(-r^{-2}\partial_{[z}\pi_\mathcal{N}([\Lambda,A])_{\Bar{z}]}+r^{-2}[\pi_\mathcal{N}(\partial\Lambda)_{[z},A_{\Bar{z}]}]+[\Lambda,\mathcal{P}_{N_{E_{ru}}}(r^{-2}\partial_{[z}A_{\Bar{z}]})]\big) \\
            &=-r^{-2}\mathcal{P}_{\mathcal{T}_2N_{E_{ru}}}\big([\mathcal{P}_{N_\Lambda\setminus N_{A_z}}\partial_z\Lambda,A_{\bar z}]+[\Lambda,\mathcal{P}_{N_{A_{\bar z}}\setminus\mathcal{T}_2N_{E_{ru}}}\partial_z A_{\bar z}]\big)-(z\leftrightarrow\bar z) \;.
        \end{aligned}
    \end{equation}
    One can prove that the first term is zero by means of the following sequence of inclusions
    \begin{equation}
        \begin{aligned}
            \mathcal{T}_2N_{E_{ru}}\cap\left(N_\Lambda\setminus N_{A_z}+N_{A_{\bar z}}\right)&\subseteq N_{A_z}\cap(N_\Lambda\setminus N_{A_z}+N_\Lambda) \\
            &\subseteq (N_\Lambda\setminus N_{A_z}+N_\Lambda)^\complement\cap(N_\Lambda\setminus N_{A_z}+N_\Lambda)=\emptyset \;,
        \end{aligned}
    \end{equation}
    whereas the second term vanishes iff
    \begin{equation}
        \label{pf:NEru2}
        \mathcal{T}_2N_{E_{ru}}\cap(N_{A_{\bar z}}\setminus\mathcal{T}_2N_{E_{ru}}+N_\Lambda)=\emptyset \;.
    \end{equation}
    Analogous results hold for the last two terms in $T_{2,ru}$ (the ones with $z\leftrightarrow\bar{z}$). By putting \eqref{pf:NEru1} and \eqref{pf:NEru2} together, and using \eqref{eq:translation distribution on sums}, one obtains equation \eqref{eq:condition on N 8}.

    The proof of the last two equations of this lemma is conceptually analogous to this last one. The $r\bar z$-component of the Leibniz rule \eqref{pf:LeibN lambda A all components} is
    \begin{equation}
        \begin{aligned}
            \mathrm{Leib}_\mathcal{N}(\Lambda,A)_{r\bar{z}}&=4\mathcal{P}_{N_{E_{r\bar z}}}\big(-\partial_{[r}\pi_\mathcal{N}([\Lambda,A])_{\bar z]}+[\pi_\mathcal{N}(\partial\Lambda)_{[r},A_{\bar z]}]+[\Lambda,\mathcal{P}_{N_{E_{r\bar z}}}\partial_{[r}A_{\bar z]}]\big) \;.
        \end{aligned}
    \end{equation}
    This expression can be expanded using \eqref{def:better_truncation}, obtaining a sum of many terms. We find it useful to group them as follows:
    \begin{subequations}
        \begin{align}
            \begin{split}\label{pf:T1rbz}
                T_{1,r\bar z}&:=2\mathcal{P}_{N_{E_{r\bar z}}\setminus\{-1\}}\big([\mathcal{P}_{N_\Lambda\setminus N_{A_{\bar z}}}\partial_{\bar z}\Lambda,A_r]-[\mathcal{P}_{\mathcal{T}_{-1}N_\Lambda\setminus(N_{A_r}\cup\{-1\})}\partial_r\Lambda,A_{\bar z}]\big) \\
                &\quad +2\delta_{E_{r\bar z},-1}\mathcal{P}_{-1}\big([\Lambda,\mathcal{P}_{N_{A_r}\setminus N_{E_{r\bar z}}}\partial_{\bar z} A_r]+[\mathcal{P}_{N_{A_r}\setminus\{-1\}}\partial_r\Lambda,A_{\bar z}] \\
                &\quad +[\Lambda,\mathcal{P}_{N_{E_{r\bar z}}\setminus\{-1\}}\partial_rA_{\bar z}]\big)
            \end{split} \\
            \begin{split}\label{pf:T2rbz}
                T_{2,r\bar z}&:=2\mathcal{P}_{N_{E_{r\bar z}}\setminus\{-1\}}\big([\Lambda,\mathcal{P}_{N_{A_r}\setminus N_{E_{r\bar z}}}\partial_{\bar z} A_r]-[\Lambda,\mathcal{P}_{\mathcal{T}_{-1}N_{A_{\bar z}}\setminus(N_{E_{r\bar z}}\cup\{-1\})}\partial_rA_{\bar z}]\big) \\
                &\quad+2\delta_{E_{r\bar z},-1}\mathcal{P}_{-1}[\mathcal{P}_{N_\Lambda\setminus N_{A_{\bar z}}}\partial_{\bar z}\Lambda,A_r] \;,
            \end{split}
        \end{align}
    \end{subequations}
    where $\delta_{E_{r\bar z},-1}\coloneqq1$ if $N_{E_{r\bar z}}\ni-1$ and it vanishes otherwise. Once again, we use a sequence of inclusions and equation \eqref{eq:P property 7} to prove that $T_{1,r\bar z}=0$ for all $\Lambda,A_\mu$. Let us start by focusing on the first term of \eqref{pf:T1rbz}. We distinguish two cases:
    \begin{enumerate}
        \item if $N_{A_r}\ni-1$, then\footnote{Use equations \eqref{eq:condition on N 4.1}, \eqref{eq:condition on N 1}, \eqref{pf:key obs 3}, \eqref{eq:condition on N 7}, \eqref{pf:U+Vcompl cup U} and \eqref{eq:condition on N 11}, in this order.}
        \begin{equation}
            \begin{aligned}
                N_{E_{r\bar z}}&\setminus\{-1\}\cap(N_\Lambda\setminus N_{A_{\bar z}}+N_{A_r}) \\
                &\subseteq \mathcal{T}_{-1}N_{A_{\bar z}}\cap\left[N_\Lambda\setminus N_{A_{\bar z}}+(\mathcal{T}_{-1}N_\Lambda\cup \{-1\})\right] \\
                &= \mathcal{T}_{-1}N_{A_{\bar z}}\cap\mathcal{T}_{-1}[(N_\Lambda\setminus N_{A_{\bar z}}+N_\Lambda)\cup N_\Lambda\setminus N_{A_{\bar z}}] \\
                &\subseteq \mathcal{T}_{-1}\big\{(N_\Lambda\setminus N_{A_{\bar z}}+N_\Lambda)^\complement\cap [(N_\Lambda\setminus N_{A_{\bar z}}+N_\Lambda)\cup N_\Lambda\setminus N_{A_{\bar z}}]\big\}=\emptyset \;.
            \end{aligned}
        \end{equation}
        The last line above is equal to the empty set since we require condition \eqref{eq:condition on N 11} to hold (the steps to prove this closely mirror those in \eqref{pf:finding condition on 0 in N_Lambda}).
        \item If $N_{A_r}\not\ni-1$, then
        \begin{equation}
            \begin{aligned}
                N_{E_{r\bar z}}&\setminus\{-1\}\cap(N_\Lambda\setminus N_{A_{\bar z}}+N_{A_r}) \\
                &\subseteq \mathcal{T}_{-1}(N_\Lambda\setminus N_{A_{\bar z}}+N_\Lambda)^\complement\cap \mathcal{T}_{-1}(N_\Lambda\setminus N_{A_{\bar z}}+N_\Lambda)=\emptyset \;.
            \end{aligned}
        \end{equation}
    \end{enumerate}
    In both cases, the first term of $T_{1,r\bar z}$ vanishes. The second term in \eqref{pf:T1rbz} is zero as well, thanks to the following chain of inclusions:
    \begin{equation}
        \begin{aligned}
            N_{E_{r\bar z}}&\setminus\{-1\}\cap\left[\mathcal{T}_{-1}N_\Lambda\setminus(N_{A_r}\cup\{-1\})+N_{A_{\bar z}}\right] \\
            &\subseteq N_{A_r}\setminus\{-1\}\cap\left[\mathcal{T}_{-1}N_\Lambda\setminus(N_{A_r}\cup\{-1\})+N_\Lambda\right] \\
            &\subseteq\left[\mathcal{T}_{-1}N_\Lambda\setminus(N_{A_r}\cup\{-1\})+N_\Lambda\right]^\complement\cap \left[\mathcal{T}_{-1}N_\Lambda\setminus(N_{A_r}\cup\{-1\})+N_\Lambda\right]=\emptyset \;.
        \end{aligned}
    \end{equation}
    We observe that the other terms in $T_{1,r\bar z}$ vanish when $N_{E_{r\bar z}}\not\ni-1$, because of the overall factor $\delta_{E_{r\bar z},-1}$. Hence, suppose $N_{E_{r\bar z}}\ni-1$. Then, from equations \eqref{eq:condition on N 4} and \eqref{eq:condition on N 11} we deduce that $N_{A_r}\ni-1$ and $N_\Lambda\ni0$. Moreover, $N_{E_{r\bar z}}$ cannot contain zero\footnote{$N_{E_{r\bar z}}$ is contained in $N_{A_r}$ by hypothesis, see \eqref{eq:condition on N 4}, so \eqref{eq:0_not_in_NAr} implies that $N_{E_{r\bar z}}\not\ni0$.}. Thus,
    \begin{equation}
        \label{pf:max of sets for rbarz component}
        \begin{aligned}
            \max N_\Lambda=0 \;, \quad \max N_{A_r}=\max N_{E_{r\bar z}}=-1\;,
        \end{aligned}
    \end{equation}
    which implies that
    \begin{equation}
        \begin{aligned}
            \max\left(N_{A_r}\setminus N_{E_{r\bar z}}+N_\Lambda\right)<-1 \ &\Rightarrow\ -1\not\in N_{A_r}\setminus N_{E_{r\bar z}}+N_\Lambda \\
            \max\left(N_{A_r}\setminus\{-1\}+N_{A_{\bar z}}\right)<-1 \ &\Rightarrow\ -1\not\in N_{A_r}\setminus\{-1\}+N_{A_{\bar z}} \\
            \max\left(N_{E_{r\bar z}}\setminus\{-1\}+N_\Lambda\right)<-1 \ &\Rightarrow \ -1\not\in N_{E_{r\bar z}}\setminus\{-1\}+N_\Lambda \;.
        \end{aligned}
    \end{equation}
    This means that each term in the second and third line of \eqref{pf:T1rbz} vanish\footnote{By virtue of \eqref{eq:P property 7}.}. Therefore, we conclude that $T_{1,r\bar z}=0$ for all $\Lambda,A_\mu$, and it does not provide any additional conditions on the power selection sets. Conversely, $T_{2,r\bar z}=0$ for all $\Lambda,A_\mu$ iff
    \begin{equation}
        \label{pf:conditionsE_rbarz}
        \begin{cases}
            N_{E_{r\bar z}}\setminus\{-1\}\cap (N_{A_r}\setminus N_{E_{r\bar z}}+N_\Lambda)=\emptyset \\
            N_{E_{r\bar z}}\setminus\{-1\}\cap [\mathcal{T}_{-1}N_{A_{\bar z}}\setminus (N_{E_{r\bar z}}\cup\{-1\})+N_\Lambda]=\emptyset \\
            \max(N_\Lambda\setminus N_{A_{\bar z}}+N_{A_r})<-1 \;.
        \end{cases}
    \end{equation}
    The first two conditions are equivalent to \eqref{eq:condition on N 9}. Given \eqref{pf:max of sets for rbarz component}, we deduce that the third equation in the system above is true only when $N_{A_{\bar z}}\ni0$, so that the maximum of the set $N_\Lambda\setminus N_{A_{\bar z}}$ is strictly less than zero. In other words, the third equation in \eqref{pf:conditionsE_rbarz} corresponds to condition \eqref{eq:condition on N 12}. This concludes the study of the $r\bar z$-component of the Leibniz rule.

    Finally, the $uz$-component of \eqref{pf:LeibN lambda A all components} is
    \begin{equation}
        \label{pf:Leib lambda A uz component}
        \begin{aligned}
            \mathrm{Leib}_\mathcal{N}(\Lambda,A)_{uz}&=4\mathcal{P}_{N_{E_{uz}}}\big(-\partial_{[u}\pi_\mathcal{N}([\Lambda,A])_{z]}+[\pi_\mathcal{N}(\partial\Lambda)_{[u},A_{z]}]+[\Lambda,\mathcal{P}_{N_{E_{uz}}}\partial_{[u}A_{z]}]\big) \\
            &=-2\mathcal{P}_{N_{E_{uz}}}\big([\mathcal{P}_{N_\Lambda\setminus N_{A_u}}\partial_u\Lambda,A_z]+[\Lambda,\mathcal{P}_{N_{A_z}\setminus N_{E_{uz}}}\partial_u A_z]\big)-(u\leftrightarrow z) \;.
        \end{aligned}
    \end{equation}
    By applying one last time the procedure described multiple times in this proof, one can show that the first term above is always zero thanks to the inclusion
    \begin{equation}
        \begin{aligned}
            N_{E_{uz}}\cap\left(N_\Lambda\setminus N_{A_u}+N_{A_z}\right)\subseteq N_{A_u}\cap\left(N_\Lambda\setminus N_{A_u}+N_\Lambda\right)=\emptyset \ ,
        \end{aligned}
    \end{equation}
    while the necessary and sufficient condition for the second term in \eqref{pf:Leib lambda A uz component} to vanish is $N_{E_{uz}}\subseteq (N_{A_u}\setminus N_{E_{uz}}+N_\Lambda)^\complement$. Similar considerations hold for the other two terms in \eqref{pf:Leib lambda A uz component}, with $u\leftrightarrow z$. Thus, with the hypotheses of this lemma, $\mathrm{Leib}_\mathcal{N}(\Lambda,A)_{uz}$ vanishes for all $\Lambda,A_\mu$ iff the power selection sets satisfy equation \eqref{eq:condition on N 10}.
\end{proof}

\subsubsection{Physically relevant solutions}
Let us highlight that the two lemmas above hold for arbitrary sets $N_I\in\mathcal{N}$, with $I\in\{\Lambda,A_\mu,E_{\mu\nu}\}$, since in none of the steps of their proofs did we rely on the explicit form of these sets. In particular, $N_I$ may contain gaps, or they may only consist of negative numbers far from zero. In such cases, however, the power selection sets describe a physically non-interesting situation, since the presence of gaps in the powers of the $1/r$ expansion and/or the absence of the leading powers (the ones closer to zero) mean that the hierarchy of subleading terms is not respected. To eliminate these scenarios and extract the physical solution from the master system, we shall introduce two additional (and quite natural) assumptions on the power selection sets:
\begin{enumerate}[label=\roman*)]
    \item $N_I$ are discrete intervals (i.e. each set contains consecutive integers without gaps). For convenience, we introduce here the following notation: for any two integers $a\leq b$, we define the discrete interval
    \begin{equation}
        \llbracket a,b\rrbracket\coloneqq [a,b]\cap\mathbb{Z}=\{a,\dots,b\} \;,
    \end{equation}
    where $[a,b]\subseteq\mathbb{R}$ is an interval in the usual sense. With this notation, the assumption we are making can be expressed as follows: the power selection sets in $\mathcal{N}$ are of the form $N_I=\llbracket a_I,b_I\rrbracket$ for some integers $a_I\leq b_I$ and for all $I$.
    \item We assume that each power selection set $N_I$ contains the maximum integer allowed, where by ``allowed" we mean compatible with \eqref{def:pi_N_contraints_on_N} and with the constraints given by lemmas \ref{lemma:conditions on N from commuting diagram} and \ref{lemma:conditions on N from Leib}. A more formal way to state this assumption is the following: for each $I$, define
    \begin{equation}
        \label{def:space of generic solutions}
        \mathfrak{N}_I^\mathrm{ms}\coloneqq\{N_I\subseteq\mathbb{Z} \ \mathrm{satisfying \ \eqref{def:pi_N_contraints_on_N}, \eqref{eq:conditions on N from commuting diagram}, \eqref{eq:conditions on N from Leibniz} \ and \ \eqref{eq:conditions on N from Leibniz last ones}}\} \;.
    \end{equation}
    In other words, $\mathfrak{N}_I^\mathrm{ms}$ is the space of solutions of the first three equations of the master system for the power selection set $N_I$. Then, we shall restrict to the subspace
    \begin{equation}
        \tilde{\mathfrak{N}}_I^\mathrm{ms}\coloneqq\Big\{\tilde{N}_I\in\mathfrak{N}_I^\mathrm{ms} \ \Big| \ \max \tilde{N}_I=\max_{N_I\in\mathfrak{N}_I}\left( \max N_I\right) \Big\} \;.
    \end{equation}
\end{enumerate}
To sum up, building on the two assumptions above, we state that the physically relevant solutions of the first three equations of the master system are given by projectors $\pi_\mathcal{N}$ with power selection sets $N_I\in\mathcal{N}$ of the form
\begin{equation}
    \label{eq:ansats for power selection sets}
    N_I=\llbracket a_I,b_I\rrbracket\in\tilde{\mathfrak{N}}_I^\mathrm{ms}\;,
\end{equation}
for some non-positive integers $a_I\leq b_I$ and for all $I\in\{\Lambda,A_\mu,E_{\mu\nu}\}$.
This means that each power selection set in the collection $\mathcal{N}$ is a discrete interval whose maximum is as close to zero as possible, consistent with \eqref{def:pi_N_contraints_on_N} and the constraints given by lemmas \ref{lemma:conditions on N from commuting diagram} and \ref{lemma:conditions on N from Leib}.

\subsubsection{Master system for \texorpdfstring{$\mathcal{N}$}{N}}
One last condition in the master system needs to be examined: the Jacobiator of the two-bracket, equation \eqref{eq:master_system_EQ4}. It turns out that a morphism $\pi_\mathcal{N}$ with power selection sets of the form \eqref{eq:ansats for power selection sets}\footnote{We emphasize that \eqref{eq:ansats for power selection sets} implies, by definition, that $N_I$ satisfies \eqref{def:pi_N_contraints_on_N}, \eqref{eq:conditions on N from commuting diagram}, \eqref{eq:conditions on N from Leibniz} and \eqref{eq:conditions on N from Leibniz last ones} $\forall I$.} solves the Jacobiator equation for $\Pi$ automatically, without the need to impose further constraints on $\mathcal{N}$. This can be proven by following a procedure that is fully analogous to the proofs of lemmas \ref{lemma:conditions on N from commuting diagram} and \ref{lemma:conditions on N from Leib}, although considerably longer and more intricate. A different strategy to show this result is to check that the projections $\pi_{\mathcal{N}_k}$ in \eqref{res:master system solution} (which, as we will see shortly, are the solutions of the first three equations of the master system) solve the Jacobiator, too. One can do that explicitly for arbitrary values of $k\in\mathbb{Z}^{\leq0}$ by employing the Mathematica package developed by one of the authors, to be presented in \cite{Giorgio:2025}.

To conclude, in this subsection we proved that the most general solution for the first three equations of the master system \eqref{eq:master_system} is a projection of the form $\Pi=\pi_\mathcal{N}$ with the sets in the collection $\mathcal{N}$ satisfying all the relations contained in \eqref{def:pi_N_contraints_on_N}, \eqref{eq:conditions on N from commuting diagram}, \eqref{eq:conditions on N from Leibniz} and \eqref{eq:conditions on N from Leibniz last ones}. For brevity, we call this set of $16$ equations the \textit{master system for $\mathcal{N}$}. We reproduce it here for convenience:
\begin{subnumcases}{\label{test}}
    N_I\subseteq\mathbb{Z}^{\leq0} \quad \forall I\in\{\Lambda,A_\mu,E_{\mu\nu}\} \tag{\ref{def:pi_N_contraints_on_N 1}} \\
    N_{E_{\mu\nu}}=N_{E_{\nu\mu}} \tag{\ref{def:pi_N_contraints_on_N 2}}\\
    N_{E_{z\bar z}}=\mathcal{T}_2N_{E_{ru}} \tag{\ref{def:pi_N_contraints_on_N 3}} \\
    N_{A_\mu}\subseteq \mathcal{T}_{-\delta_{r\mu}}N_\Lambda \tag{\ref{eq:condition on N 1}, \ref{eq:condition on N 2}, \ref{eq:condition on N 11}} \\
    N_{E_{ru}}\subseteq N_{A_r}\cap\mathcal{T}_{-1}N_{A_u}\cap\mathcal{T}_{-2}\left(N_{A_z}\cap N_{A_{\bar z}}\right) \tag{\ref{eq:condition on N 3}} \\
    N_{E_{r\bar z}}\subseteq N_{A_r}\cap\mathcal{T}_{-1}N_{A_{\bar z}} \tag{\ref{eq:condition on N 4}, \ref{eq:condition on N 4.1}, \ref{eq:condition on N 12}} \\
    N_{E_{uz}}\subseteq N_{A_u}\cap N_{A_z} \tag{\ref{eq:condition on N 5}} \\
    N_{A_r}\setminus\{-1\}\cap\big[\mathcal{T}_{-1}N_\Lambda\setminus (N_{A_r}\cup\{-1\})+N_\Lambda\big]=\emptyset \tag{\ref{eq:condition on N 6}} \\
    N_{A_\alpha}\cap\big[N_\Lambda\setminus N_{A_\alpha}+N_\Lambda\big]=\emptyset \tag{\ref{eq:condition on N 7}} \\
    N_{E_{ru}}\cap\big[\big(N_{A_r}\cup \mathcal{T}_{-1}(N_{A_u}\setminus\{0\})\cup \mathcal{T}_{-2}(N_{A_z}\cup N_{A_{\bar{z}}})\big)\setminus N_{E_{ru}}+N_\Lambda\big]=\emptyset \tag{\ref{eq:condition on N 8}} \\
    N_{E_{r\bar z}}\setminus\{-1\}\cap \big[\big(N_{A_r}\cup\mathcal{T}_{-1}(N_{A_{\bar z}}\setminus\{0\})\big)\setminus N_{E_{r\bar z}}+N_\Lambda\big]=\emptyset \tag{\ref{eq:condition on N 9}} \\
    N_{E_{uz}}\cap\big[\big(N_{A_u}\cup N_{A_z}\big)\setminus N_{E_{uz}}+N_\Lambda\big]
    =\emptyset \;. \tag{\ref{eq:condition on N 10}}
\end{subnumcases}
Moreover, we showed how to impose additional natural conditions on the power selection sets to ensure that the corresponding solutions of the master system for $\mathcal{N}$ are physically relevant. In the next paragraph, we show how to construct a family $\{\mathcal{N}_k\}_{k\in\mathbb{Z}^{\leq0}}$ of such solutions.

\subsection{Solutions of the master system for \texorpdfstring{$\mathcal{N}$}{N}}
\label{subsec:Solutions of the master system for N}
In this subsection, thanks to the observations above, we are finally able to provide the proof of the main result of \cref{sec:SDYM at null infinity via slice truncation}, that is to say equations \eqref{res:master system solution} and \eqref{res:master system solution sets}.

Let us recap the situation up to this point. The master system \eqref{eq:master_system} is solved by a projector $\pi_\mathcal{N}$ with power selection sets $N_I\in\mathcal{N}$ satisfying the master system for $\mathcal{N}$ above. Moreover, we know that the physically interesting solutions of this system are sets $N_I$ as in equation \eqref{eq:ansats for power selection sets}. In particular, $N_I$ are discrete intervals, so in order to determine them it is sufficient to find their extrema. Equation \eqref{def:space of generic solutions} tells us that the maximum of each interval $N_I$ must be the greatest possible number compatible with the system above.

The following result not only allows us to find $\max N_I$ for every $I$, but also provides a prescription to construct the solution  of the master system for $\mathcal{N}$.

\begin{lemma}[Fall-off]
    \label{lemma:fall off for Ar}
    Let $\mathcal{N}=\{N_\Lambda,N_{A_\mu},N_{E_{\mu\nu}}\}$ be a collection of power selection sets obeying condition \eqref{eq:ansats for power selection sets}. Then
    \begin{equation}
        \label{res:fall-offs}
        N_{A_r}\cup N_{E_{ru}}\cup N_{E_{r\bar z}}\not\ni-1,0 \;.
    \end{equation}
\end{lemma}
\begin{proof}
    By hypothesis\footnote{See equations \eqref{def:space of generic solutions}-\eqref{eq:ansats for power selection sets}.}, the sets $N_I\in\mathcal{N}$ are discrete intervals satisfying the master system for $\mathcal{N}$\footnote{Recall that the master system for $\mathcal{N}$ consists of equations \eqref{def:pi_N_contraints_on_N}, \eqref{eq:conditions on N from commuting diagram}, \eqref{eq:conditions on N from Leibniz} and \eqref{eq:conditions on N from Leibniz last ones}.}. In particular, as discussed in the proof of \cref{lemma:conditions on N from Leib}, we know that
    \begin{equation}
        \label{eq:fall off from first two conditions}
        \begin{aligned}
            \eqref{def:pi_N_contraints_on_N 1},\eqref{def:pi_N_contraints_on_N 3} &\Rightarrow N_{E_{ru}}\not\ni-1,0 \\
            \eqref{def:pi_N_contraints_on_N 1},\eqref{eq:condition on N 1},\eqref{eq:condition on N 4} &\Rightarrow N_{A_r}\cup N_{E_{r\bar z}}\not\ni0 \;,
        \end{aligned}
    \end{equation}
    see, for example, the line of reasoning below equations \eqref{eq:0_not_in_NAr}. Thus, what is left to prove is $N_{A_r}\cup N_{E_{r\bar z}}\not\ni-1$. Notice that if $N_{A_r}$ does not contain $-1$, $N_{E_{r\bar z}}$ does not contain it either, by virtue of \eqref{eq:condition on N 4}.
    
    Demonstrating that $N_{A_r}\not\ni-1$ requires additional steps. We argue by contradiction: suppose $N_{A_r}$ does contain $-1$. We then show that the following collection of power selection sets
    \begin{equation}
        \mathcal{N}^\mathrm{test}\coloneqq\{N_\Lambda,N_{A_\mu},N_{E_{\mu\nu}}\ | \ N_I \ \mathrm{are \ physical \ solutions\ of \ \eqref{def:pi_N_contraints_on_N},\eqref{eq:conditions on N from commuting diagram},\eqref{eq:conditions on N from Leibniz last ones}}\}
    \end{equation}
    fails to satisfy the Leibniz rule condition \eqref{eq:conditions on N from Leibniz}. Consequently, if $N_{A_r}\ni-1$ there is no solution for the master system for $\mathcal{N}$.
    
    The first step of this proof by contradiction is to understand the structure of the sets $N_I\in\mathcal{N}^\mathrm{test}$. As above, the term ``physical" in the definition of $\mathcal{N}^\mathrm{test}$ means that
    \begin{equation}
        N_I=\llbracket\min N_I,\max N_I\rrbracket
    \end{equation}
    are discrete intervals with maxima as close to zero as possible. Moreover, by definition, $N_I\in \mathcal{N}^\mathrm{test}$ satisfy the no divergences condition \eqref{def:pi_N_contraints_on_N} and the cochain map condition \eqref{eq:conditions on N from commuting diagram}. In particular, they cannot contain positive integers and equation \eqref{eq:fall off from first two conditions} holds. This allows us to determine the maxima of all the power selection sets in $\mathcal{N}^\mathrm{test}$:
    \begin{equation}
        \label{pf:maxima of sets}
        \max N_I=\begin{cases}
            0 & \mathrm{if \ } I=\Lambda,A_\alpha,E_{uz} \\
            -1 & \mathrm{if \ } I=A_r,E_{r\bar z} \\
            -2 & \mathrm{if \ } I=E_{ru} \;.
        \end{cases}
    \end{equation}
    Now, observe that equations \eqref{eq:condition on N 3}-\eqref{eq:condition on N 5} and \eqref{eq:condition on N 12} can be rearranged as follows:
    \begin{equation}
        \label{pf:inclusions for N_A_mu}
        \begin{aligned}
            N_{A_r}&\supseteq N_{E_{ru}}\cup N_{E_{r\bar z}} \\
            N_{A_u}&\supseteq \mathcal{T}_1N_{E_{ru}}\cup N_{E_{uz}} \\
            N_{A_z}&\supseteq \mathcal{T}_2N_{E_{ru}}\cup N_{E_{uz}} \\
            N_{A_{\bar z}}&\supseteq \mathcal{T}_2N_{E_{ru}}\cup \mathcal{T}_1N_{E_{r\bar z}} \;,
        \end{aligned}
    \end{equation}
    while equations \eqref{eq:condition on N 1}, \eqref{eq:condition on N 2} and \eqref{eq:condition on N 11} can be rewritten as
    \begin{equation}
        \label{pf:inclusions for N_Lambda}
        N_\Lambda\supseteq\mathcal{T}_1N_{A_r}\cup N_{A_u}\cup N_{A_z}\cup N_{A_{\bar z}} \;.
    \end{equation}
    In the list of inclusions above, the ``innermost" sets are the intervals $N_{E_{\mu\nu}}$. This means that these intervals are not constrained by \eqref{eq:conditions on N from commuting diagram} and \eqref{eq:conditions on N from Leibniz last ones}, but only by the system \eqref{def:pi_N_contraints_on_N}, which fixes their maxima. Therefore, the sets $N_{E_{\mu\nu}}\in\mathcal{N}^\mathrm{test}$ are of the form
    \begin{equation}
        \label{pf:sets N_E_munu}
        N_{E_{ru}}=\llbracket m,-2\rrbracket\,, \quad N_{E_{r\bar z}}=\llbracket n,-1\rrbracket\,, \quad N_{E_{uz}}=\llbracket l,0\rrbracket
    \end{equation}
    for some integers $m\leq-2$, $n\leq-1$ and $l\leq0$. This is the starting point to find the other intervals in $\mathcal{N}^\mathrm{test}$. Indeed, using \eqref{pf:inclusions for N_A_mu} we find
    \begin{equation}
        \label{pf:trial solution N_A_mu}
        \begin{aligned}
            N_{A_r}&\supseteq\llbracket\min(m,n),-1\rrbracket \\
            N_{A_u}&\supseteq\llbracket\min(m+1,l),0\rrbracket \\
            N_{A_z}&\supseteq\llbracket\min(m+2,l),0\rrbracket \\
            N_{A_{\bar z}}&\supseteq\llbracket\min(m+2,n+1),0\rrbracket \;.
        \end{aligned}
    \end{equation}
    On the one hand, notice that the maxima of the intervals in the RHS of the inclusions above coincide with the values found in \eqref{pf:maxima of sets}. On the other hand, we do not have any further constraints on the minima. For example, the set $N_{A_r}\in\mathcal{N}^\mathrm{test}$ is of the form
    \begin{equation}
        N_{A_r}=\llbracket\min N_{A_r},-1\rrbracket\,, \quad \min N_{A_r}\leq\min(m,n) \;.
    \end{equation}
    Thus, from \eqref{pf:inclusions for N_Lambda} we obtain
    \begin{equation}
        \label{pf:trial solution N_Lambda}
        N_\Lambda\supseteq\llbracket\min_\mu\left(\min N_{A_\mu}+\delta_{\mu r}\right),0\rrbracket \;.
    \end{equation}
    Nonetheless, for simplicity, let us restrict to the case where all the inclusions in \eqref{pf:trial solution N_A_mu} and \eqref{pf:trial solution N_Lambda} reduce to equalities:
    \begin{equation}
        \begin{aligned}
            \label{pf:trial solution N_A_mu with equalities}
            N_{A_r}&=\llbracket\min(m,n),-1\rrbracket \\
            N_{A_u}&=\llbracket\min(m+1,l),0\rrbracket \\
            N_{A_z}&=\llbracket\min(m+2,l),0\rrbracket \\
            N_{A_{\bar z}}&=\llbracket\min(m+2,n+1),0\rrbracket \\
            N_\Lambda&=\llbracket\min(m+1,n+1,l),0\rrbracket \;.
        \end{aligned}
    \end{equation}
    One can prove that all the steps that follow in this proof remain valid even in the general case. Therefore, the restriction we are assuming here does not affect the generality or validity of our proof.

    To sum up, we found that the sets $N_I\in\mathcal{N}^\mathrm{test}$ are given by \eqref{pf:sets N_E_munu} and \eqref{pf:trial solution N_A_mu with equalities}. A visual representation of such intervals is provided in \autoref{fig:proof fall off}.
    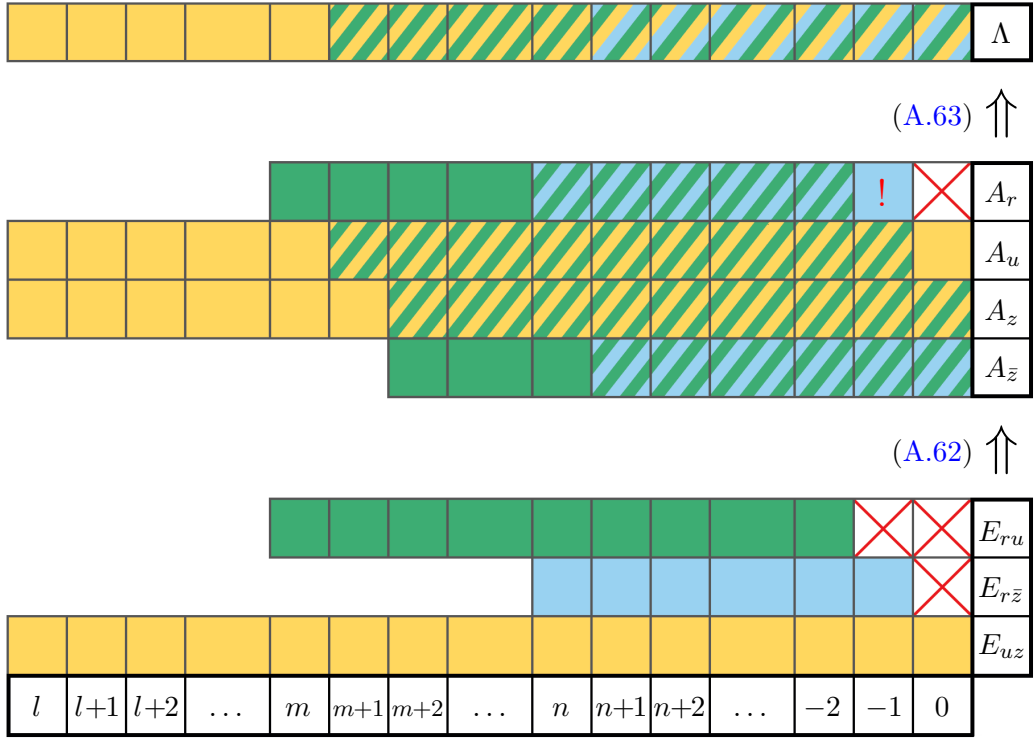
\begin{figure}[h!]
        \centering
        \begin{overpic}[width=0.9\linewidth]{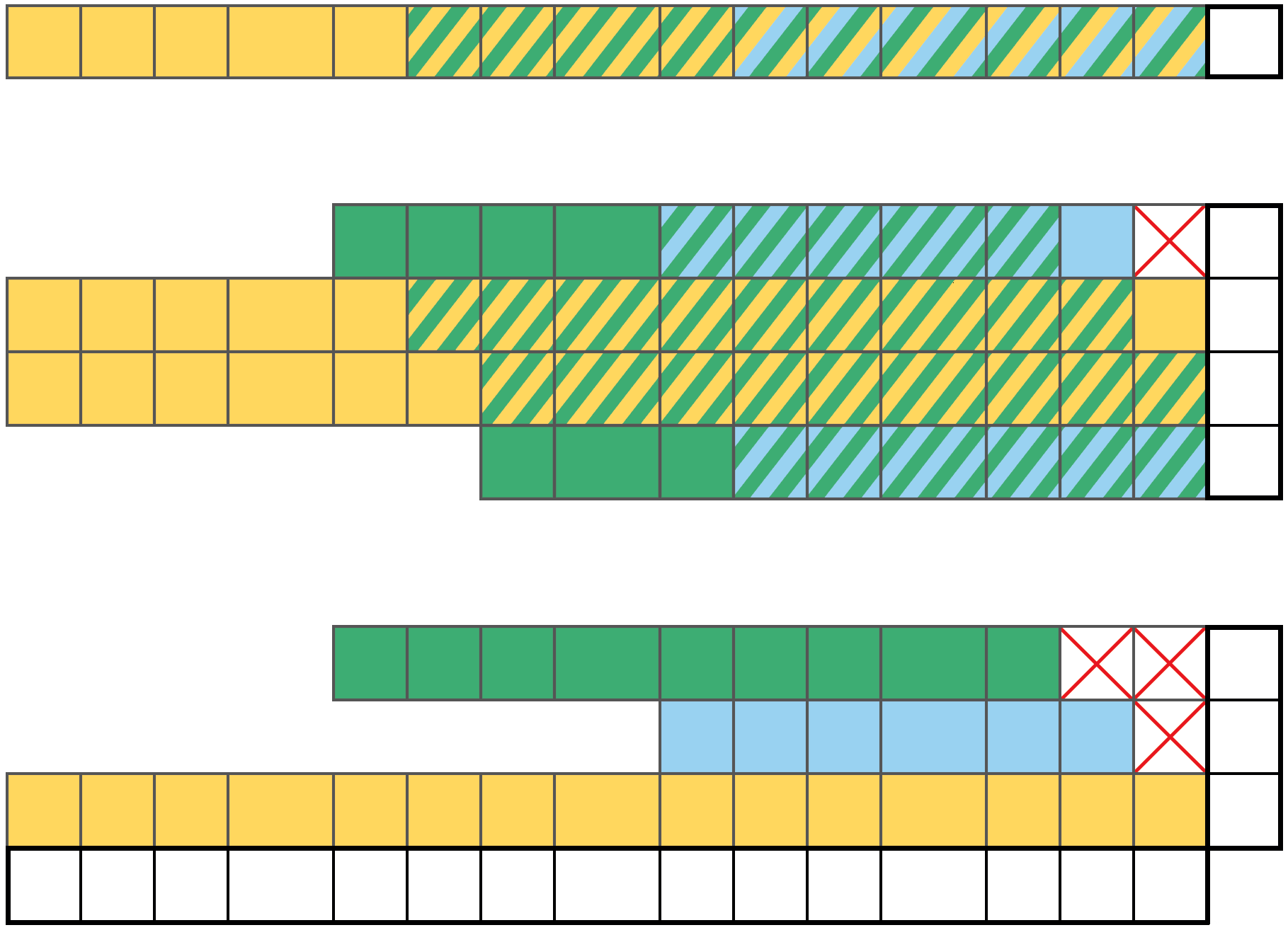} 
            \put(2.9,2.4){$l$}
            \put(7,2.4){$l$}
            \put(8,2.4){$+$}
            \put(10,2.4){$1$}
            \put(12.7,2.4){$l$}
            \put(13.7,2.4){$+$}
            \put(15.7,2.4){$2$}
            \put(20,2.4){$\dots$}
            \put(27.4,2.4){$m$}
            \put(32,2.4){\footnotesize$m$}
            \put(34,2.4){\footnotesize$+$}
            \put(35.8,2.4){\footnotesize$1$}
            \put(37.7,2.4){\footnotesize$m$}
            \put(39.7,2.4){\footnotesize$+$}
            \put(41.5,2.4){\footnotesize$2$}
            \put(45.5,2.4){$\dots$}
            \put(53.2,2.4){$n$}
            \put(57.4,2.4){$n$}
            \put(59,2.4){$+$}
            \put(61,2.4){$1$}
            \put(63,2.4){$n$}
            \put(64.6,2.4){$+$}
            \put(66.6,2.4){$2$}
            \put(70.7,2.4){$\dots$}
            \put(77.5,2.4){$-2$}
            \put(83.5,2.4){$-1$}
            \put(90.2,2.4){$0$}
            \put(94.4,19.5){$E_{ru}$}
            \put(94.4,14){$E_{r\bar z}$}
            \put(94.4,8.5){$E_{uz}$}
            \put(95,52.5){$A_r$}
            \put(95,46.7){$A_u$}
            \put(95,41){$A_z$}
            \put(95,35.5){$A_{\bar z}$}
            \put(95.6,68){$\Lambda$}
            \put(96,25.7){\huge\rotatebox{90}{$\Rightarrow$}}
            \put(96,58.5){\huge\rotatebox{90}{$\Rightarrow$}}
            \put(86,60){\eqref{pf:inclusions for N_Lambda}}
            \put(86,27.3){\eqref{pf:inclusions for N_A_mu}}
            \put(84.6,52.1){\Large\red{!}}
        \end{overpic}
        \caption{Visual representation of the discrete intervals $N_I\in\mathcal{N}^\mathrm{test}$, see \eqref{pf:sets N_E_munu} and \eqref{pf:trial solution N_A_mu with equalities}. Each square represents an integer. The red crosses indicate the numbers forbidden by \eqref{eq:fall off from first two conditions}. The square with the exclamation mark corresponds to the element $-1\in N_{A_r}$, that is the hypothesis of the proof by contradiction. The use of colours help illustrate how the intervals $N_{A_\mu}$ and $N_\Lambda$ depend on $N_{E_{\mu\nu}}$ through equations \eqref{pf:inclusions for N_A_mu} and \eqref{pf:inclusions for N_Lambda}, respectively.}
        \label{fig:proof fall off}
    \end{figure}
    One can verify that they satisfy both equations \eqref{eq:condition on N 6} and \eqref{eq:condition on N 7}. However, crucially, they violate condition \eqref{eq:condition on N 8}, as we shall prove below. Let $\mathfrak{lhs}$ be the left hand side of equation \eqref{eq:condition on N 8}. Using the explicit expressions for $N_I$ found above, we obtain
    \begin{equation}
        \label{pf:lhs of Leibniz}
        \begin{aligned}
            \mathfrak{lhs}&\coloneqq N_{E_{ru}}\cap\big(J+N_\Lambda\big)=\llbracket m,-2\rrbracket\cap\big(J+\llbracket\min(m+1,n+1,l),0\rrbracket\big) \;,
        \end{aligned}
    \end{equation}
    where we defined
    \begin{equation}
        \label{pf:J}
        \begin{aligned}
            J&\coloneqq \big(N_{A_r}\cup \mathcal{T}_{-1}(N_{A_u}\setminus\{0\})\cup \mathcal{T}_{-2}(N_{A_z}\cup N_{A_{\bar{z}}})\big)\setminus N_{E_{ru}} \\
            &=\big(\llbracket\min(m,n),-1\rrbracket\cup\llbracket\min(m,l-1),-2\rrbracket\cup\llbracket\min(m,l-2),-2\rrbracket \\
            &\quad\cup\llbracket\min(m,n-1),-2\rrbracket\big)\setminus\llbracket m,-2\rrbracket \\
            &=\llbracket\min(m,n-1,l-2),-1\rrbracket\setminus\llbracket m,-2\rrbracket \\
            &=\{-1\}\cup\begin{cases}
                \emptyset &\mathrm{if} \ \min(m,n-1,l-2)=m \\
                \llbracket\min(n-1,l-2),m-1\rrbracket &\mathrm{otherwise}\;.
            \end{cases}
        \end{aligned}
    \end{equation}
    Notice that, when $\min(m,n-1,l-2)\neq m$, one can prove that
    \begin{equation}
        \min(n-1,l-2)+\min(m+1,n+1,l)=2\min(n,l-1) \;.
    \end{equation}
    Using this identity, we rewrite \eqref{pf:lhs of Leibniz} as
    \begin{equation}
        \begin{aligned}
            \mathfrak{lhs}&=\llbracket m,-2\rrbracket\cap\left(\llbracket\min(m,n,l-1),-1\rrbracket\cup\begin{cases}
                \emptyset \\
                \llbracket 2\min(n,l-1),m-1\rrbracket
            \end{cases}\right) \\
            &=\llbracket m,-2\rrbracket\cap\llbracket\min(m,n,l-1),-1\rrbracket=\llbracket m,-2\rrbracket \;,
        \end{aligned}
    \end{equation}
    where the two cases in the first line correspond to the same conditions of the cases in \eqref{pf:J}. We have established that $\mathfrak{lhs}\neq\emptyset$, which implies that equation \eqref{eq:condition on N 8} is violated. Thus, we have reached the anticipated contradiction. Consequently, we conclude that our initial hypothesis is false: $N_{A_r}\not\ni-1$.
    
    As a consistency check, one can show that the sets in the collection
    \begin{equation}
        \mathcal{N}^\mathrm{sol}\coloneqq\{N_I\in\mathcal{N}^\mathrm{test} \ | \ N_{A_r}\not\ni-1\}
    \end{equation}
    do satisfy the Leibniz rule condition \eqref{eq:conditions on N from Leibniz}, thus providing the most general physically relevant solution of the master system for $\mathcal{N}$. By repeating the construction outlined above, we find that $N_I\in\mathcal{N}^\mathrm{sol}$ are of the form
    \begin{equation}
        \label{pf:general solution of the master system for N}
        \begin{aligned}
            N_\Lambda&\supseteq\llbracket\min_\mu\left(\min N_{A_\mu}+\delta_{\mu r}\right),0\rrbracket \\
            N_{A_r}&\supseteq\llbracket\min(m,n),-2\rrbracket \\
            N_{A_u}&\supseteq\llbracket\min(m+1,l),0\rrbracket \\
            N_{A_z}&\supseteq\llbracket\min(m+2,l),0\rrbracket \\
            N_{A_{\bar z}}&\supseteq\llbracket\min(m+2,n+1),0\rrbracket \\
            N_{E_{ru}}&=\llbracket m,-2\rrbracket \\
            N_{E_{r\bar z}}&=\llbracket n,-2\rrbracket \\
            N_{E_{uz}}&=\llbracket l,0\rrbracket
        \end{aligned}
    \end{equation}
    for some integers $m,n\leq-2$ and $l\leq0$. They differ from the sets in $\mathcal{N}^\mathrm{test}$ only in that $\max N_{A_r}=\max N_{E_{r\bar z}}=-2$ (recall that the second equality comes from \eqref{eq:condition on N 4}). Using the same reasoning as before, one can verify that the intervals in \eqref{pf:general solution of the master system for N} solve all the inclusions in \eqref{eq:conditions on N from Leibniz}. In particular, the left hand side of \eqref{eq:condition on N 8}\footnote{In the simplified case where we take all the inclusions in \eqref{pf:general solution of the master system for N} to be equalities. The proof for the general case is analogous.} reduces to
    \begin{equation}
        \begin{aligned}
            \mathfrak{lhs}=\llbracket m,-2\rrbracket\cap\begin{cases}
                \emptyset &\mathrm{if} \ \min(m,n-1,l-2)=m \\
                \llbracket2\min(n,l-1),m-1\rrbracket & \mathrm{otherwise}
                
            \end{cases}=\emptyset \;.
        \end{aligned}
    \end{equation}
    This completes the proof of \eqref{res:fall-offs}.
\end{proof}
With \cref{lemma:fall off for Ar}, we can finally conclude this appendix. Indeed, the result above not only provides a proof of the fall-offs given in equation \eqref{A_falls_no_u_gauge}, but also allows us to explicitly construct the power selection sets that solve the master system for $\mathcal{N}$. In fact, starting from \eqref{pf:general solution of the master system for N} and taking the minimal sets\footnote{That is to say, the sets with minimum cardinality. This choice is made to avoid including terms in the $1/r$-expansion of the gauge fields (parameters) that would not contribute to the corresponding equations of motion (gauge fields).} for $N_\Lambda$ and $N_{A_\mu}$ and sets of equal cardinality for each component of the equations of motion (i.e. setting $m=n=l-2$), we obtain result \eqref{res:master system solution sets}, which in turn gives the projection \eqref{res:master system solution} solving the master system \eqref{eq:master_system}.

\section{Other proofs}
\label{sec:appendix other proofs}

\subsection{\texorpdfstring{$\mathfrak{Kin}^{(k)}$}{Kin k}}
\begin{remark}[Projected chain complex]\label{rmk:Kk bk chain complex}
    Let $\Pi_k$ with $k\leq0$ be the projection in \eqref{res:master system solution}, and let $\mathcal{K}^{(k)}$ be the graded vector space resulting from the color-stripping of $\mathcal{X}^{(k)}$ in \eqref{def:space and brackets for slices}. Moreover, define $b^{(k)}\coloneqq\Pi_kb$, where $b$ is the differential in \eqref{def:differential b in components}\footnote{See \cref{ftn:projection on kin space} for a clarification regarding $\Pi_k$ acting on the color-stripped space.}. Then $(\mathcal{K}^{(k)},b^{(k)})$ is a chain complex.
\end{remark}
\begin{proof}
    This is equivalent to show that $b^{(k)}$ is a nilpotent map of degree $-1$. Since the projection $\Pi_k$ does not affect the degree of $b$ (that is $-1$) we just need to prove that $b^{(k)}$ is nilpotent. To begin with, we write the explicit action of $b$ on gauge fields $\mathcal{A}$ and equations of motions $\mathcal{E}$ in the initial space $\mathcal{K}$, using \eqref{def:differential b in components}:
    \begin{align}
        \label{eq:bA in components}
        b(\mathcal{A})&=2\left(\partial_{(r}\mathcal{A}_{u)}-r^{-2}\partial_{(z}\mathcal{A}_{\bar z)}+r^{-1}\mathcal{A}_u\right)\in K_0 \\
        \label{eq:bE in components}
        b(\mathcal{E})_\mu&=\begin{pmatrix}
        r^{-2}\partial_z\mathcal{E}_{r\bar z}-(\partial_r+2r^{-1})\mathcal{E}_{ru} \\
        r^{-2}\partial_{\bar z}\mathcal{E}_{uz}+\partial_u\mathcal{E}_{ru} \\
        \partial_r\mathcal{E}_{uz}+\partial_z\mathcal{E}_{ru} \\
        \partial_u\mathcal{E}_{r\bar z}-\partial_{\bar z}\mathcal{E}_{ru}
    \end{pmatrix}\in K_1 \;.
    \end{align}
    From the definition of $b^{(k)}$, employing \cref{lemma:properties of projector P} and the idempotence of $\Pi_k$, one can prove that
    \begin{align}
        \label{eq:bkA in components}
        b^{(k)}(\mathcal{A})&=\mathcal{P}_{\llbracket k-1,-2\rrbracket}b(\mathcal{A})+2\mathcal{P}_{-1}(r^{-1}\mathcal{A}_u)\in K_0^{(k)} \\
        \label{eq:bkE in components}
        b^{(k)}(\mathcal{E})_\mu&=\begin{pmatrix}
        \mathcal{P}_{\llbracket k-2,-3\rrbracket}b(\mathcal{E})_r \\
        \mathcal{P}_{\llbracket k-1,-2\rrbracket}b(\mathcal{E})_u \\
        \mathcal{P}_{\llbracket k,-2\rrbracket}b(\mathcal{E})_z \\
        \mathcal{P}_{\llbracket k-1,-2\rrbracket}b(\mathcal{E})_{\bar z}
    \end{pmatrix}\in K_1^{(k)}\;,
    \end{align}
    where now $\mathcal{A}\in K_{1}^{(k)}$ and $\mathcal{E}\in K_2^{(k)}$. For instance, the $r$-component of $b^{(k)}(\mathcal{E})$ results from
    \begin{equation}
        \begin{split}
            b^{(k)}(\mathcal{E})_r&=(\Pi_kb(\mathcal{E}))_r=(\Pi_kb(\Pi_k\mathcal{E}))_r \\
            &=\mathcal{P}_{N_{A_r}^k}\Big(r^{-2}\partial_z\mathcal{P}_{N_{E_{r\bar z}}^k}\mathcal{E}_{r\bar z}-(\partial_r+2r^{-1})\mathcal{P}_{N_{E_{ru}}^k}\mathcal{E}_{ru}\Big) \\
            &=\mathcal{P}_{N_{A_r}^k\cap\mathcal{T}_{-2}N_{E_{r\bar z}}^k}(r^{-2}\partial_z\mathcal{E}_{r\bar z})-\mathcal{P}_{N_{A_r}^k\cap\mathcal{T}_{-1}N_{E_{ru}}^k}[(\partial_r+2r^{-1})\mathcal{E}_{ru}] \\
            &=\mathcal{P}_{\llbracket k-2,-4\rrbracket}(r^{-2}\partial_z\mathcal{E}_{r\bar z})-\mathcal{P}_{\llbracket k-2,-3\rrbracket}[(\partial_r+2r^{-1})\mathcal{E}_{ru}] \\
            &=\mathcal{P}_{\llbracket k-2,-3\rrbracket}b(\mathcal{E})_r-r^{-2}\partial_z\cancel{\mathcal{P}_{-1}\mathcal{E}_{r\bar z}} \;,
        \end{split}
    \end{equation}
    where the second term in the last line vanishes\footnote{Indeed, $\mathcal{P}_{-1}\mathcal{E}_{r\bar z}=\mathcal{P}_{\{-1\}\cap\llbracket k-2,-2\rrbracket}\mathcal{E}_{r\bar z}=0$.}. Clearly, $b^{(k)}(b^{(k)}(u))$ vanishes for all $u\in K_0^{(k)}\cup K_1^{(k)}$, by degree. Using \eqref{eq:bkA in components}-\eqref{eq:bkE in components} we find
    \begin{equation}
        \begin{split}
            b^{(k)}(b^{(k)}(\mathcal{E}))&=\mathcal{P}_{\llbracket k-1,-2\rrbracket}b(b^{(k)}(\mathcal{E}))+2r^{-1}\cancel{\mathcal{P}_{0}b^{(k)}(\mathcal{E})_u} \\
            &=2\mathcal{P}_{\llbracket k-1,-2\rrbracket}\Big(\partial_{(r}b^{(k)}(\mathcal{E})_{u)}-r^{-2}\partial_{(z}b^{(k)}(\mathcal{E})_{\bar z)}+r^{-1}b^{(k)}(\mathcal{E})_u\Big) \\
            &=\mathcal{P}_{\llbracket k-1,-2\rrbracket}\Big(\mathcal{P}_{\llbracket k-2,-3\rrbracket}\partial_{r}b(\mathcal{E})_u+\mathcal{P}_{\llbracket k-2,-3\rrbracket}\partial_{u}b(\mathcal{E})_r-\mathcal{P}_{\llbracket k-3,-4\rrbracket}(r^{-2}\partial_{z}b(\mathcal{E})_{\bar z}) \\
            &\quad -\mathcal{P}_{\llbracket k-2,-4\rrbracket}(r^{-2}\partial_{\bar z}b(\mathcal{E})_z)+2\mathcal{P}_{\llbracket k-2,-3\rrbracket}(r^{-1}b(\mathcal{E})_u)\Big) \\
            &=2\mathcal{P}_{\llbracket k-1,-3\rrbracket}\Big(\partial_{(r}b(\mathcal{E})_{u)}-r^{-2}\partial_{(z}b(\mathcal{E})_{\bar z)}+r^{-1}b(\mathcal{E})_u\Big)+2r^{-2}\cancel{\mathcal{P}_{-1}\partial_{(z}b(\mathcal{E})_{\bar z)}}
        \end{split}
    \end{equation}
    where the cancelled terms vanish because of \eqref{eq:bkE in components}\footnote{For instance, $\mathcal{P}_{0}b^{(k)}(\mathcal{E})_u=\mathcal{P}_{\{0\}\cap\llbracket k-1,-2\rrbracket}b(\mathcal{E})_u=0$.}. Therefore,
    \begin{equation}
        b^{(k)}(b^{(k)}(\mathcal{E}))=\mathcal{P}_{\llbracket k-1,-3\rrbracket}b(b(\mathcal{E}))=0\;,
    \end{equation}
    which follows from the nilpotency of $b$.
\end{proof}

\begin{remark}[Projected box operator]\label{rmk:box k}
    Let $\Pi_k$ with $k\leq0$ be the projection in \eqref{res:master system solution}, and let $\Box^{(k)}$ be the operator defined in \eqref{def: projected box and kin bracket}. Then $\Box^{(k)}=\Pi_k\Box$.
\end{remark}
\begin{proof}
    We observe that
    \begin{equation}
        \begin{aligned}
            \Box^{(k)}-\Pi_k\Box&=[b^k,m_1^k]-\Pi_k[b,m_1]=\Pi_k(b\Pi_km_1+\cancel{m_1\Pi_kb}-bm_1-\cancel{m_1b})=\Pi_kb\Pi_k^\complement m_1 \;,
        \end{aligned}
    \end{equation}
    where we used the definition of $m_1^{(k)}$, $b^{(k)}$ and the equation $\Pi_km_1\Pi_k=\Pi_km_1$ (the color-stripped differential $m_1^{(k)}$ inherits such identity from the commuting diagram \eqref{diagram:physical_requirement} of the differential $B_1^{(k)}$). Thus, it suffices to prove that $\Pi_kb\Pi_k^\complement m_1(u)$ vanishes for all\footnote{Notice that $\Pi_kb\Pi_k^\complement m_1(u)=0$ for all $u\in K_2^{(k)}$, by degree, so there is no need to consider this case.} $u\in\mathcal{K}^{(k)}$. This can be done using the solution \eqref{res:master system solution} for $\Pi_k$ and the properties of the projection $\mathcal{P}$ listed in \cref{sec:appendix_proofs}. For instance, consider the case where $u=\Lambda\in K_0^{(k)}$:
    \begin{equation}
        \begin{split}
            \Pi_kb\Pi_k^\complement m_1(\Lambda)&=\pi_{\mathcal{N}^k}\nabla^\mu\pi_{(\mathcal{N}^k)^\complement}\partial_\mu\Lambda=-\mathcal{P}_{N_\Lambda^k}\big\{\big(\partial_u\mathcal{P}_{(N_{A_r}^k)^\complement}\partial_r\Lambda+r\leftrightarrow u\big) \\
            &\quad -r^{-2}\big(\partial_{\bar z}\mathcal{P}_{(N_{A_z}^k)^\complement}\partial_z\Lambda+z\leftrightarrow\bar z\big)+2r^{-1}\mathcal{P}_{(N_{A_u}^k)^\complement}\partial_u\Lambda\big\} \\
            &=-2\partial_r\partial_u\mathcal{P}_1\Lambda+2r^{-2}\partial_z\partial_{\bar z}\mathcal{P}_{\{1,2\}}\Lambda+2r^{-1}\partial_u\mathcal{P}_{1}\Lambda=0 \;,
        \end{split}
    \end{equation}
    where we used the fact that $\mathcal{P}_n\Lambda=0$ for all positive $n$.
\end{proof}

\subsection{\texorpdfstring{$\mathfrak{Kin}^{(k,f)}$}{Kin k f}}
\label{subsec:kin k f}
In \cref{sec:SDYM at null infinity via slice truncation}, our goal was to define an $L_\infty$ algebra to encode SDYM near $\mathcal{I}$ by suitably truncating the $r$-expansion of the elements \eqref{def:general_expansion_for_psi} in the initial graded vector space $\mathcal{X}$. We showed that the only way to achieve this is by defining subspaces $\mathcal{X}^{(k)}\subseteq\mathcal{X}$ through the projection $\Pi_k$ and the power selection sets in $\mathcal{N}_k$, on which the required $L_\infty$ algebras can be constructed - see equations \eqref{res:master system solution}-\eqref{def:space and brackets for slices}. In particular, the maxima of the power selection sets (which correspond to the fall-off) are fixed and common to all the slices $\mathfrak{S}_k$, as we emphasize in \cref{fig:generic slice}. Our goal here is to construct strict kinematic algebras by further projecting components of the $r$-expansion to zero.

To this end, we construct the refined slices $\mathfrak{S}_{k,f}$ which are defined by projecting components of the $r$-expansion of $\psi\in\mathcal{X}$ to zero, and introducing a projection that modifies the functional dependence of the gauge parameters $\Lambda\in X_{-1}$. Indeed, if one makes a gauge choice for the fields, the gauge choice imposes constraints on the gauge parameters. To clarify this point, let us consider an explicit example. As we show in \cref{subsec:strict kin algebra for k=-1}, one way to refine $\mathfrak{S}_{-1}$ in order to strictify its kinematic algebra is to modify the fall-off of $A_u$ from $0$ to $-1$, i.e. to replace the power selection set of $A_u$ as follows:
\begin{equation}
    \label{eq:replacement falloff -1}
    N_{A_u}^{(-1)}=\{-2,-1,0\} \ \rightarrow \ N_{A_u}^{(-1,f)}:=\{-2,-1\}\;.
\end{equation}
Let $\Pi_{-1,f}$ be the projection defined as in \eqref{res:master system solution}, but with the substitution above for the set $N_{A_u}^{(-1)}$. As expected, the triple $(\mathcal{X}^{(-1,f)},B_1^{(-1,f)},B_2^{(-1,f)})$ constructed accordingly\footnote{That is, using the definition \eqref{def:space and brackets for slices} with $\Pi_{-1}$ replaced by $\Pi_{-1,f}$.} does not form an $L_\infty$ algebra. Nevertheless, one can prove that the brackets $B_i^{(-1,f)}$ satisfy the $L_\infty$ relations up to terms proportional to $\partial_u\Lambda^{(0)}$. This means that enforcing $\Lambda^{(0)}$ to be independent of $u$ guarantees the restoration of an $L_\infty$ algebra, thus resolving the problem. Therefore, we require that the projection $\Pi_{-1,f}$ acts on $\Lambda\in X_{-1}$ as
\begin{equation}
    \label{eq:non trivial transf on lambda}
    \Lambda(r,u,z,\bar z)=\sum_{n=-\infty}^\infty r^n\Lambda^{(n)}(u,z,\bar z) \ \xrightarrow{\Pi_{-1,f}} \ \Lambda^{(0)}(z,\bar z)+\sum_{n=-2}^{-1} r^n\Lambda^{(n)}(u,z,\bar z)
\end{equation}
Physically, this procedure has a clear interpretation: the replacement \eqref{eq:replacement falloff -1} simply amounts to setting $A_u^{(0)}=0$, and the consistency of the gauge transformation then implies that $\partial_u\Lambda^{(0)}$ vanishes.

Building on this reasoning and generalizing it, we now provide the explicit definition of the projector $\Pi_{k,f}$ introduced in \cref{Strict kin algebras on refined slices}, which is used to construct the refined slices $\mathfrak{S}_{k,f}$. The first ingredient that we need is the following collection of power selection sets:
\begin{equation}
    \mathcal{N}_{k,f}\coloneqq\{N^{(k,f)}_\Lambda,N^{(k,f)}_{A_\mu},N^{(k)}_{E_{\mu\nu}}\subseteq\mathbb{Z}^{\leq0}\}\;,
\end{equation}
where the sets $N_I^{(k,f)}\subseteq N_I^{(k)}$, $I=\Lambda,A_\mu$, are defined as
\begin{align}
    N_\Lambda^{(k,f)}&\coloneqq\{k-1,\dots,f_r^k+1,0\} \\
    N_{A_\mu}^{(k,f)}&\coloneqq\begin{cases}
       \{k-2,\dots,f_r^k\} &\mu=r \\
       \{k-1,\dots,f_u^k\} &\mu=u \\
       \{k,\dots,f_z^k\} &\mu=z \\
       \{k-1,\dots,f_{\bar z}^k\} &\mu=\bar z
    \end{cases}\;, \tag{\ref{res:master system solution sets gauge field strict}}
\end{align}
and the new variable fall-off are given by
\begin{equation}
    f_\mu^k\coloneqq\max N_{A_\mu}^{(k,f)}\in\mathbb{Z} \quad \text{s.th.} \quad \begin{cases}
        k-2\leq f^k_r\leq-2 \\
        k-1\leq f^k_u,f^k_{\bar z}\leq0 \\
        k\leq f^k_z\leq0
    \end{cases}. \tag{\ref{eq:def inequalities for f}}
\end{equation}
The second ingredient is an operator $\mathfrak{D}_{k,f}$ that generalizes \eqref{eq:non trivial transf on lambda}, by appropriately modifying the functional dependence of the gauge parameter. To construct such operator, which of course depends on the modified fall-offs $f_\mu^k$, we need one more definition. Let $m\in\mathbb{Z}$ and $\alpha\in\{u,z,\bar z\}$. We introduce a map $\mathfrak{d}_{\alpha|m}\in\mathrm{End}(\mathcal{C}_r^\infty(\mathcal{M}))$ that acts on $g\in\mathcal{C}_r^\infty(\mathcal{M})$ as
\begin{equation}
    \mathfrak{d}_{\alpha|m}(g)\coloneqq\sum_{n=-\infty}^m\mathcal{P}_ng+\sum_{n=m+1}^\infty\mathcal{P}_ng\Big|_{\partial_\alpha\mathcal{P}_ng=0} \;.
\end{equation}
We list below two examples to clarify how this map works:
\begin{equation}
    \begin{split}
        \mathfrak{d}_{u|4}(g)&=\sum_{n=-\infty}^4r^ng^{(n)}(u,z,\bar z)+\sum_{n=5}^\infty r^ng^{(n)}(z,\bar z) \\
        \mathfrak{d}_{z|-1}(\mathfrak{d}_{u|4}(g))&=\sum_{n=-\infty}^{-1}r^ng^{(n)}(u,z,\bar z)+\sum_{n=0}^{4}r^ng^{(n)}(u,\bar z)+\sum_{n=5}^\infty r^ng^{(n)}(\bar z) \;.
    \end{split}
\end{equation}
The operator $\mathfrak{D}_{k,f}$ is then defined as
\begin{equation}
    \mathfrak{D}_{k,f}\coloneqq\prod_\alpha\mathfrak{d}_{\alpha|f_\alpha^k}\in\mathrm{End}(\mathcal{C}_r^\infty(\mathcal{M}))\;.
\end{equation}
At this point, we can finally state the definition of the desired projection from $\mathcal{X}$ to the refined slice $\mathfrak{S}_{k,f}$ (with $k\leq0$ and fall-off $f_\mu^k$):
\begin{equation}
\label{res:master system solution strict}
    \Pi_{k,f}(\psi)\coloneqq\begin{cases}
        \mathfrak{D}_{k,f}\circ\pi_{\mathcal{N}_{k,f}}(\psi) &\psi\in X_{-1} \\
        \pi_{\mathcal{N}_{k,f}}(\psi) &\text{otherwise}
    \end{cases}
\end{equation}
To provide a couple of examples, we now explicitly write how $\Pi_{k,f}$ projects the gauge fields when selecting the first two fall-offs from \cref{table: refined slices for k leq -2}:
\begin{equation}
    \begin{aligned}
        \mathfrak{D}_{k,f_4}(\Lambda)&=\sum_{n=k+2}^{0}r^n\Lambda^{(n)}(z)+\sum_{n=k-1}^{k+1}r^n\Lambda^{(n)}(u,z,\bar z)\in X_{-1}^{(k,f_4)} \\
        \mathfrak{D}_{k,f_5}(\Lambda)&=\Lambda^{0}(u,\bar z)+r^{k-1}\Lambda^{(k-1)}(u,z,\bar z)\in X_{-1}^{(k,f_5)} \;.
    \end{aligned}
\end{equation}

Thanks to the projection above, for any non positive integer $k$ and choice of fall-off $f_\mu^k$ satisfying \eqref{eq:def inequalities for f}, it is possible to construct a kinematic algebra $\mathfrak{Kin}^{(k,f)}$ on the refined slice $\mathfrak{S}_{k,f}$ following a procedure analogous to that presented in \cref{subsec:Kinematic algebras on slices}. Hence, we can prove the following result.
\begin{result}[Strict kinematic algebra]
    \label{res:strict kin algebra}
    Let the $\mathrm{BV}_{\infty}^{\Box^{(k,f)}}$ algebra $\mathfrak{Kin}^{(k,f)}$ be the kinematic algebra of the refined slice $\mathfrak{S}_{k,f}$, for some $k\in\mathbb{Z}^{\leq0}$ and $f_\mu^k$ satisfying \eqref{eq:def inequalities for f}. Then
    \begin{equation}
        \mathfrak{Kin}^{(k,f)} \text{\ is strict} \quad \Leftrightarrow \quad f_\mu^k \text{\ satisfy}\quad \begin{cases}
            f_r^k+f_u^k+f_z^k<k \\
            f_r^k+f_u^k+f_{\bar z}^k<k-1 \\
            f_r^k+f_z^k+f_{\bar z}^k<k \\
            f_u^k+f_z^k+f_{\bar z}^k<k+1 \\
            f_r^k+f_u^k<k \\
            f_r^k+f_{\bar z}^k<k \\
            f_u^k+f_z^k<k+2 \\
            f_z^k+f_{\bar z}^k<k+2 \;.
        \end{cases} \tag{\ref{res:strict kin algebras on refined slices}}
    \end{equation}
    Note that the system above is symmetric under $f_r^k\leftrightarrow f_z^k-2$ and $f_u^k\leftrightarrow f_{\bar z}^k$. 
\end{result}
\begin{proof}
    The kinematic algebra $\mathfrak{Kin}^{(k,f)}$ of a refined slice $\mathfrak{S}_{k,f}$ with fall-off $f_\mu^k$ is strict if the corresponding trilinear bracket $\theta_3^{(k,f)}$ vanishes (see discussion at the end of \cref{sub:kin alg gen}). Thus, proving the above result is equivalent to finding the conditions such that
    \begin{equation}
        \theta_3^{(k,f)}(u_1,u_2,u_3)=0 \quad \forall u_1,u_2,u_3\in\mathcal{K}^{(k,f)}
    \end{equation}
    To study when this happens, recall that
    \begin{equation}
        \theta_3^{(k,f)}(u_1,u_2,u_3)=\Pi_{k,f}\theta_3(u_1,u_2,u_3)=\Pi_{k,f}\theta_3(\Pi_{k,f}u_1,\Pi_{k,f}u_2,\Pi_{k,f}u_3) \;.
    \end{equation}
    We start by considering $\theta_3^{(k,f)}$ acting on three gauge fields. The $r$-component is
    \begin{equation}
        \label{eq:theta3kfAAA}
        \theta_{3}^{(k,f)}(\cA_{1},\cA_{2},\cA_{3})_r=\mathcal{P}_{N_{A_r}^{(k,f)}}\big[
        6r^{-2}\big(\mathcal{P}_{N_{A_r}^{(k,f)}}\mathcal{A}_{[1\,r}\big)\big(\mathcal{P}_{N_{A_z}^{(k,f)}}\mathcal{A}_{2\,z}\big)\big(\mathcal{P}_{N_{A_{\bar z}}^{(k,f)}}\mathcal{A}_{3]\,\bar z}\big)\big]
    \end{equation}
    To determine when this expression vanishes, we employ a generalisation of identity \eqref{eq:P property 7} of \cref{sec:appendix_proofs}. Let $i\in\mathbb{N}$ and $L,M_i\subseteq\mathbb{Z}$, one can show that\footnote{The procedure is analogous to the proof of \eqref{eq:P property 7}.}
    \begin{equation}
        \label{eq:projector on a product}
        \mathcal{P}_L\Big(\prod_i\mathcal{P}_{M_i}g_i\Big)=0 \quad\forall g_i\in C_r^\infty(\mathcal{M}) \quad \Leftrightarrow\quad L\cap\sum_iM_i=\emptyset \;.
    \end{equation}
    Thus, we conclude that \eqref{eq:theta3kfAAA} vanishes for all $\mathcal{A}_1,\mathcal{A}_2,\mathcal{A}_3\in K_1^{(k,f)}$ iff
    \begin{equation}
        \begin{aligned}
            \emptyset&=\mathcal{T}_2 N_{A_r}^{(k,f)}\cap\big( N_{A_r}^{(k,f)}+ N_{A_z}^{(k,f)}+ N_{A_{\bar z}}^{(k,f)}\big) \\
            &=\llbracket k,f_r^k+2\rrbracket\cap\big(\llbracket k-2,f_r^k\rrbracket+\llbracket k,f_z^k\rrbracket+\llbracket k-1,f_{\bar z}^k\rrbracket\big) \\
            &=\llbracket k,f_r+2\rrbracket\cap\llbracket 3k-3,f_r^k+f_z^k+f_{\bar z}^k\rrbracket \;.
        \end{aligned}
    \end{equation}
    We note that $f_r^k+2\geq 3k-3$, since by hypothesis\footnote{See eq. \eqref{eq:def inequalities for f}.} $f_r^k\geq k-2$ and $k\leq0$. Thus, the above equation is true iff
    \begin{equation}
        f_r^k+f_z^k+f_{\bar z}^k<k \;.
    \end{equation}
    Repeating the same procedure for the other components, we find that
    \begin{equation}
        \label{eq:condition th3AAA vanishing}
        \theta_3^{(k,f)}(\mathcal{A}_1,\mathcal{A}_2,\mathcal{A}_3)=0 \quad \forall\mathcal{A}_1,\mathcal{A}_2,\mathcal{A}_3\in K_1^{(k,f)} \quad \Leftrightarrow \quad \begin{cases}
            \sum_{\mu\neq r}f_\mu^k<k+1 \\
            \sum_{\mu\neq u}f_\mu^k<k \\
            \sum_{\mu\neq z}f_\mu^k<k-1 \\
            \sum_{\mu\neq \bar z}f_\mu^k<k
        \end{cases}\;,
    \end{equation}
    which prove the first four inequalities of the system in \eqref{res:strict kin algebras on refined slices}. Now, let us turn to the study of $\theta_3^{(k,f)}(\mathcal{E},\mathcal{A}_1,\mathcal{A}_2)$, where $\mathcal{E}\in K_2^{(k,f)}$. Its $ru$-component is
    \begin{equation}
        \begin{aligned}
            \theta_{3}^{(k,f)}(\mathcal{E},\cA_{1},\cA_{2})_{ru}&=\mathcal{P}_{N_{E_{ru}}^{(k)}}\big\{2r^{-2}\big[\big(\mathcal{P}_{N_{E_{uz}}^{(k)}}\mathcal{E}_{uz}\big)\big(\mathcal{P}_{N_{A_r}^{(k,f)}}\mathcal{A}_{[1\,r}\big)\big(\mathcal{P}_{N_{A_{\bar z}}^{(k,f)}}\mathcal{A}_{2]\,\bar z}\big) \\
            &\quad\quad\quad\quad\quad\quad -\big(\mathcal{P}_{N_{E_{r\bar z}}^{(k)}}\mathcal{E}_{r\bar z}\big)\big(\mathcal{P}_{N_{A_u}^{(k,f)}}\mathcal{A}_{[1\,u}\big)\big(\mathcal{P}_{N_{A_z}^{(k,f)}}\mathcal{A}_{2]\,z}\big)\big]\big\}
        \end{aligned}
    \end{equation}
    The above expression vanishes $\forall\mathcal{A}_1,\mathcal{A}_2,\in K_1^{(k,f)}$ and $\forall\mathcal{E}\in K_2^{(k,f)}$ iff
    \begin{equation}
        \begin{aligned}
            \emptyset&=\mathcal{T}_2N_{E_{ru}}^{(k)}\cap\big(N_{E_{uz}}^{(k)}+ N_{A_r}^{(k,f)}+ N_{A_{\bar z}}^{(k,f)}\big) \\
            &=\llbracket k,0\rrbracket\cap\big(\llbracket k,0\rrbracket+\llbracket k-2,f_r^k\rrbracket+\llbracket k-1,f_{\bar z}^k\rrbracket\big) \\
            &=\llbracket k,0\rrbracket\cap\llbracket 3k-3,f_r^k+f_{\bar z}^k\rrbracket
        \end{aligned}
    \end{equation}
    and
    \begin{equation}
        \begin{aligned}
            \emptyset&=\mathcal{T}_2N_{E_{ru}}^{(k)}\cap\big(N_{E_{r\bar z}}^{(k)}+ N_{A_u}^{(k,f)}+ N_{A_z}^{(k,f)}\big) \\
            &=\llbracket k,0\rrbracket\cap\big(\llbracket k-2,-2\rrbracket+\llbracket k-1,f_u^k\rrbracket+\llbracket k,f_z^k\rrbracket\big) \\
            &=\llbracket k,0\rrbracket\cap\llbracket 3k-3,f_u^k+f_z^k-2\rrbracket
        \end{aligned}
    \end{equation}
    Since $k\leq0$, the above equation hold, respectviley, iff
    \begin{equation}
        f_r^k+f_{\bar z}^k<k \quad \text{and} \quad f_u^k+f_z^k\leq k+2 \;.
    \end{equation}    
    Applying the same procedure to the other two independent components, we obtain that
    \begin{equation}
        \label{eq:condition th3EAA vanishing}
        \theta_3^{(k,f)}(\mathcal{E},\mathcal{A}_1,\mathcal{A}_2)=0 \quad \forall\mathcal{A}_1,\mathcal{A}_2\in K_1^{(k,f)},\mathcal{E}\in K_2^{(k,f)} \quad \Leftrightarrow \quad \begin{cases}
            f_r^k+f_u^k<k \\
            f_r^k+f_{\bar z}^k<k \\
            f_u^k+f_z^k<k+2 \\
            f_z^k+f_{\bar z}^k<k+2
        \end{cases},
    \end{equation}
    thus completing the proof of the system in \eqref{res:strict kin algebras on refined slices}.
\end{proof}

\section{Homotopy transfer}
\label{app:Homotopy transfer}
In this Appendix, we introduce the notion of homotopy transfer, a mathematical tool used to transfer algebraic (non-linear) structure from one cochain complex to another. In the following, we illustrate how homotopy transfer works for $L_{\infty}$ algebras. However, one can extend this concept to other algebras such as $C_{\infty}$ or BV$_{\infty}$ algebras. 

Given an $L_{\infty}$ algebra $(\mathcal{X}, B_{1}, B_{n})$ defined on a graded vector space $\mathcal{X}$ with differential $B_{1}$ and multilinear maps $B_{n}$ with $n>1$, we wish to find an $L_{\infty}$ algebra structure on a subspace $\bar{\mathcal{X}}\subseteq \mathcal{X}$ equipped with a differential $\bar B_{1}$. To that end, following our discussion on quasi-isomorphisms in \ref{sec:cochain}, we introduce two cochain maps: a projection map $\pi:\mathcal{X}\to \bar{\mathcal{X}}$ and an inclusion map $\iota: \bar{\mathcal{X}}\to \mathcal{X}$, i.e.,
\be
\pi \circ B_{1} = \bar B_{1} \circ \pi\;,\;\;\; \iota \circ \bar B_{1} = B_{1} \circ \iota\;,
\ee
which trivially obey
\be
\label{invtriv}
\pi\circ\iota = 1_{\bar{\mathcal{X}}}\;,
\ee
where $1_{\bar{\mathcal{X}}}$ is the identity map on $\bar{\mathcal{X}}$. From now on we drop the $\circ$ symbol to denote composition to simplify notation. Homotopy transfer further requires that these maps obey
\be
\iota\pi - 1_{\mathcal{X}} = [B_{1}, h]\;,
\ee
which tells us that the inclusion map $\iota$ is the inverse of the projection map $\pi$ \textit{up to homotopy}, namely that the failure of $\iota$ to be the inverse of $\pi$ is governed by the differential $B_{1}$ and a linear \textit{homotopy map} $h$ of degree $-1$ that obeys the so-called side conditions
\be
\pi h = h \iota = h^{2} = 0\;.
\ee

Even though all the above relations only involve linear maps, these relations, together with the multilinear brackets of the $B_{n}$ of the algebra defined on $\mathcal{X}$, lead to the definition of the multilinear maps $\bar B_{n}: \bar{\mathcal{X}}^{\otimes n}\to \bar{\mathcal{X}}$ that define a new $L_{\infty}$ algebra on the subspace $\bar{\mathcal{X}}$ which we present here up to trilinear order:
\begin{equation}
\begin{split}
\bar B_{2}(\bar x_{1}, \bar x_{2}) &= \pi B_{2}(\iota \bar x_{1}, \iota \bar x_{2})\;,\\
\bar B_{3}(\bar x_{1},\bar x_{2},\bar x_{3}) &= \pi B_{3}(\iota \bar x_{1}, \iota \bar x_{2}, \iota \bar x_{3}) + \pi B_{2}(hB_{2}(\iota \bar x_{1}, \iota \bar x_{2}),\iota\bar x_{3}) \\
&\quad+(-1)^{\bar x_{1}(\bar x_{2}+\bar x_{2})} \pi B_{2}(hB_{2}(\iota \bar x_{2}, \iota \bar x_{3}),\iota\bar x_{1})\\ 
&\quad+(-1)^{\bar x_{3}(\bar x_{1}+\bar x_{2})} \pi B_{2}(hB_{2}(\iota \bar x_{3}, \iota \bar x_{1}),\iota\bar x_{2})\;.
\end{split}
\end{equation}
For a general expression for all maps $\bar B_{n}$ and a more detailed explanation of homotopy transfer refer to \cite{Arvanitakis:2020rrk}. Notice that, even without a map $B_{3}$ in our original algebra defined on $\mathcal{X}$, in principle, a non-trivial three-bracket $\bar B_{3}$ may exist. This implies that if one starts with a strict algebra on $\mathcal{X}$, the transferred algebraic structure on a subspace $\bar{\mathcal{X}}$ is not generally strict. In the following, we perform homotopy transfer of self-dual Yang-Mills theory to the slices that we considered in the main body of the paper.

\subsection{Homotopy map}
\begin{result}
    The homotopy map $h\colon\mathcal{X}\rightarrow\mathcal{X}$ is given by
    \begin{equation}
        \label{res:homotopy map h}
        \begin{aligned}
            h(A)&=-\mathcal{P}_{\mathbb{N}^+}\frac{1}{\partial_r}A_r \\
            h(E)_r&=\frac{1}{2\partial_{\bar z}}\mathcal{P}_{-1}E_{r\bar z} \\
            h(E)_u&=\mathcal{P}_{\mathbb{N}^+}\frac{1}{\partial^2}\left[2\partial_uE_{ru}-\frac{1}{r^2}\left(\frac{1}{\partial_r}\partial_u\partial_zE_{r\bar z}-\partial_{\bar z}E_{uz}\right)\right] \\
            h(E)_z&=\frac{1}{\partial_u}\left(\partial_zh(E)_u-\frac{1}{2}\mathcal{P}_{\mathbb{N}^+}E_{uz}\right) \\
            h(E)_{\bar z}&=-\mathcal{P}_{\mathbb{N}^+}\frac{1}{2\partial_r}E_{r\bar z} \;,
        \end{aligned}
    \end{equation}
    where $\mathbb{N}^+\coloneqq\mathbb{N}\setminus\{0\}$ and $\partial^2$ is defined in \eqref{def:partial square}.
\end{result}

\begin{proof}
    We want to prove that the map $h$ above satisfies
    \begin{equation}
        \label{pf:h defining relation}
        [B_1,h]\psi=\left(\iota\pi-1_{\mathcal{X}}\right)\psi \;.
    \end{equation}
    for all $\psi\in\mathcal{X}$. To do so, it is useful to first study how the projector $\mathcal{P}$ commutes with the inverse of the derivatives. For any set $M\in\mathbb{Z}$, the following identities hold:
    \begin{equation}
        \mathcal{P}_{M\setminus\{0\}}\frac{1}{\partial_r}=\frac{1}{\partial_r}\mathcal{P}_{T_{-1}M\setminus\{-1\}}\;, \quad \mathcal{P}_M\frac{1}{\partial_\alpha}=\frac{1}{\partial_\alpha}\mathcal{P}_M \;.
    \end{equation}
    The second equation is trivial, since $\partial_\alpha^{-1}$ leaves the $r$-expansion invariant. To prove the first one, we notice that for any integer $m\neq0$
    \begin{equation}
        \mathcal{P}_m\frac{1}{\partial_r}f=\sum_{n\in\mathbb{Z}}\frac{\delta_{m,n+1}}{n+1}r^{n+1}f^{(n)}=\frac{1}{m}r^mf^{(m-1)}=\frac{1}{\partial_r}(r^{m-1}f^{(m-1)})=\frac{1}{\partial_r}\mathcal{P}_{m-1}f \;,
    \end{equation}
    where $\mathcal{P}_m$ is the projection defined in \eqref{pf:definition Pn}. In particular, this implies that expressions of the form $\mathcal{P}_{M\setminus\{0\}}\partial_r^{-1}f$ are well defined. For instance, $h(A)$ in \eqref{res:homotopy map h} can be rewritten as
    \begin{equation}
        \begin{aligned}
            h(A)&=-\mathcal{P}_{\mathbb{N}^+}\frac{1}{\partial_r}A_r=-\sum_{m\in\mathbb{N}^+}\frac{1}{m}\mathcal{P}_m(rA_r) \;.
        \end{aligned}
    \end{equation}
    Equation \eqref{pf:h defining relation} with $\psi=\Lambda\in X_{-1}$ reads
    \begin{equation}
        \begin{aligned}
            [B_1,h]\Lambda&=-\mathcal{P}_{\mathbb{N}^+}\frac{1}{\partial_r}\partial_r\Lambda=-\mathcal{P}_{\mathbb{N}^+}\Lambda=\left(\mathcal{P}_{\mathbb{Z}^{\leq0}}-\mathcal{P}_{\mathbb{Z}}\right)\Lambda=\left(\iota\pi-1_{\mathcal{X}}\right)\Lambda \;,
        \end{aligned}
    \end{equation}
    where we used \eqref{eq:P property 2} and the fact that $\mathcal{P}_{\mathbb{Z}^{\leq0}}\Lambda=\iota\pi\Lambda$ and $\mathcal{P}_{\mathbb{Z}}\Lambda=\Lambda$. For $\psi=A_\mu\in X_0$ the right hand side of \eqref{pf:h defining relation} has components
    \begin{equation}
        \begin{aligned}
            \left(\iota\pi-1_{\mathcal{X}}\right)A_r&=\left(\mathcal{P}_{\mathbb{Z}^{<0}\setminus\{-1\}}-\mathcal{P}_{\mathbb{Z}}\right)A_r=-\mathcal{P}_{\mathbb{N}\cup\{-1\}}A_r \\
            \left(\iota\pi-1_{\mathcal{X}}\right)A_\alpha&=\left(\mathcal{P}_{\mathbb{Z}^{\leq0}}-\mathcal{P}_{\mathbb{Z}}\right)A_\alpha=-\mathcal{P}_{\mathbb{N}^+}A_\alpha \;,
        \end{aligned}
    \end{equation}
    while the components of the left hand side are
    \begin{equation}
        \begin{aligned}
            ([B_1,h]A)_r&=-\partial_r\mathcal{P}_{\mathbb{N}^+}\frac{1}{\partial_r}A_r+\frac{2}{\partial_{\bar z}}\mathcal{P}_{-1}\partial_{[r}A_{\bar z]}=-\mathcal{P}_{\mathbb{N}\cup\{-1\}}A_r \\
            ([B_1,h]A)_u&=-\mathcal{P}_{\mathbb{N}^+}\frac{1}{\partial_r}\partial_uA_r \\
            &\quad +\mathcal{P}_{\mathbb{N}^+}\frac{4}{\partial^2}\bigg[\partial_u\partial_{[r}A_{u]}-\frac{1}{r^2}\left(\frac{1}{\partial_r}\partial_u\partial_z\partial_{[r}A_{\bar z]}-\partial_{\bar z}\partial_{[u}A_{z]}-\partial_u\partial_{[z}A_{\bar z]}\right)\bigg] \\
            &=\mathcal{P}_{\mathbb{N}^+}\left\{\frac{2}{\partial^2}\left[\left(-\partial_r\partial_u+\frac{1}{r^2}\partial_z\partial_{\bar z}\right)\left(\frac{1}{\partial_r}\partial_uA_r-A_u\right)\right]-\frac{1}{\partial_r}\partial_uA_r\right\} \\
            &=-\mathcal{P}_{\mathbb{N}^+}A_u \\
            ([B_1,h]A)_z&=\mathcal{P}_{\mathbb{N}^+}\left[\frac{1}{\partial_u}\left(\frac{1}{\partial_r}\partial_u\partial_zA_r-\partial_zA_u-2\partial_{[u}A_{z]}\right)-\frac{1}{\partial_r}\partial_zA_r\right]=-\mathcal{P}_{\mathbb{N}^+}A_z \\
            ([B_1,h]A)_{\bar z}&=-\mathcal{P}_{\mathbb{N}^+}\frac{1}{\partial_r}\partial_{\bar z}A_r-2\mathcal{P}_{\mathbb{N}^+}\frac{1}{\partial_r}\partial_{[r}A_{\bar z]}=-\mathcal{P}_{\mathbb{N}^+}A_{\bar z} \;.
        \end{aligned}
    \end{equation}
    This proves that the homotopy map $h$ in \eqref{res:homotopy map h} satisfies its defining equation \eqref{pf:h defining relation}. 
\end{proof}

\subsection{Higher brackets}
\begin{lemma}
    \begin{equation}
        \label{pf:diagram B_2}
        \begin{tikzcd}[row sep=10mm]
            \mathcal{X}\times\mathcal{X}\arrow{r}{B_2}&\mathcal{X}\\
            \bar{\mathcal{X}}\times\bar{\mathcal{X}}\arrow{u}{\iota\times\iota}\arrow{r}{\bar B_2}&\bar{\mathcal{X}}\arrow{u}{\iota}
        \end{tikzcd} \quad \textit{commutes} \quad \Leftrightarrow \quad \mathrm{Jac}(\bar B_2)=\pi\mathrm{Jac}(B_2)\iota
    \end{equation}
    Moreover, if \eqref{pf:diagram B_2} and the second condition $h \iota = 0$ are satisfied and $\mathrm{Jac}(B_2)=0$, then $\bar B_3=0$.
\end{lemma}
\begin{proof}
    The above diagram commutes iff $\iota\bar B_2=B_2\iota$, which is equivalent to $\iota\pi B_2\iota=B_2\iota$ (recalling that by definition $\bar B_2=\pi B_2\iota$). In other words, $\mathrm{Im}(B_2\iota)\subseteq\mathrm{Ker}[B_1,h]$. Thus
    \begin{equation}
        \begin{aligned}
            \mathrm{Jac}(\bar B_2)(\bar\psi_1,\bar\psi_2,\bar\psi_3)&=\bar B_2(\bar B_2(\bar\psi_1,\bar\psi_2),\bar\psi_3)+ {\dots}=\pi B_2(\iota\pi B_2(\iota\bar\psi_1,\iota\bar\psi_2),\iota\bar\psi_3)+\dots \\
            &=\pi B_2(B_2(\iota\bar\psi_1,\iota\bar\psi_2),\iota\bar\psi_3)+{\dots}=\pi\mathrm{Jac}(B_2)(\iota\bar\psi_1,\iota\bar\psi_2,\iota\bar\psi_3)
        \end{aligned}
    \end{equation}
    for all $\bar\psi_1,\bar\psi_2,\bar\psi_3\in\bar{\mathcal{X}}$. 
    From this, we see that if $B_2$ satisfies the Jacobi identity, then also $\bar B_2$ does. In this case $[\bar B_1,\bar B_3]=0$, that does not imply, in general, $\bar B_3=0$. However, from
    \begin{equation}
        \begin{aligned}
            \bar B_3(\bar\psi_1,\bar\psi_2,\bar\psi_3)&=\pi B_2(hB_2(\iota\bar\psi_1,\iota\bar\psi_2)\iota\bar\psi_3)=\pi B_2(h\iota\pi B_2(\iota\bar\psi_1,\iota\bar\psi_2)\iota\bar\psi_3) \;,
        \end{aligned}
    \end{equation}
    we deduce that $\bar B_3$ vanishes if the homotopy map satisfies the side condition $h\iota=0$.
\end{proof}

\begin{result}
    \label{res:bar B_3=0}
    $\bar B_3=0$
\end{result}
\begin{proof}
    Thanks to the above lemma and since $B_2$ satisfies the Jacobi identity and the side condition $h \iota = 0$ is true, it is sufficient to prove that $\iota\pi B_2\iota=B_2\iota$. For example, $\forall\bar\Lambda^1,\bar\Lambda^2\in\bar X_{-1}$
    \begin{equation}
        \begin{aligned}
            \iota\pi B_2(\iota\bar\Lambda^1,\iota\bar\Lambda^2)&=-\iota\mathcal{P}_{\mathbb{Z}^{\leq0}}[\iota\bar\Lambda^1,\iota\bar\Lambda^2]=-\iota\mathcal{P}_{\mathbb{Z}^{\leq0}}[\mathcal{P}_{\mathbb{Z}^{\leq0}}\iota\bar\Lambda^1,\mathcal{P}_{\mathbb{Z}^{\leq0}}\iota\bar\Lambda^2] \\
            &=-[\mathcal{P}_{\mathbb{Z}^{\leq0}}\iota\bar\Lambda^1,\mathcal{P}_{\mathbb{Z}^{\leq0}}\iota\bar\Lambda^2]=-[\iota\bar\Lambda^1,\iota\bar\Lambda^2]=B_2(\iota\bar\Lambda^1,\iota\bar\Lambda^2)
        \end{aligned}
    \end{equation}
    where in the second equality we used $\bar\Lambda^i=\mathcal{P}_{\mathbb{Z}^{\leq0}}\bar\Lambda^i$ while in the third we used \eqref{eq:P property 5} with $\mathbb{Z}^{\leq0}+\mathbb{Z}^{\leq0}=\mathbb{Z}^{\leq0}$. All the other cases can be proved in the same way, recalling \eqref{eq:P property 2} and that $I_i+I_2=I_2$ for $i=1,2$, with $I_1=\mathbb{Z}^{\leq0}$ and $I_2=\mathbb{Z}^{\leq0}\setminus\{-1\}$. In particular, $\bar B_2=B_2|_{\bar{\mathcal{X}}}$ where $\bar{\mathcal{X}}$ is viewed as a subset of $\mathcal{X}$.
\end{proof}

\subsection{Kinematic algebra from homotopy transfer}
Since the operators in \eqref{homotopy transfer diagram} are color blinded, the linear structure of the algebra $(\mathcal{K},m_1,m_2)$ obtained by colour stripping and the corresponding homotopy transfer are preserved.
\begin{equation}
    \label{homotopy transfer diagram for K}
    \begin{tikzcd}[row sep=10mm]
        0\arrow{r}{m_1}&K_0\arrow{d}{\pi}\arrow{r}{m_1}&K_1\arrow{d}{\pi}\arrow{r}{m_1}&K_2\arrow{d}{\pi}\arrow{r}{m_1}&0\\
        0\arrow{r}{\bar m_1}&\bar K_0\arrow[shift left=2mm]{u}{\iota}\arrow{r}{\bar m_1}&\bar K_1\arrow[shift left=2mm]{u}{\iota}\arrow{r}{\bar m_1}&\bar K_2\arrow[shift left=2mm]{u}{\iota}\arrow{r}{\bar m_1}&0
    \end{tikzcd}
\end{equation}
where $\bar m_1\coloneqq\pi m_1\iota$. The above observation together with $\iota\bar B_2=b_2\iota$ imply $\iota\bar m_2=m_2\iota$. Moreover, we define 
\begin{equation}
    \label{def:bar b}
    \bar b:=\pi b\iota,
\end{equation}
\begin{lemma}
    $\bar\Box=\pi\Box\iota$
\end{lemma}
\begin{proof}
    From the homotopy transfer \eqref{homotopy transfer diagram for K} we know $\bar m_1\pi=\pi m_1$ and $\iota\bar m_1=m_1\iota$, so
    \begin{equation}
        \begin{aligned}
            \bar\Box=\bar m_1\bar b+\bar b\bar m_1=\bar m_1\pi b\iota+\pi b\iota\bar m_1=\pi m_1b\iota+\pi bm_1\iota=\pi\Box\iota
        \end{aligned}
    \end{equation}
\end{proof}
\begin{lemma}
    \label{lm:commuting of b iota}
    \begin{equation}
        \label{pf:diagram b}
        \begin{tikzcd}[row sep=10mm]
            \mathcal{K}\arrow[leftarrow]{r}{b}&\mathcal{K}\\
            \bar{\mathcal{K}}\arrow{u}{\iota}\arrow[leftarrow]{r}{\bar b}&\bar{\mathcal{K}}\arrow{u}{\iota}
        \end{tikzcd} \quad \textit{commutes}
    \end{equation}
\end{lemma}
\begin{proof}
    The above diagram commutes if $\iota\bar b=b\iota$. From the definition of $\bar b$, this is equivalent to $\iota\pi b\iota=b\iota$, i.e. $\iota\pi=1_{\mathcal{K}}$ on $\mathrm{Im}(b\iota)$. We make use of \eqref{eq:P_m properties} to conclude that $\forall\bar A\in\bar K_1$
    \begin{equation}
        \begin{aligned}
            \iota\pi b\iota\bar A&=-\mathcal{P}_{\mathbb{Z}^{\leq0}}\left(\nabla ^\mu\bar A_\mu\right)=\mathcal{P}_{\mathbb{Z}^{\leq0}}\left(\nabla_{\{r}\bar A_{u\}}-r^{-2}\nabla_{\{z}\bar A_{\bar z\}}\right) \\
            &=\mathcal{P}_{\mathbb{Z}^{\leq0}}\left[\mathcal{P}_{\mathbb{Z}^<\setminus\{-1\}}\left(\nabla_{\{r}\bar A_{u\}}-r^{-2}\nabla_{\{z}\bar A_{\bar z\}}\right)+2\mathcal{P}_{-1}(r^{-1}\bar A_u)\right] \\
            &=\mathcal{P}_{\mathbb{Z}^<\setminus\{-1\}}\left(\nabla_{\{r}\bar A_{u\}}-r^{-2}\nabla_{\{z}\bar A_{\bar z\}}\right)+2\mathcal{P}_{-1}(r^{-1}\bar A_u) \\
            &=\nabla_{\{r}\bar A_{u\}}-r^{-2}\nabla_{\{z}\bar A_{\bar z\}}=b\iota\bar A
        \end{aligned}
    \end{equation}
\end{proof}

\begin{result}
    $\bar b_2=\pi b_2\iota$
\end{result}
\begin{proof}
    Result \eqref{res:bar B_3=0} implies $\iota\bar m_2=m_2\iota$, since colour stripping does not affect the analogous property for $B_2$. Using the definition of $\bar m_2$ \eqref{def:bar b}, this relation translates to $\iota\pi m_2\iota=m_2\iota$. Recalling Lemma \eqref{lm:commuting of b iota}, we conclude that $\iota\pi=1_{\mathcal{K}}$ on both $\mathrm{Im}(m_2\iota)$ and $\mathrm{Im}(b\iota)$. Thus, the bracket $\bar b_2$, that is defined from $\bar b$ and $\bar m_2$, is
    \begin{equation}
        \begin{aligned}
            \bar b_2(\bar u_1,\bar u_2)&\coloneqq\bar b\bar m_2(\bar u_1,\bar u_2)-\bar m_2(\bar b\bar u_1,\bar u_2)-(-1)^{|\bar u_1|}
    \bar m_2(\bar u_1,\bar b\bar u_2) \\
    &=\pi b\iota\pi m_2(\iota\bar u_1,\iota\bar u_2)-\pi m_2(\iota\pi b\iota \bar u_1,\iota\bar u_2)-(-1)^{|\bar u_1|}
    \pi m_2(\iota\bar u_1,\iota\pi b\iota\bar u_2) \\
    &=\pi b_2(\iota\bar u_1,\iota\bar u_2)
        \end{aligned}
    \end{equation}
\end{proof}

\begin{result}
    $\bar\theta_3=\pi\theta_3\iota$
\end{result}
\begin{proof}
    The Poisson identity is: 
    \begin{equation}
    \label{pf:bar poisson identity}
        \begin{aligned}
            \overline{\mathrm{Poiss}}(\bar u_1,\bar u_2,\bar u_3)&\coloneqq\bar b_2(\bar m_2(\bar u_1,\bar u_2),\bar u_3)-(-)^{\bar u_1(\bar u_2+\bar u_3)}\bar m_2(\bar b_2(\bar u_2,\bar u_3),\bar u_1) \\
            &\quad -(-)^{\bar u_3(\bar u_1+\bar u_2)}\bar m_2(\bar b_2(\bar u_3,\bar u_1),\bar u_2) \\
            &=[\bar m_1,\bar \theta_3](\bar u_1,\bar u_2,\bar u_3)
        \end{aligned}
    \end{equation}
    $\forall \bar u_1,\bar u_2,\bar u_3\in\bar{\mathcal{K}}$. The first term in the first line above is
    \begin{equation}
        \begin{aligned}
            \bar b_2(\bar m_2(\bar u_1,\bar u_2),\bar u_3)=\pi b_2(\iota\pi m_2(\iota\bar u_1,\iota\bar u_2),\iota\bar u_3)=\pi b_2(m_2(\iota\bar u_1,\iota\bar u_2),\iota\bar u_3)
        \end{aligned}
    \end{equation}
    where we used $\iota\pi=1_{\mathcal{K}}$ on $\mathrm{Im}(m_2\iota)$. Note that we also have $\iota\pi=1_{\mathcal{K}}$ on $\mathrm{Im}(b_2\iota)$, since
    \begin{equation}
        \begin{aligned}
            \iota\pi b_2(\iota\bar u_1,\iota\bar u_2)&=\iota\pi bm_2(\iota\bar u_1,\iota\bar u_2)-\iota\pi m_2(b\iota\bar u_1,\iota\bar u_2)-(-)^{\bar u_1}\iota\pi m_2(\iota\bar u_1,b\iota\bar u_2) \\
            &=\iota\pi b\iota \bar m_2(u_1,u_2)-\iota\pi m_2(\iota \bar bu_1,\iota\bar u_2)-(-)^{\bar u_1}\iota\pi m_2(\iota\bar u_1,\iota\bar b\bar u_2) \\
            &=b\iota \bar m_2(u_1,u_2)-m_2(\iota \bar bu_1,\iota\bar u_2)-(-)^{\bar u_1}m_2(\iota\bar u_1,\iota\bar b\bar u_2) \\
            &=bm_2(\iota\bar u_1,\iota\bar u_2)-m_2(b\iota\bar u_1,\iota\bar u_2)-(-)^{\bar u_1}m_2(\iota\bar u_1,b\iota\bar u_2) \\
            &=b_2(\iota\bar u_1,\iota\bar u_2)
        \end{aligned}
    \end{equation}
    $\forall\bar u_1,\bar u_2\in\bar{\mathcal{K}}$, where we used $\iota\bar m_2=m_2\iota$ and \eqref{pf:diagram b}. Thus, the second term in the first line of \eqref{pf:bar poisson identity} is
    \begin{equation}
        \begin{aligned}
            \bar m_2(\bar b_2(\bar u_2,\bar u_3),\bar u_1)=\pi m_2(\iota\pi b_2(\iota\bar u_2,\iota\bar u_3),\iota\bar u_1)=\pi m_2(b_2(\iota\bar u_2,\iota\bar u_3),\iota\bar u_1)
        \end{aligned}
    \end{equation}
    The third term in the Poissonator can be treated analogously; putting everything together, we have
    \begin{equation}
        \begin{aligned}
            [\bar m_1,\bar\theta_3]&=\overline{\mathrm{Poiss}}=\pi \mathrm{Poiss}\iota=\pi[m_1,\theta_3]\iota=\pi m_1\theta_3\iota+\pi\theta_3 m_1\iota=\bar m_1\pi\theta_3\iota+\pi\theta_3\iota\bar m_1
        \end{aligned}
    \end{equation}
\end{proof}


\bibliography{biblio}
\bibliographystyle{JHEP}

\end{document}